\documentclass[12pt,a4paper]{book}
\usepackage[top=3cm,bottom=2cm,left=3cm,right=2cm]{geometry}

\usepackage{amsmath,amscd,amssymb,amsthm,array,colortbl,enumitem}

\usepackage[english]{babel}
\usepackage[utf8]{inputenc}
\usepackage[T1]{fontenc}
\usepackage{dsfont}

\usepackage{tikz}
\usepackage{stmaryrd}
\usepackage{dsfont}
\usepackage{ulem}
\usepackage{setspace}
\usepackage{adjustbox}

\usepackage{color}
\usepackage{graphicx,psfrag}
\usepackage{amssymb,amsmath,amsthm,mathrsfs}
\usepackage{wrapfig}
\usepackage{makeidx}
\usepackage{appendix}
\usepackage[mathscr]{eucal}

\usepackage{mathabx}
\usepackage{slashbox}

\makeindex

\newcommand{\RR}{\mathbb{R}}
\newcommand{\NN}{\mathbb{N}}
\newcommand{\ZZ}{\mathbb{Z}}

\newcommand{\dis}{\displaystyle}
\renewcommand{\epsilon}{\varepsilon}

\newcommand{\Supp}{\text{\rm supp}}
\newcommand{\Card}{\text{\rm card}}

\theoremstyle{definition}

\newtheorem{lemma}{Lemma}
\newtheorem{theorem}{Theorem}
\newtheorem{corollary}{Corollary}
\newtheorem{proposition}{Proposition}
\newtheorem{definition}{Definition}
\newtheorem{example}{Example}
\newtheorem{observation}{Observation}

\newtheorem{remark}{Remark}

\newtheorem{notation}{Notation}

\begin{document}

\thispagestyle{empty}

\includegraphics[scale=1, height=1.7cm]{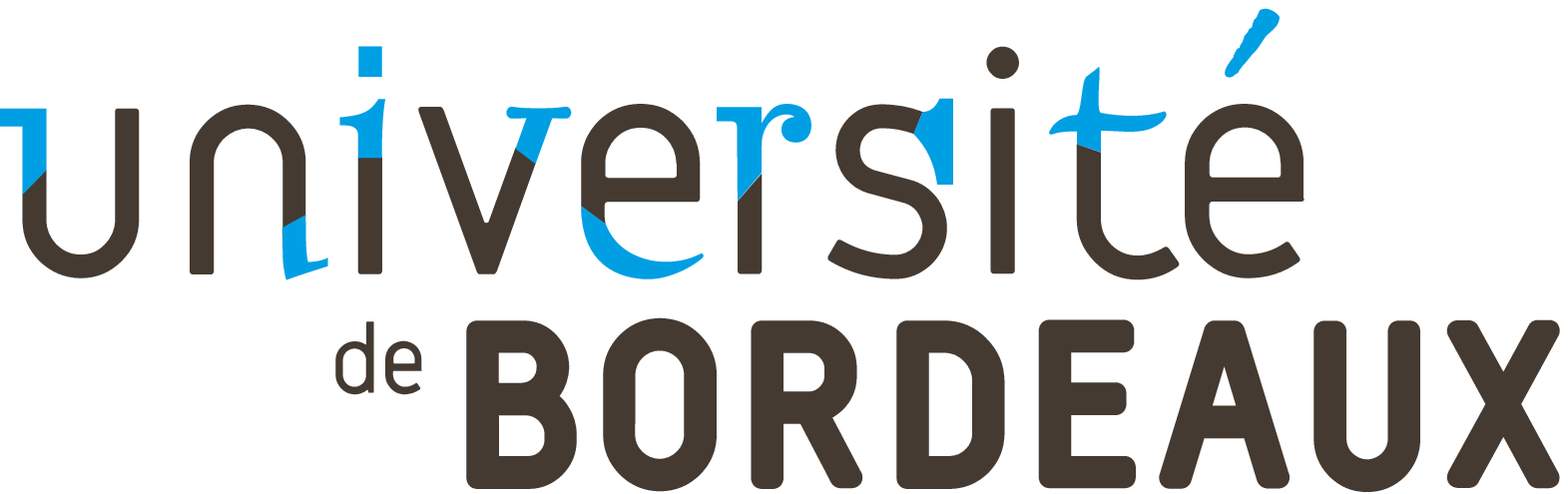}
\hfill
\includegraphics[scale=1, height=1.7cm]{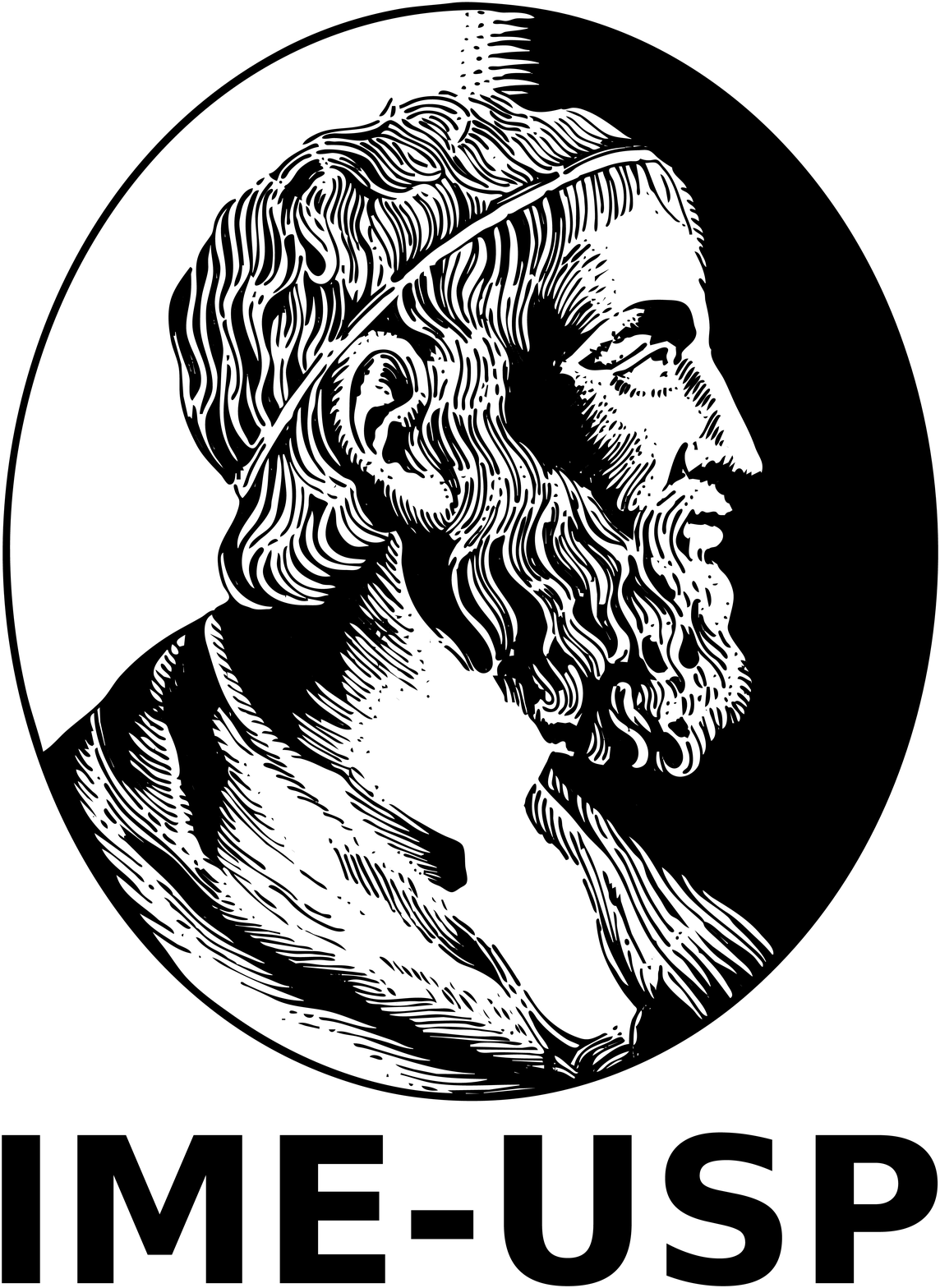}
\hfill

\begin{center}
	
\normalsize{THÈSE EN COTUTELLE PRÉSENTÉE \\
POUR OBTENIR LA GRADE DE}\\ [0.5cm]
\Large{{\bf DOCTEUR DE \\
L'UNIVERSITÉ DE BORDEAUX\\
ET DE L'UNIVERSITÉ DE SÃO PAULO}}\\ [0.5cm]
\normalsize{ÉCOLE DOCTORALE MATHÉMATIQUES ET INFORMATIQUE\\
INSTITUTO DE MATEMÁTICA E ESTATÍSTICA \\ [0.5cm]
SPÉCIALITÉ: Mathématiques Appliquées et Calcul Scientifique}

\vspace{1cm}

\Large{\bf Par Gregório DALLE VEDOVE NOSAKI}

\vspace{0.5cm}

\textbf{\LARGE{Chaos and Turing Machines on Bidimensional \\
			 Models at Zero Temperature}}
\vspace{1cm}

\normalsize{Sous la direction de Philippe THIEULLEN \\
	et de Rodrigo BISSACOT}

\end{center}

\vspace{1cm}

Soutenu le 15 décembre de 2020

\vfill

{\bf Membres du jury :}
\begin{table}[b]
\begin{footnotesize}
	\centering
	\makebox[\textwidth]{%
		\begin{tabular}{lllr}
			M. Eduardo GARIBALDI      & Professeur associé  & Universidade de Campinas & Examinateur   \\
			M. Samuel PETITE  & Maître de conférences    & Université de Picardie Jules Verne & Examinateur\\
			M. Mathieu SABLIK       & Professeur & Université de Toulouse III Paul Sabatier & Rapporteur  \\
			M. Aernout VAN ENTER       & Professeur émérite  & University of Groningen & Rapporteur  \\
			M. Pierre PICCO       & Directeur de recherche  & Institut de Math. de Marseille & Examinateur \\
			M. Philippe THIEULLEN         & Professeur & Université de Bodeaux & CoDirecteur  \\
			M. Rodrigo BISSACOT     & Professeur associé & Universidade de São Paulo & CoDirecteur     \\
			Mme. Nathalie AUBRUN & Chagé de recherche & Université Paris-Saclay & Examinatrice \\
			M. Artur LOPES & Professeur associé & Universidade Federal do Rio Grande do Sul & Invité \\
		\end{tabular}
	}
\end{footnotesize}
\end{table}
\newpage

\noindent \textbf{Titre:} Machine de Turing et Chaos pour des Modèles Bidimensionnels à Température Zéro

\noindent \textbf{Résumé:} En mécanique statistique d'équilibre ou formalisme thermodynamique un des objectifs est de décrire le comportement des familles de mesures d'équilibre pour un potentiel paramétré par la température inverse $\beta$. Nous considérons ici une mesure d'équilibre comme une mesure shift invariante qui maximise la pression. Il existe d'autres constructions qui prouvent le comportement chaotique de ces mesures lorsque le système se fige, c'est-à-dire lorsque $\beta\rightarrow + \infty $. Un des exemples les plus importants a été donné par Chazottes et Hochman~\cite {CH} où ils prouvent la non-convergence des mesures d'équilibre pour un potentiel localement constant lorsque la dimension est supérieure à 3. Dans ce travail, nous présentons une construction et un exemple potentiel localement constant tel qu'il e\-xis\-te une suite $(\beta_k)_{k\geq 0}$ où la non-convergence est assurée pour toute choix suite de mesures d'équilibre à l'inverse de la température $\beta_k$ lorsque $\beta_k \rightarrow+\infty$. Pour cela nous utilisons la construction décrite par Aubrun et Sablik~\cite{AS} qui améliore le résultat de Hochman~\cite{Hochman} utilisé dans la construction de Chazottes et Hochman~\cite{CH}.
\\

\noindent \textbf{Mots clés:} formalisme thermodynamique, measure d'équilibre, décalage.

\noindent \rule{16cm}{0.1pt}

\vspace{0.5cm}

\noindent\textbf{Title:} Chaos and Turing Machine on Bidimensional Models at Zero Temperature

\noindent\textbf{Abstract:} In equilibrium statistical mechanics or thermodynamics formalism one of the main objectives is to describe the behavior of families of equilibrium measures for a potential parametrized by the inverse temperature $\beta$. Here we consider equilibrium measures as the shift invariant measures that maximizes the pressure. Other constructions already prove the chaotic behavior of these measures when the system freezes, that is, when $\beta\rightarrow+\infty$. One of the most important examples was given by Chazottes and Hochman~\cite{CH} where they prove the non-convergence of the equilibrium measures for a locally constant potential when the dimension is bigger than or equal to 3. In this work we present a construction of a bidimensional example described by a finite alphabet and a locally constant potential in which there exists a subsequence $(\beta_k)_{k\geq 0}$ where the non-convergence occurs for any sequence of equilibrium measures at inverse temperatures $\beta_k$ when $\beta_k\rightarrow+\infty$. In order to describe such an example, we use the construction described by Aubrun and Sablik~\cite{AS} which improves the result of Hochman~\cite{Hochman} used in the construction of Chazottes and Hochman~\cite{CH}.
\\


\newpage

\noindent \textbf{Título:} Caos e Máquinas de Turing em Modelos Bidimensionais à Temperatura Zero

\noindent \textbf{Resumo:} Em mecânica estatística de equilíbrio ou formalismo termodinâmico um dos principais objetivos é descrever o comportamento das famílias de medidas de equilíbrio para um dado potencial parametrizado pelo inverso da temperatura $\beta$. Entendemos aqui por medidas de equilíbrio as medidas shift invariantes que mazimizam a pressão. Diversas construções já demonstraram um comportamento caótico destas medidas quando o sistema congela, ou seja, $\beta\rightarrow+\infty$. Um dos principais exemplos é o construído por Chazottes e Hochman \cite{CH} onde eles conseguem provar a não convergência de uma família de medidas de equilíbrio para um dado potential localmente constante nos casos onde a dimensão é maior ou igual a 3. Neste trabalho apresentaremos a construção de um exemplo no caso bidimensional sobre um alfabeto finito e um potencial localmente constante tal que existe uma sequencia $(\beta_k)_{k\geq 0}$ onde não ocorre a convergência para qualquer sequência de medidas de equilíbrio ao inverso da temperatura $\beta_k$ quando $\beta_k\rightarrow+\infty$. Para tal, usaremos a construção descrita por Aubrun e Sablik em \cite{AS} que melhora o resultado de Hochman \cite{Hochman} usado na construção de Chazottes e Hochman \cite{CH}.
\\

\noindent \textbf{Palavras-chave:} formalismo termodinâmico, medida de equilíbrio, subshift.

\vspace{7cm}

\textsc{ This study was financed in part by the \\
	Coordenação de Aperfeiçoamento de Pessoal de Nível Superior \\
	Brasil (CAPES) – Finance Code 001}

\chapter*{Résumé étendu}

L'un des problèmes les plus importants dans la mécanique statistique à l'équilibre consiste à décrire des familles de mesures de Gibbs pour un potentiel donné ou pour une famille d'interactions. Nous travaillons avec des systèmes classiques, ce qui signifie que notre espace de configuration sera
\[
\Sigma^d(\mathcal{A}):=\mathcal{A}^{\ZZ^d}
\]
où $\mathcal{A}$ est un alphabet fini et $d\in\NN$ est la dimension du réseau. Nous introduisons la fonction
\[
\varphi:\Sigma^d(\mathcal{A})\to\RR
\]
qu'il s'appelle {\it potentiel par site} et peut être physiquement interprétée comme la contribution énergétique de l'origine du réseau pour chaque configuration $x\in\Sigma^d(\mathcal{A})$.

À partir de ces éléments, nous désignons pour chaque $\beta>0$ l'ensemble $\mathcal{G}(\beta\varphi)$ qui est l'ensemble des {\it mesures de Gibbs} associées à $\beta\varphi$ à la température inverse $\beta$. Il existe plusieurs définitions que nous pouvons considérer comme une mesure de Gibbs, en utilisant des mesures conformes, des équations DLR, des limites thermodynamiques, etc. Voir Georgii~\cite{Georgii}, le livre classique sur les mesures de Gibbs et \cite{Kimura} pour les équivalences de plusieurs des ces définitions. Par compacité, nous savons que cet ensemble a au moins une mesure de Gibbs invariante pour translation. Dans cette thèse, nous nous intéressons au comportement de l'ensemble des mesures de Gibbs qui sont des mesures de probabilité invariantes, appelées mesures d'équilibre, lorsque la température tend vers zéro, c'est-à-dire lorsque $\beta \rightarrow +\infty$.

Une mesure de probabilité $\mu_\beta$ sur $\Sigma^d(\mathcal{A})$ est une {\it mesure d'équilibre} (ou {\it état d'équilibre}) à   la température inverse $\beta>0$ pour un potentiel $\beta\varphi$ si c'est une mesure invariante par décalage (ou mesure invariante par translation) qui maximise la pression, c'est-à-dire si
\[\dis P(\beta\varphi):=\sup_{\mu \in \mathcal{M}_\sigma(\Sigma^d (\mathcal{A}))} \left\{h(\mu)-\int\beta\varphi d\mu \right\} = h(\mu_\beta)-\int \beta\varphi d\mu_\beta.\]

Nous considérerons par la suite l'ensemble uniquement ces mesures d'équilibre $\mu_\beta$, celles qui maximisent la pression $P(\beta\varphi)$ ci-dessus sur toutes les mesures de probabilité invariantes pour translation définies sur $\Sigma^d(\mathcal{A})$. La fonction $h(nu)$ dans l'expression de $P(\beta\varphi)$ est l'entropie de Kolmogorov-Sinai de $\nu$.

Dans le cas unidimensionnel, si un potentiel $\varphi$ est Hölder continu, nous avons toujours une mesure de Gibbs unique qui est aussi la seule mesure d'équilibre. Pour une dimension $d>1$ la situation est radicalement différente et nous pouvons avoir plusieurs mesures de Gibbs même pour un potentiel à courte portée, l'exemple le plus connu est le modèle d'Ising.

Les états d'équilibre à température zéro ({\it les états fondamentaux}) sont les mesures de probabilité invariantes qui minimisent
\[\displaystyle \int \varphi d\nu\]
sur toutes les mesures de probabilité invariantes $\nu$. En autres termes, étant donné un potentiel, nous avons que les points d'accumulation  pour la topologie faible*  des états d'équilibre quand $\beta\rightarrow+\infty $ sont nécessairement les mesures minimisantes pour le potentiel $\varphi$. Une étude plus détaillée sur les limites possibles lorsque le système se fige et comment elle sont liées aux configurations avec une énergie minimale peut être trouvée dans \cite{vEFS}.

Chazottes et Hochman~\cite{CH} ont montré dans le cas unidimensionnel un exemple de potentiel Lipschitz $\varphi$ (mais à longue portée) où la suite $\mu_{\beta\varphi}$ ne converge pas lorsque $\beta\rightarrow +\infty$. Ici, $\mu_{\beta\varphi}$ est l'unique mesure de Gibbs invariante par translation (ou l'unique mesure de Gibbs) à  la température inverse $\beta>0$ (qui est également l'unique mesure d'équilibre). En revanche, \cite{Bremont, CGU, GT, Leplaideur} ont montré qu'une interaction de courte portée dans le cas unidimensionnel sur un alphabet fini implique la convergence de $\mu_{\beta\varphi}$. Le cas où $\mathcal{A}$ est un ensemble dénombrable a également été étudié dans~\cite{Kempton}. La construction d'exemples de non-convergence a été donnée par van Enter et W. Ruszel~\cite{vER}, où un exemple de potentiel de courte portée sur un espace d'états continu et un comportement chaotique ont été construits. Récemment, l'argument de van Enter et Ruszel a été implémenté pour le cas où $\mathcal{A} $ est un ensemble fini dans~\cite{BGT, BLL, CR-L}.

Chazottes et Hochman~\cite{CH} ont également montré que le même type de non-convergence peut être observé lorsque la dimension est $d\geq 3$ même pour un potentiel localement constant (à courte portée). La construction de leur exemple n'est possible que pour $d\geq 3$ car ils s'appuient fortement sur la théorie des sous-shifts multidimensionnels de type fini et des Machines de Turing, développée par Hochman~\cite{Hochman} qui fournit une méthode pour transférer une construction  unidimensionnelle  à un sous-shifts de type fini, mais de dimension supérieure. Grâce au théorème de Hochman, Chazottes et Hochman ont pu construire un exemple pour $d=3$ avec un potentiel $\varphi $ localement constant sur un espace d'états fini. Leur construction peut être facilement étendue à n'importe quelle dimension $ d\geq3 $. Ces résultats nous amènent à croire que l'énoncé est également vrai pour $ d = 2 $. Notre résultat principal est double: nous étendons le théorème du comportement chaotique de Chazottes-Hochman en dimension 2 en utilisant une approche différente impliquant le diagramme espace-temps d'une machine de Turing développée par Aubrun-Sablik et nous clarifions le rôle de la reconstruction et la complexité relative fonction de l'extension par un sous-shift de type fini qui manque dans les arguments de Chazottes-Hochman.

Le résultat principal d'Aubrun et Sablik~\cite{AS}, appelé théorème de simulation, affirme que tout sous-shift $d$-dimensionnel défini par un ensemble de motifs interdits énumérés par une machine de Turing est une sous-action d'un sous-shift de type fini $(d+1)$-dimensionnel. Il existe d'autres travaux dans lesquels les résultats de simulation obtenus jusqu'ici dans cette théorie ont été améliorés~\cite{DRS, DRS2}. Dans ces travaux les auteurs améliorent les résultats obtenus jusqu'à présent en diminuant la dimension du sous-shift de type fini qui génère le sous-shift effectif ; cependant les preuves sont basées sur le théorème du point fixe de Kleene et n'utilisent pas d'arguments géométriques.

La construction d'Aubrun et Sablik~\cite{AS} améliore la méthode de Hochman~\cite{Hochman} en augmentant uniquement de 1 la dimension du SFT, en particulier, elle permet d'obtenir  la construction de Chazottes et Hochman~\cite{CH} en dimension 2.

Dans le deuxième chapitre, nous présentons les principales définitions du formalisme thermodynamique, les résultats classiques et les notations standards. Nous commençons par la définition des sous-shift et définissons une classe spéciale de sous-shift basée sur la concaténation de blocs de même taille afin de former chaque configuration possible. Dans la deuxième section de ce chapitre, nous présentons une brève revue de l'entropie traitant des partitions, de l'entropie d'une partition, de l'entropie métrique et topologique et des concepts de pression, de mesure d'équilibre et de mesure de Gibbs. Dans la troisième section nous donnons une idée générale des opérations transformant un sous-shift en un autre basé sur \cite{Aubrun} afin d'appréhender la notion de simulation d'un sous-décalage par un autre. Enfin, nous présentons une définition formelle d'une machine de Turing, comment représenter le travail d'une machine de Turing dans un diagramme espace-temps et aussi une idée de la construction d'Aubrun et de Sablik~\cite{AS}.

Le troisième chapitre est dédié à la construction de notre exemple en s'inspirant de la construction présentée dans les travaux de Chazottes et Hochman~\cite{CH}. Nous définissons d'abord un sous-shift unidimensionnel basé sur un processus d'itération qui nous donne à chaque étape des blocs de même longueur qui sont concaténés pour former un sous-shift tel que défini au deuxième chapitre. Nous montrons que le contrôle que nous avons obtenu sur l'ensemble des mots interdits de ce sous-shift, implique qu'il existe une machine de Turing qui liste tous les mots interdits, c'est-à-dire que notre sous-shift est un sous-shift effectivement fermé. De là, nous pouvons utiliser le théorème de simulation d'Aubrun-Sablik~\cite{AS} et obtenir un sous-shift bidimensionnel de type fini qui simule notre sous-shift effectivement fermé unidimensionnel précédent. Toujours dans la deuxième section de ce chapitre, nous prouvons quelques résultats importants qui expliquent comment déconstruire une configuration dans le sous-shift bidimensionnel en tant que motifs concaténés dans un dictionnaire donné. Dans la troisième et dernière partie de ce chapitre, nous définissons une nouvelle coloration pour le sous-shift bidimensionnel, comme dans Chazottes et Hochman~\cite{CH}, qui consiste à dupliquer un symbole distinctif, afin de transférer l'entropie du sous-shift initial vers le sous-shift de type fini obtenu par le théorème de simulation.

Après toutes ces constructions, on se retrouve avec un sous-shift de type fini bidimensionnel $X$ défini sur un alphabet fini $\mathcal{A}$, un entier $D\geq 1$ et un ensemble fini de motifs interdits $\mathcal{F}\subset \mathcal{A}^{\llbracket1, D \rrbracket^2} $. On définit ensuite le potentiel localement constant par site suivant
\[
\begin{array}{rcl}
\varphi:\mathcal{A}^{\ZZ^2} = \Sigma^2 (\mathcal{A}) & \to & \RR \\
x & \mapsto & \varphi(x) = \mathds{1}_F (x) \\
\end{array}
\]
où $F$ est l'ensemble clopen égal à l'union des cylindres générés par chaque motif dans $\mathcal{F}$.

Le dernier chapitre est dédié à la d\'emonstration du résultat principal qui est le suivant.

\begin{theorem}
	Il existe un potentiel localement constant $\varphi:\Sigma^2(\mathcal{A}) \to \RR $, il existe une sous-suite $(\beta_k)_{k\geq0} $ qui tend vers l'infini et deux ensembles compacts et invariants qui sont disjoints et non vides $X_A,X_B$ de $\Sigma^2(\mathcal{A})$, tels que si $\mu_{\beta_k} $ est une mesure d'équilibre  la température inverse $\beta_k$ associée au potentiel $\beta_k\varphi$, le support de n'importe quelle mesure d'accumulation pour la topologie faible* de la suite $(\mu_{\beta_{2k}})_{k\geq0} $ est inclus dans $X_B$, et le support de n'importe quelle mesure d'accumulation pour la topologie faible* de $(\mu_{\beta_{2k+1}})_{k\geq0}$ est inclus dans $X_A$.
\end{theorem}

Le théorème précédent affirme qu'il existe une sous-suite $(\beta_k)_{k\in\NN} $ avec $\beta_k \to +\infty $ telle que tout choix de mesure d'équilibre associé au potentiel $ \beta_k\varphi $ alterne entre deux mesures de probabilité supportées par des ensembles compacts et disjoints. C'est-à-dire qu'il existe un potentiel localement constant par site qui présente une convergence chaotique à température zéro.

Nous calculons en annexe une borne supérieure de la complexité relative et de la fonction de reconstruction du sous-shift de type fini donnée dans \cite{AS}; nous remercions Sebasti\'an Babieri pour de nombreuses discussions sur ce sujet.

\tableofcontents

\chapter{Introduction}

One of the most important problems in equilibrium statistical mechanics consists in des\-cri\-bing families of Gibbs states for a given potential or an interaction family. We work with classical lattice systems, which means that our configuration space will be
\[
\Sigma^d(\mathcal{A}):=\mathcal{A}^{\ZZ^d}
\]
where $\mathcal{A}$ is a finite set and $d\in\NN$ is the dimension of our lattice. Let us introduce the function
\[
\varphi:\Sigma^d(\mathcal{A})\to\RR
\]
which is called {\it per site potential } and can be physically interpreted as the energy contribution of the origin of the lattice for each configuration $x\in \Sigma^d(\mathcal{A})$, since we are only considering only translation invariant measures.

Given these elements we denote for every $\beta>0$ the set $\mathcal{G}(\beta\varphi)$ which is the set of {\it Gibbs measures} associated to $\beta\varphi$ at the inverse temperature $\beta$. The are several definitions we could consider as a Gibbs measure, using conformal measures, DLR equations, thermodynamic limits etc. See Georgii~\cite{Georgii}, the classical book about Gibbs measures and \cite{Kimura} for the equivalence of several of these definitions. By compactness we know that this set has at least one shift translation invariant Gibbs measure. In the present thesis we are interested on the behavior of the set of Gibbs measures which are translational-invariant probability measures, called equilibrium measures, when the temperature goes to zero, that is, when $\beta\rightarrow+\infty$.

A probability measure $\mu_\beta$ over $\Sigma^d(\mathcal{A})$ is an {\it equilibrium measure} (or {\it equilibrium state}) at inverse temperature $\beta>0$ for a potential $\beta\varphi$ if it is a shift invariant (or translation invariant) measure which maximizes the pressure, that is if
\[
\dis P(\beta\varphi):=\sup_{\mu \in\mathcal{M}_\sigma(\Sigma^d(\mathcal{A}))}\left\{ h(\mu)-\int \beta\varphi d\mu\right\} =h(\mu_\beta)-\int \beta\varphi d\mu_\beta.
\]
We will consider later the whole set of equilibrium measures $\mu_\beta$ which maximize the pressure $P(\beta\varphi)$ above over all shift invariant probability measures on $\Sigma^d(\mathcal{A})$. The function $h(\nu)$ in the expression of $P(\beta\varphi)$ is the Kolmogorov-Sinai entropy of $\nu$.

In the one-dimensional case if a potential $\varphi$ is Hölder continuous we always have a unique Gibbs measure which is also the only equilibrium measure. For a dimension $d>1$ the situation is dramatically different and we can have multiple Gibbs states even for a potential with finite range, the most famous example is the Ising model. 

The zero-temperature equilibrium states ({\it ground states}) are the shift invariant probability measures which minimize
\begin{displaymath}
\displaystyle\int\varphi d\nu
\end{displaymath}
over all shift-invariant probability measures $\nu$. In other words, given a potential, we have that the weak* accumulation points of equilibrium states as $\beta\rightarrow +\infty$ are necessarily minimizing measures for the potential $\varphi$. A more detailed study on the limit when the system freezes and how it is related with the configurations with minimal energy can be found in \cite{vEFS}.

Chazottes and Hochman~\cite{CH} showed in the one-dimensional case an example of a Lipschitz potential $\varphi$ (but long-range) where the sequence $\mu_{\beta\varphi}$ does not converge when $\beta\rightarrow +\infty$. Here $\mu_{\beta\varphi}$ is the unique shift-invariant Gibbs measure (or the unique Gibbs measure) at the inverse temperature $\beta>0$ (which is also the unique equilibrium measure). On the other hand, \cite{Bremont, CGU, GT, Leplaideur} showed that an interaction of finite-range in the one-dimensional case over a finite alphabet implies the convergence of $\mu_{\beta \varphi}$. The case when $\mathcal{A}$ is a countable set was also studied in~\cite{Kempton}. The breakthrough for the construction of examples of the non-convergence was given by van Enter and W. Ruszel~\cite{vER}, where an example of finite range potential on a continuous state space and chaotic behavior was constructed. Recently the argument of van Enter and Ruszel was implemented for the case where $\mathcal{A}$ is a finite set in~\cite{BGT, BLL, CR-L}.

Chazottes and Hochman~\cite{CH} also showed that the same kind of non-convergence may occur when the dimension is $d\geq 3$ even for a locally constant potential. The construction of their example is possible only for $d\geq 3$ because they rely heavily on the theory of multidimensional subshifts of finite type and Turing Machines, developed by Hochman~\cite{Hochman} that provides a method to transfer a one-dimensional construction to a higher-dimensional subshift of finite type. Thanks to Hochman's theorem, Chazottes and Hochman could construct an example for $d=3$ with a potential $\varphi$ locally constant on a finite state space. Their construction can be easily extended to any dimension $d\geq3$. These results led us to believe that the statement is also true for $d=2$. Our main result is two-fold: we extend Chazottes-Hochman's theorem of chaotic behavior to dimension 2 using a different approach involving the space-time diagram of a Turing machine developed by Aubrun-Sablik and we clarify the role of the reconstruction and relative complexity function of the extension by a subshift of finite type that is missing in Chazottes-Hochman's arguments.

The main result of Aubrun and Sablik~\cite{AS}, called simulation theorem, asserts that any $d$-dimensional subshift defined by a set of forbidden patterns that is enumerated by a Turing machine is a subaction of a $(d+1)$-dimensional subshift of finite type. There are other works in which the simulation results obtained so far in this theory have been improved~\cite{DRS,DRS2}. In these works they improve the results obtained so far by decreasing the dimension of the subshift of finite type which generates the effective subshift, but they are based on Kleene's fixed point theorem and they do not uses geometric arguments.

The construction of Aubrun and Sablik~\cite{AS} improves the method of Hochman~\cite{Hochman}, because they increase the dimension by 1 and this leads us to improve the Chazottes and Hochman~\cite{CH} construction for the dimension 2.

In the second chapter we present the main definitions of thermodynamic formalism and computability, classical results and standard notations. We begin with the definition of subshifts and define a special class of subshifts based on the concatenation of blocks of the same size in order to form each possible configuration. In the second section of this chapter we provide a brief review of entropy dealing with partitions, entropy of a partition, metric and topological entropy and the concepts of pressure, equilibrium measure and Gibbs measure. In the third section we give a general idea of operations transforming a subshift into another one based on \cite{Aubrun} in order to comprehend the notion of simulating a subshift by another one. Finally, we  present a formal definition of a Turing machine, how to represent the work of a Turing machine in a space-time diagram and also an idea of the construction of Aubrun and Sablik~\cite{AS}.

The third chapter is dedicated to define and construct our example that is inspired by the construction presented in the work of Chazottes and Hochman~\cite{CH}. First we define a one-dimensional subshift based on an iteration process that gives us at each step blocks of the same length that are concatenated to form a subshift as defined in Chapter 2. We prove that the control we have obtained over the set of forbidden words of this subshift, implies there exists a Turing machine that lists all of the forbidden words, that is, our subshift is an effectively closed subshift. From there we are able to use the simulation theorem of Aubrun-Sablik~\cite{AS} and obtain a bidimensional subshift of finite type that simulates our previous one-dimensional effectively closed subshift. Also in the second section of this chapter, we prove some important results that explain how to deconstruct a configuration in the $2$-dimensional subshift as concatenated patterns in a given dictionary. In the third and last part of this chapter, we define a new coloring for the bidimensional subshift, as in Chazottes and Hochman~\cite{CH}, that consists in duplicating a distinguished symbol, in order to transfer the entropy of the initial effective subshift to the simulated subshift of finite type.

After all these constructions, we end up with a bidimensional SFT $X$ defined over a finite alphabet $\mathcal{A}$, an integer $D\geq 1$ and a finite set of forbidden patterns $\mathcal{F}\subset\mathcal{A}^{\llbracket1,D\rrbracket^2}$. We then define the following locally constant per site potential
\[
\begin{array}{rcl}
\varphi:\mathcal{A}^{\ZZ^2}=\Sigma^2(\mathcal{A}) & \to & \RR \\
x & \mapsto & \varphi(x)=\mathds{1}_F(x) \\
\end{array}
\]
where $F$ is the clopen set equal to the union of cylinders generated by every pattern in $\mathcal{F}$.

The last chapter is dedicated to prove the main result which is the following.

\begin{theorem}
There exists a locally constant potential $\varphi:\Sigma^2(\mathcal{A})\to \RR$, there exists a subsequence $(\beta_k)_{k\geq0}$ going to infinity and two disjoint non-empty compact invariant sets $X_A,X_B$ of $\Sigma^2(\mathcal{A})$, such that if $\mu_{\beta_k}$ is an equilibrium measure  at inverse temperature $\beta_k$ associated to the potential $\beta_k\varphi$, the support of any weak${}^*$ accumulation point of $(\mu_{\beta_{2k}})_{k\geq0}$ is included in $X_B$,  the support of any weak${}^*$ accumulation point of $(\mu_{\beta_{2k+1}})_{k\geq0}$ is included in $X_A$.
\end{theorem}

The previous theorem asserts that there exists a subsequence $(\beta_k)_{k\in\NN}$ with $\beta_k\to+\infty$ such that any choice of equilibrium measure associated with the potential $\beta_k\varphi$ alternates between two disjoint compact sets of probability measures. That is there exists a locally constant per site potential that exhibits a zero-temperature chaotic convergence.

We compute in the appendix an upper bound of the relative complexity and reconstruction functions of the SFT given in \cite{AS}; we thank S.B. for many discussions on this topic.

\chapter{Subshifts}

\section{Forbidden words}

In this chapter we establish the basic definitions, notations and main results of the objects that we use in this work. We begin by two definitions of a subshift: one topological and one combinatorial. These two definitions coincide.

We will always work with a finite set of letters that we call {\it alphabet} and we will denote it with a cursive letter $\mathcal{A}$. With this alphabet we construct the set of configurations defined over $\ZZ^d$ where $d\geq 1$ is the dimension.

\begin{definition}
\label{def.pattern}
Let $\mathcal{A}$ be a finite alphabet, and $d\geq1$. Let  $S \subseteq \mathbb{Z}^d$ be a subset. A {\it pattern with support  $S$ } is an element of $p$ of $\mathcal{A}^{S}$. We write $S=\Supp(p)$ for the {\it support} of the pattern $p$. If $S' \subseteq S$, the pattern $p'=p|_{S'}$ denotes the restriction of $p$ to $S'$.  A {\it configuration} is a pattern with full support $S=\mathbb{Z}^d$. 

When $d=1$ a one-dimensional finite pattern is called a {\it word}.
\end{definition}

The set of all possible $\ZZ^d$-configurations defined over an alphabet $\mathcal{A}$ is denoted by $\Sigma^d(\mathcal{A}) := \mathcal{A}^{\ZZ^d}$. On this set we define the shift action as follows.

\begin{definition}
\label{def.shift-action}
The shift action on a configuration space $\Sigma^d(\mathcal{A})$ is a collection $\sigma = (\sigma^u)_{u\in \ZZ^d}$ such that
\begin{displaymath}
\begin{array}{rcl}
\sigma^u : \Sigma^d(\mathcal{A}) & \to & \Sigma^d(\mathcal{A}) \\
x & \mapsto & \sigma^u(x)=y, \text{where} \  \ \forall\, v \in\mathbb{Z}^d,\  y_v = x_{u+v}.
\end{array}
\end{displaymath}
\end{definition}

We will use the same notation for the shift acting on a finite pattern, that is, if $S\subset \ZZ^d$ is a finite set and $p\in\mathcal{A}^S$ is a pattern, then we can write for all $u\in\ZZ^d$ the shift acting on the pattern $p$ as
\begin{displaymath}
\sigma^u (p)=w\in\mathcal{A}^{S-u}\mbox{ where } w_v=u_{v+u}, \, \forall v\in S-u
\end{displaymath}

\begin{remark}
Sometimes we will use the term shift invariant patterns for a class of patterns $p \sim q$ if and only if $q=\sigma^u(p)$, for some $u\in\ZZ^d$. In that sense, the {\bf shape} of the support of the pattern is fixed, but the form can be located in any translate of this support.
\end{remark}

Let $S,T \subset \mathbb{Z}^d$ are two subsets, and $p,q$ be two patterns with support $S$ and $T$, respectively. We say that $p$ is {\it a sub-pattern} of $q$, if $S \subseteq T$ and $p=q|_S$. Similarly we say that $p$ is a sub-pattern of a configuration  $x \in \mathcal{A}^{\mathbb{Z}^d}$, if $p = x|_S$. We can also say that a pattern $p\in\mathcal{A}^S$ {\it appears} in another pattern $q\in \mathcal{A}^T$ (respectively, in a configuration $x\in\mathcal{A}^{\ZZ^d}$) if there exists a vector $u\in\ZZ^d$ such that $\sigma^u(p)$ is a sub-pattern of $q$ (respectively, $\sigma^u(p)$ is a sub-pattern of $x$). In that case we write $p\sqsubset q$ (respectively, $p \sqsubset x$).

\begin{definition}
\label{def.cylinder}
If $p\in\mathcal{A}^{S}$ is a pattern with support $S$, the {\it cylinder generated} by $p$, denoted by $[p]$, is the subset of configurations defined by
\begin{displaymath}
\dis [p] := \{x \in \Sigma^d(\mathcal{A}) : x|_S = p \}.
\end{displaymath}
For $a\in\mathcal{A}$ and $i\in\ZZ^d$ we denote the cylinder
\[
[a]_i=\{x\in\Sigma^d(\mathcal{A}): x_i=a\}.
\]
\end{definition}

\begin{definition}
\label{def.cylinder-set}
Let $P \subseteq \mathcal{A}^{S}$ be a subset of patterns of support $S$. The cylinder generated by $P$ is the subset,
\begin{displaymath}
\dis [P] := \bigcup_{p\in P}[p].
\end{displaymath}
\end{definition}

The following is the topological definition of one of the most important objects that we work with.

\begin{definition}
\label{def.subshift}
A subshift $X$ is a closed subset of $\Sigma^d(\mathcal{A})$ which is invariant under $\sigma^u : \Sigma^d(\mathcal{A}) \to \Sigma^d(\mathcal{A})$ for all $u\in \ZZ^d$, that is, $\sigma^u(X)=X$.
\end{definition}

As said before, there is a combinatorial definition of a subshift, which is given by the set of forbidden patterns as presented below.

\begin{definition}
\label{def.subshiftcomb}
Let $X$ be a subset of $\Sigma^d(\mathcal{A})$. We say that $X$ is a {\it subshift generated by a set $\mathcal{F}$ of forbidden patterns} if  $\mathcal{F} \subseteq \bigsqcup_{R\geq1} \mathcal{A}^{\llbracket 1, R \rrbracket^d}$ is a subset of patterns with finite support and
\begin{displaymath}
X=\Sigma^d(\mathcal{A},\mathcal{F}):=\{ x \in \Sigma^d(\mathcal{A}) : \forall\, p \in \mathcal{F}, \ p \not\sqsubset x \}. 
\end{displaymath}
\end{definition}

The following proposition assures that every subshift is generated by a set of forbidden patterns.

\begin{proposition}
The two definitions of subshift (Definition~\ref{def.subshift} and Definition~\ref{def.subshiftcomb}) coincide. 
\end{proposition}

The entire configuration space $\Sigma^d(\mathcal{A})=\mathcal{A}^{\ZZ^d}$ is a subshift, and we call it the {\it full shift}. We will denote by $(\Sigma^d(\mathcal{A}),\mathcal{B})$ the measurable space where $\mathcal{B}$ is the Borel $\sigma$-algebra generated by the cylinder sets in $\Sigma^d(\mathcal{A})$. We will describe a classification for the subshifts based on the set of forbidden patterns. For the full shift the set of forbidden patterns is empty. If the set of forbidden patterns is finite we will say that subshift is a {\it subshift of finite type} or SFT. When the set of forbidden patterns can be enumerated by a Turing machine, then we say that the subshift is an {\it effectively closed subshift} (we explain what we are considering as a set enumerated by a Turing machine in Section~\ref{sec.Turing.machine}).

Another way of describing a subshift is by its language, that we define next.

\begin{definition}
\label{def.language}
Let $\mathcal{A}$ be a finite alphabet, and $d\geq1$. Let $X$ be a subshift of $\mathcal{A}^{\mathbb{Z}^d}$. The {\it language of $X$}, denoted $\mathcal{L}(X)$, is the set of square patterns that appear in $X$, or more formally,
\begin{equation}
\label{language}
\dis \mathcal{L}(X) := \bigsqcup_{\ell\geq1} \Big\{ p \in \mathcal{A}^{\llbracket 1,\ell \rrbracket^d} :  \exists x \in X, \ \text{s.t.} \ p \sqsubset x \Big\}.
\end{equation}
	
We will denote the set of square patterns of a fixed length $\ell$ as
\begin{equation}
\label{languagelk}
\dis \mathcal{L}(X, \ell) :=  \Big\{ p \in \mathcal{A}^{\llbracket 1,\ell \rrbracket^d} :  \exists x \in X, \ \text{s.t.} \ p = x|_{\llbracket 1, \ell \rrbracket^2} \Big\}.
\end{equation}
\end{definition}

A {\it dictionary $L$ of size $\ell$ and dimension $d$} over the alphabet $\mathcal{A}$ is a subset of $\mathcal{A}^{\llbracket 1,\ell \rrbracket^d}$. A dictionary is a specialized subset of patterns. We say that a dictionary $L$ of size $\ell$ is a {\it sub-dictionary} of $L'$ of size $\ell'$ (where both have the same dimension $d$), if every pattern of $L$ is a sub-pattern of a pattern of $L'$. Given a dictionary we can define the set of all configurations obtained by the infinite concatenation of patterns of this dictionary. In fact, this subset is a subshift as described below.

\begin{definition}
\label{def.concatenatedsubshift}
The {\it concatenated subshift} of a dictionary $L$ of size $\ell$ and dimension $d$ is the subshift of the form
\begin{displaymath}
\begin{array}{rcl}
\langle L \rangle & = & \dis \bigcup_{u \in \llbracket 1,\ell \rrbracket^d} \bigcap_{v \in \mathbb{Z}^d} \sigma^{-(u+v \ell)} [L], \\
   &   &  \\
& = & \dis \Big\{ x \in \Sigma^d(\mathcal{A}) : \exists u \in \llbracket 1, \ell \rrbracket^d, \ \forall\, v \in \mathbb{Z}^d, \ (\sigma^{u+\ell v}(x) )|_{\llbracket 1,\ell \rrbracket^d} \in L \Big\}. \\
\end{array}
\end{displaymath}
\end{definition}

Another important concept concerns the admissibility of a pattern. Given a set of forbidden patterns, we define local and global admissibility.

\begin{definition}
\label{def.local/globaladmissible}
Let $\mathcal{F}\subseteq \mathcal{A}^{\llbracket1,D\rrbracket^d}$ for a fixed $D\geq 2$. We say that a pattern $w\in\mathcal{A}^{\llbracket1,R\rrbracket^d}$ where $R\geq D$ is {\it locally $\mathcal{F}$-admissible} if
\[
\sigma^u(x)|_{\llbracket 1, D \rrbracket^d} \not \in \mathcal{F}, \ \forall\, u\in\llbracket0,R-D\rrbracket^d,
\]
that is, we do not find a pattern of $\mathcal{F}$ inside the pattern $w$. We say that a pattern $w\in\mathcal{A}^{\llbracket1,R\rrbracket^d}$ is {\it globally $\mathcal{F}$-admissible} if there exists $x\in\Sigma^d(\mathcal{A},\mathcal{F})$ such that
\begin{displaymath}
\dis x|_{\llbracket1,R\rrbracket^d}=w.
\end{displaymath}
\end{definition}

It is clear that if a pattern is globally admissible, then it is locally admissible, but the reverse it not always true. The next proposition assures that for every $d$-dimensional subshift, every really large pattern that is locally admissible has a central block that is globally admissible.

\begin{proposition}
\label{proposition.Sebastian}
Let $X = \Sigma^d(\mathcal{A},\mathcal{F})$ be a subshift given by a set of forbidden patterns $\mathcal{F}$. There exists a function $R \colon \NN \to \NN$ so that if $q \in \mathcal{A}^{\llbracket-R(n),R(n)\rrbracket^d}$ is locally admissible, then $p = q|_{\llbracket -n,n\rrbracket^d}$, the restriction of $q$ to $\mathcal{A}^{\llbracket -n,n\rrbracket^d}$, is globally admissible.
\end{proposition}

\begin{proof}
The proof follows from a standard compactness argument as described in Lemma 4.3 of~\cite{Barbieri} in a more general setting.

Suppose such a function does not exist, then there exists $n \in \NN$ such that for every $m \geq n$ there exists a locally admissible pattern $q_m$ of size $m$ such that $p_m = q_m|_{\llbracket -n,n\rrbracket^d}$ is not globally admissible. Let $x_m \in \Sigma^d(\mathcal{A})$ be a configuration such that $x_m|_{\llbracket -m,m\rrbracket^d} = q_m$. By compactness of $\Sigma^d(\mathcal{A})$, we may extract a converging subsequence $x_{m(k)}$ which converges to some $\bar{x} \in A^{\ZZ^d}$.
	
We claim $\bar{x} \in X$. Indeed, if not, there is a forbidden pattern which occurs somewhere in $\bar{x}$. In particular, there is $k \in \NN$ such that the pattern is completely contained in ${\llbracket -m(k),m(k)\rrbracket^d}$. It follows by convergence of the sequence $\{x_{m(k)}\}_{k \in \NN}$ that eventually every pattern $q_{m(k)}$ contains the forbidden pattern. This is a contradiction because $q_m$ is locally admissible. Hence $\bar{x} \in X$.
	
As $\bar{x} \in X$, then $\bar{x}|_{\llbracket -n,n\rrbracket^d}$ is globally admissible, but this is equal to $p_m$ for some $m \in \NN$ and thus not globally admissible. This is again a contradiction. Therefore the function $R$ must exist. It is non-decreasing as subpatterns of globally admissible patterns are themselves globally admissible.
\end{proof}

\section{Entropy and variational principle}

We establish here some of the most important results about entropy of subshifts. The results here were developed by several authors in different approaches and they were able to generalize these results even for amenable group actions and non-compact configuration spaces. Here we focus on the $\ZZ^d$-action over a compact configuration space $\Sigma^d(\mathcal{A}) = \mathcal{A}^{\ZZ^d}$.

We always consider $\Sigma^d(\mathcal{A}) = \mathcal{A}^{\ZZ^d}$ and $\sigma = (\sigma^u)_{u\in\ZZ^d}$ the shift action. We will denote by $\mathcal{M}_1(\Sigma^d(\mathcal{A}))$ the set of all probability measures defined on $\Sigma^d(\mathcal{A})$ and by $\mathcal{M}_\sigma(\Sigma^d(\mathcal{A}))$ the set of shift-invariant probability measures. Here we always consider $(\Sigma^d(\mathcal{A}), \mathcal{B}, \mu)$ as a {\it probability space} where $\mathcal{B}$ is the sigma algebra generated by the cylinder sets and $\mu\in\mathcal{M}_{\sigma}(\Sigma^d(\mathcal{A}))$. 

\begin{definition}
A collection $\mathcal{P}=\{P_1,P_2,...,P_n\}$ of measurable sets is a {\it finite partition} of $\Sigma^d(\mathcal{A})$ if
\begin{itemize}
	\item $P_i\cap P_j = \varnothing$ for $i\neq j$; and
	
	\item $\bigcup_{i} P_i = \Sigma^d(\mathcal{A})$.
\end{itemize}
For a probability space $(\Sigma^d(\mathcal{A}),\mathcal{B},\mu)$ we call a collection of measurable sets $\mathcal{P}=\{P_1,P_2,...,P_n\}$ a {\it $\mu$-partition} if
\begin{itemize}
	\item $\mu(P_i)>0$, $\forall i$;
	
	\item $\mu(P_i\cap P_j)=0$, for $i\neq j$; and
	
	\item $\dis \mu\left(\Sigma^d(\mathcal{A})\setminus\bigcup_{i=1}^{n}P_i\right) = 0$.
\end{itemize}
\end{definition}

One of the most important concepts in thermodynamics is the entropy of a system. Here we present the definition of Shannon entropy and some useful properties that we use in this text. The definitions and results can be found in Keller~\cite{Keller} and Kerr-Li~\cite{KL}.

\begin{definition}
\label{def.entropy-of-a-partition}
The {\it information} of a $\mu$-partition $\mathcal{P}=\{P_1,P_2,...,P_n\}$ is the function $I_\mathcal{P}:\Sigma^d(\mathcal{A})\to\RR$ defined as
\begin{displaymath}
\dis I_\mathcal{P}(x):= -\sum_{P\in\mathcal{P}}\log(\mu(P))\cdot \mathds{1}_P(x).
\end{displaymath}
The {\it entropy of a partition} with respect a measure $\mu$ is given by
\begin{displaymath}
H(\mathcal{P},\mu):=\int I_\mathcal{P}(x)d\mu = -\sum_{i=1}^{n}\mu(P_i)\log(\mu(P_i))
\end{displaymath}
\end{definition}

We will use the notation $H(\mathcal{P})=H(\mathcal{P},\mu)$ when there is no confusion over which measure we are considering in order to not overload the notation.

Given two $\mu$-partitions $\mathcal{P}=\{P_1,P_2,...,P_n\}$ and $\mathcal{Q}=\{Q_1,...,Q_m\}$ of a configuration space $\Sigma^d(\mathcal{A})$, we can define the {\it conditional information of $\mathcal{P}$ given $\mathcal{Q}$} as the function $I_{\mathcal{P}|\mathcal{Q}}:\Sigma^d(\mathcal{A})\to\RR$ defined as
\begin{displaymath}
\dis I_{\mathcal{P}|\mathcal{Q}}(x):= -\sum_{i=1}^{n}\sum_{j=1}^{m}\log\left(\frac{\mu(P_i\cap Q_j)}{\mu(Q_j)} \right)\cdot \mathds{1}_{P_i\cap Q_j}(x).
\end{displaymath}
In the same fashion we can define the {\it conditional entropy of $\mathcal{P}$ given $\mathcal{Q}$} with respect to a measure $\mu$ as the value
\begin{equation}
\label{def.conditional-entropy}
H(\mathcal{P}|\mathcal{Q},\mu):= \int I_{\mathcal{P}|\mathcal{Q}}d\mu = \int H(\mathcal{P},\mu_x^\mathcal{Q})d\mu(x)
\end{equation}
where $(\mu_x^{\mathcal{Q}})_{x\in\Sigma^d(\mathcal{A})}$ is the family of conditional probabilities with respect to $\mathcal{Q}$. We can also express the conditional entropy as the sum
\begin{displaymath}
H(\mathcal{P}|\mathcal{Q},\mu) = -\sum_{i=1}^{n}\sum_{j=1}^{m}\mu(P_i\cap Q_j)\log\left(\frac{\mu(P_i\cap Q_j)}{\mu(Q_j)}\right).
\end{displaymath}

As before we will use the notation $H(\mathcal{P}|\mathcal{Q})=H(\mathcal{P}|\mathcal{Q},\mu)$ when there is no confusion over which measure we are considering in order to not overload the notation.

We say that a partition $\mathcal{P}'$ is a {\it refinement} of another partition $\mathcal{P}$ if every element of $\mathcal{P}'$ is contained in an element of $\mathcal{P}$. We denote as $\mathcal{P}'\succeq\mathcal{P}$.

We denote the {\it common refinement} of two partitions denoted by $\mathcal{P}\vee \mathcal{Q}$ as the partition generated by
\begin{displaymath}
\mathcal{P}\vee\mathcal{Q}:=\{P_i\cap Q_j: P_i\in \mathcal{P}, \, Q_j\in\mathcal{Q} \}.
\end{displaymath}
For a subset $S\subseteq \ZZ^d$ we denote by
\begin{displaymath}
\mathcal{P}^{S}:= \bigvee_{u\in S} \sigma^{-u}\mathcal{P}
\end{displaymath}
the common refinement of the partitions $\sigma^{-u}\mathcal{P}$ where $u\in S$. A partition $\mathcal{P}$ is a {\it $\mu$-generated partition} of $(\Sigma^d(\mathcal{A}),\mathcal{B},\mu)$ if the sigma algebra generated by $\mathcal{P}^{S}$ for every finite subset $S \subset \mathbb{Z}^d$ is equal to $\mathcal{B}  \mod \mu$.

The next lemma gives us the Jensen inequality that will be used many times.

\begin{lemma}
[Jensen's Inequality]
\label{lemma.Jensen}
Consider $I\subset\RR$ an open interval and $\psi:I\to\RR$ a concave function. If $f:\Sigma^d(\mathcal{A})\to I$ a $\mu$-integrable function, then the integral of $\psi\circ f$ is well defined and
\begin{displaymath}
\dis \psi\left(\int fd\mu\right) \geq \int \psi\circ fd\mu.
\end{displaymath}
\end{lemma}

If we consider $\psi:[0,1]\to\RR$ defined as
\begin{equation}
\label{def.psi}
\psi(x)=
\left\{
\begin{array}{ll}
-x\log(x), & 0<x\leq 1 \\
0, & x=0,
\end{array}
\right.
\end{equation}
then $\psi$ is a strictly concave function and therefore we obtain
\begin{equation}
\label{eq.psi}
\psi\left(\sum_{i=1}^n \lambda_i x_i\right)\geq \sum_{i=1}^n \lambda_i \psi(x_i),
\end{equation}
where $x_i\in[0,1]$ and $\lambda_i>0$ for each $i\in\llbracket1,n \rrbracket$ with $\sum_{i=1}^n \lambda_i = 1$. We will use this inequality for the proof of the next lemma which presents some important properties of the entropy.

\begin{lemma}
\label{lemma.prop.entropy}
Consider $\mathcal{P}=\{P_1,...,P_n\}$ and $\mathcal{Q}=\{Q_1,Q_2,...,Q_m\}$ two $\mu$-partitions of $\Sigma^d(\mathcal{A})$. Then
\begin{description}
	\item[$(i)$] $0\leq H(\mathcal{P}|\mathcal{Q})\leq H(\mathcal{P})\leq \log|\mathcal{P}|$;
	
	\item[$(ii)$] $H(\mathcal{P}\vee\mathcal{Q})=H(\mathcal{P})+H(\mathcal{Q}|\mathcal{P})$;
	
	\item[$(iii)$] $H(\mathcal{P})\leq H(\mathcal{Q})+H(\mathcal{P}|\mathcal{Q})$;
	
	\item[$(iv)$] if $\mathcal{Q}\succeq \mathcal{P}$, then $H(\mathcal{P}|\mathcal{Q})=0$.
	
	\item[$(v)$] if $\mathcal{Q}\succeq\mathcal{P}$, then $H(\mathcal{P}\vee\mathcal{Q})=H(\mathcal{Q})\geq H(\mathcal{P})$;
\end{description}
\end{lemma}

\begin{proof}
\begin{description}
	\item[$(i)$] The inequality $0\leq H(\mathcal{P}|\mathcal{Q})$ follows from the definition of the entropy of a partition. Now we will prove that if $\mathcal{R}=\{C_1,...,C_l\}$ is a partition such that $\mathcal{Q}\succeq\mathcal{R}$ we have that
	\begin{equation}
	\label{eq.entropy}H(\mathcal{P}|\mathcal{Q})\leq H(\mathcal{P}|\mathcal{R}).
	\end{equation}
	Denote
	\[
	\lambda_{k,j}:=\frac{\mu(B_j\cap C_k)}{\mu(C_k)} \quad\mbox{and}\quad x_{j,i}=\frac{\mu(A_i\cap B_j)}{\mu(B_j)}.
	\]
	As we are considering $\mathcal{Q}\succeq\mathcal{R}$, $\mu(B_j\cap C_k)$ is equal to $\mu(B_j)$ or $0$, because either $B_j\subseteq C_k$ or $B_j\cap C_k=\varnothing$. Thus for a fixed $i$ and $k$
	\[
	\sum_{j=1}^{m} \lambda_{k,j}x_{j,i}=\sum_{\substack{j\in\llbracket1,m\rrbracket \\B_j\subseteq C_k}}\frac{\mu(A_i\cap B_j)}{\mu(C_k)}=\frac{\mu(A_i\cap C_k)}{\mu(C_k)}.
	\]
	\[
	\begin{array}{rcl}
	\dis H(\mathcal{P}|\mathcal{Q}) & = & \dis \sum_{i=1}^{n}\sum_{j=1}^{m} -\mu(P_i\cap Q_j)\log\left(\frac{\mu(P_i\cap Q_j)}{\mu(Q_j)} \right) \\
	   & = & \dis \sum_{i=1}^{n}\sum_{j=1}^{m} \mu(Q_j)\psi(x_{j,i}) \\
	   & = & \dis \sum_{i=1}^{n}\sum_{j=1}^{m} \left(\sum_{k=1}^l\mu(C_k)\lambda_{k,j} \right)\psi(x_{j,i}) \\
	   & = & \dis \sum_{i=1}^{n}\sum_{k=1}^{l} \mu(C_k)\sum_{j=1}^m\lambda_{k,j}\psi(x_{j,i}) \\
	   & \leq & \dis \sum_{i=1}^n\sum_{k=1}^l\mu(C_k)\psi\left(\sum_{j=1}^m\lambda_{k,j}x_{j,i} \right) \\
	   & = & \dis H(\mathcal{P}|\mathcal{R}).
	\end{array}
	\]
	If we take $\mathcal{R}=\{\Sigma^d(\mathcal{A}) \}$ the trivial partition, we obtain $H(\mathcal{P}|\mathcal{Q})\leq H(\mathcal{P})$.

	In (\ref{eq.psi}) if we consider $x_i=\mu(P_i)$ and $\lambda_i = 1/n$ we obtain that
	\begin{displaymath}
	\begin{array}{rcl}
	\dis -\frac{1}{n} \log\left(\frac{1}{n}\right)  & = & \dis \psi\left(\frac{1}{n}\right) \\
	   & = & \dis \psi\left(\frac{1}{n}\sum_{i=1}^{n}\mu(P_i)\right) \\
	   & \geq & \dis \frac{1}{n}\sum_{i=1}^{n}\psi(\mu(P_i)) \\
	   & = & \dis \frac{1}{n} H(\mathcal{P}),
	\end{array}
	\end{displaymath}
	and therefore $H(\mathcal{P})\leq \log(n) = \log|\mathcal{P}|$.
	
	\item[$(ii)$] Each element of the partition $\mathcal{P}\vee\mathcal{Q}$ is of the form $P\cap Q$ where $P\in\mathcal{P}$ and $Q\in\mathcal{Q}$. Then
	\begin{displaymath}
	\begin{array}{rcl}
	\dis I_{\mathcal{P}\vee\mathcal{Q}}(x) & = & \dis -\sum_{P\in\mathcal{P}}\sum_{Q\in\mathcal{Q}} \log(\mu(P\cap Q))\cdot \mathds{1}_{P\cap Q}(x) \\
	 & = & \dis -\sum_{P\in\mathcal{P}}\sum_{Q\in\mathcal{Q}} \log\left(\frac{\mu(P\cap Q)}{\mu(P)}\cdot \mu(P)\right)\cdot \mathds{1}_{P\cap Q}(x) \\
	 & = & \dis -\sum_{P\in\mathcal{P}}\sum_{Q\in\mathcal{Q}} \log\left(\frac{\mu(P\cap Q)}{\mu(P)}\right)\cdot \mathds{1}_{P\cap Q}(x)-\sum_{P\in\mathcal{P}}\sum_{Q\in\mathcal{Q}} \log(\mu(P))\cdot \mathds{1}_{P\cap Q}(x) \\
	 & = & \dis -\sum_{P\in\mathcal{P}}\sum_{Q\in\mathcal{Q}} \log\left(\frac{\mu(P\cap Q)}{\mu(P)}\right)\cdot \mathds{1}_{P\cap Q}(x)-\sum_{P\in\mathcal{P}}\log(\mu(P))\cdot \mathds{1}_{P}(x) \\
	 & = & I_{\mathcal{P}|\mathcal{Q}}(x)+I_{\mathcal{P}}(x).
	\end{array}
	\end{displaymath}
	By integrating with respect to a measure $\mu$ we obtain that
	\begin{displaymath}
	\dis H(\mathcal{P}\vee\mathcal{Q})=H(\mathcal{P})+H(\mathcal{Q}|\mathcal{P}).
	\end{displaymath}
	
	\item[$(iii)$] By the previous items we obtain that
	\begin{displaymath}
	\begin{array}{rcl}
	H(\mathcal{P}) & = & H(\mathcal{P}\vee\mathcal{Q})-H(\mathcal{Q|\mathcal{P}}) \\
	   & \leq & H(\mathcal{P}\vee\mathcal{Q}) \\
	   & = & H(\mathcal{Q}) + H(\mathcal{P}|\mathcal{Q}).
	\end{array}
	\end{displaymath}
	
	\item[$(iv)$] For any two partitions $\mathcal{P}$ and $\mathcal{Q}$, we have
	\begin{displaymath}
	\begin{array}{rcl}
	H(\mathcal{P}|\mathcal{Q}) & = & \dis \sum_{P\in\mathcal{P}}\sum_{Q\in\mathcal{Q}} -\mu(P\cap Q)\log\left(\frac{\mu(p\cap Q)}{\mu(Q)}\right) \\
	   & = & \dis \sum_{P\in\mathcal{P}}\sum_{Q\in\mathcal{Q}} \mu(Q)\cdot \psi\left(\frac{\mu(P\cap Q)}{\mu(Q)}\right). \\
	\end{array}
	\end{displaymath}
	
	If we consider that $\mathcal{Q}\succeq \mathcal{P}$ each $Q\in\mathcal{Q}$ is completely contained in an element $P\in\mathcal{P}$. Hence each term of the sum above is equal to zero because either $\frac{\mu(P\cap Q)}{\mu(Q)}=0$ or $\frac{\mu(P\cap Q)}{\mu(Q)}=1$, and in both cases we have that
	\begin{displaymath}
	H(\mathcal{P}|\mathcal{Q})=\sum_{P\in\mathcal{P}}\sum_{Q\in\mathcal{Q}} \mu(Q)\cdot\psi\left(\frac{\mu(P\cap Q)}{\mu(Q)}\right) = 0.
	\end{displaymath}
	
	\item[$(v)$] It follows from the items $(iii)$ and $(iv)$.
\end{description}
\end{proof}

\begin{lemma}
Consider $(\Sigma^d(\mathcal{A}),\mathcal{B},\mu)$ a shift-invariant probability space and $\mathcal{P}$ a finite partition of $\Sigma^d(\mathcal{A})$. The {\it dynamical entropy relative to the partition $\mathcal{P}$} is given by
\begin{displaymath}
\dis h(\mathcal{P},\mu) := \inf_{n\geq 0}\frac{1}{|\Lambda_n|}H(\mathcal{P}^{\Lambda_n})=\lim_{n\to+\infty}\frac{1}{|\Lambda_n|}H(\mathcal{P}^{\Lambda_n})
\end{displaymath}
which is well defined, where $\Lambda_n:=\llbracket-n,n\rrbracket^d$ for $n \geq 1$.
\end{lemma}

\begin{proof}
For each $n\geq 1$ we will consider $\Lambda_n:=\llbracket-n,n\rrbracket^d\subset\ZZ^d$. For a fixed $m\geq 1$ we denote $\Lambda_m=\llbracket-m,m\rrbracket^d$ and $l_m=2m+1$. Consider the set
\begin{displaymath}
V_n:=\left\{p\in (l_m\ZZ)^2: (p+\Lambda_m)\cap \Lambda_n\neq \varnothing \right\}
\end{displaymath}
Then
\begin{displaymath}
\dis \Lambda_n \subseteq \tilde{\Lambda}_n:=\bigcup_{u\in V_n}\left(\Lambda_m+u \right).
\end{displaymath}

Note that $|\tilde{\Lambda}_n|=|V_n|\cdot|\Lambda_m|\leq|\Lambda_{n+m}|$. We obtain that
\begin{displaymath}
\begin{array}{rcl}
\dis H(\mathcal{P}^{\Lambda_n}) & \leq & \dis H(\mathcal{P}^{\tilde{\Lambda}_n}) \\
   & \leq & \dis \sum_{u\in V_n} H(\sigma^{-u}\mathcal{P}^{\Lambda_m}) \\
   & = & \dis |V_n|H(\mathcal{P}^{\Lambda_m}) \\
   & \leq & \dis \frac{|\Lambda_{n+m}|}{|\Lambda_m|}H(\mathcal{P}^{\Lambda_m}),
\end{array}
\end{displaymath}
and therefore
\begin{displaymath}
\dis \limsup_{n\to+\infty}\frac{1}{|\Lambda_n|}H(\mathcal{P}^{\Lambda_n}) \leq \limsup_{n\to+\infty}\frac{|\Lambda_{n+m}|}{|\Lambda_m|}\frac{1}{|\Lambda_m|}H(\mathcal{P}^{\Lambda_m})=\frac{1}{|\Lambda_m|}H(\mathcal{P}^{\Lambda_m}).
\end{displaymath}

The last estimate holds for every fixed $m$, thus we conclude that
\begin{displaymath}
\dis \limsup_{n\to+\infty}\frac{1}{|\Lambda_n|}H(\mathcal{P}^{\Lambda_n}) \leq \inf_{m>0}\frac{1}{|\Lambda_m|}H(\mathcal{P}^{\Lambda_m}) \leq \liminf_{m\to+\infty}\frac{1}{|\Lambda_m|}H(\mathcal{P}^{\Lambda_m}).
\end{displaymath}
\end{proof}

\begin{theorem}[Shannon-McMillan-Breiman]
\label{teo.SMB}
Let $(\Sigma^d(\mathcal{A}), \mathcal{B},\mu)$ a shift-invariant probability space and $\mathcal{P}$ a finite partition of $\Sigma^d(\mathcal{A})$. Then
\begin{displaymath}
\dis \lim_{n\rightarrow+\infty}-\frac{1}{|\Lambda_n|}\log(\mu(\mathcal{P}^{\Lambda_n})) = h(\mathcal{P},\mu)
\end{displaymath}
pointwise a.e. and in $L_1$.
\end{theorem}

The previous theorem has already been proved for a larger class of group actions only with the assumptions that the group is amenable \cite{Lindenstrauss, KL, TEMP}. The proof for Theorem~\ref{teo.SMB} as stated here can be found in Krengel~\cite{Krengel}.

Now we define the Kolmogorov-Sinai entropy also called dynamical entropy of a measure.

\begin{definition}
\label{def.dynamical.entropy}
The {\it entropy} of the space $(\Sigma^d(\mathcal{A}),\mathcal{B},\mu)$, also known as the {\it dynamical entropy of $\mu$} is given by
\begin{displaymath}
h(\mu) = \sup_{\mathcal{P}}\left\{h(\mathcal{P},\mu): \mathcal{P} \mbox{ is a finite partition} \right\}.
\end{displaymath}
\end{definition}

\begin{definition}
\label{def.topological.entropy}
The {\it topological entropy} of a subshift $X\subseteq\Sigma^d(\mathcal{A})$ is given by
\begin{displaymath}
h_{top}(\Sigma^d(\mathcal{A})) = \lim_{n \to + \infty} \frac{1}{|\Lambda_n|} \log (|\mathcal{L}(X, 2n+1)|).
\end{displaymath}
\end{definition}

In Chazottes-Meyerovitch~\cite{HM} they establish important results about the characterization of the entropy for multidimensional SFT. Next we present the variational principle for the entropy.

\begin{theorem}[Variational Principle]
\label{teo.variational-principle}
Let $X\subseteq\Sigma^d(\mathcal{A})$ be a subshift, then
\begin{displaymath}
h_{top}(X) = \sup_{\mu} h(\mu)
\end{displaymath}
where the supremum is taken over the set of shift-invariant probability measures $\mathcal{M}_\sigma(\Sigma^d(\mathcal{A}))$.
\end{theorem}

The Variational Principle as stated above has already been proved for amenable group actions in \cite{KL}. One important result for the characterization of the dynamical entropy of a measure is given by the following theorem.

\begin{theorem}[Kolmogorov-Sinai]
\label{teo.Kolmogorov.Sinai}
If $\mathcal{P}$ is $\mu$-generated partition for $(\Sigma^d(\mathcal{A}), \mathcal{B},\mu)$ and $H(\mathcal{P})<+\infty$, then
\begin{displaymath}
\dis h(\mu) = h(\mathcal{P},\mu).
\end{displaymath}
\end{theorem}

\begin{proof}
For any finite subset we have that
\begin{equation}
\label{eq.KS}
h(\mathcal{P}^\Lambda,\mu)=h(\mathcal{P},\mu).
\end{equation}
Indeed, consider a fixed $N>0$ such that $\Lambda\subset \Lambda_N$, then we have that
\begin{displaymath}
\begin{array}{rcl}
\dis h(\mathcal{P}^\Lambda,\mu) & = & \dis \lim_{n \to +\infty}\frac{1}{|\Lambda_n|}H\big((\mathcal{P}^\Lambda)^{\Lambda_n}\big) \\
   & \leq & \dis \lim_{n \to +\infty}\frac{1}{|\Lambda_n|}H\big(\mathcal{P}^{\Lambda_{n+N}}\big) \\
   & \leq & \dis \lim_{n \to +\infty}\frac{|\Lambda_{n+N}|}{|\Lambda_n|}\frac{1}{|\Lambda_{n+N}|}H\big(\mathcal{P}^{\Lambda_{n+N}}\big) \\
   & = & \dis h(\mathcal{P},\mu) \\
   & \leq & h(\mathcal{P}^\Lambda,\mu)
\end{array}
\end{displaymath}
since $\mathcal{P}^\Lambda\succeq\mathcal{P}$.

Now consider $\mathcal{P}$ a finite $\mu$-generated partition with finite entropy and $\mathcal{Q}$ a finite partition. From \ref{eq.KS} and Lemma~\ref{lemma.prop.entropy} we obtain that
\begin{displaymath}
\begin{array}{rcl}
\dis h(\mathcal{Q},\mu) & \leq & \dis h(\mathcal{P}^{\Lambda_n},\mu) + H(\mathcal{Q}|\mathcal{P}^{\Lambda_n}) \\
   & = & \dis h(\mathcal{P},\mu) + H(\mathcal{Q}|\mathcal{P}^{\Lambda_n}). \\
\end{array}
\end{displaymath}
As $\lim_{n \to +\infty}H(\mathcal{Q}|\mathcal{P}^{\Lambda_n})=H(\mathcal{Q}|\mathcal{B})=0$, it follows that for an arbitrary partition $\mathcal{Q}$, is true that $h(\mathcal{Q},\mu)\leq h(\mathcal{P},\mu)$, and therefore the result follows.
\end{proof}

\section{Potential}

A function $f:\Sigma^d(\mathcal{A})\to\RR$ is {\it upper semi-continuous} if the set $\{x\in\Sigma^d(\mathcal{A}): f(x)<c \}$ is an open set for every $c\in\RR$.

\begin{definition}
A potential $\varphi:\Sigma^d(\mathcal{A})\to\RR$ is {\it regular} if
\begin{displaymath}
\dis \sum_{n=1}^{+\infty}n^{d-1}\delta_n(\varphi) < +\infty,
\end{displaymath}
where $\delta_n(\varphi):= \sup\{|\varphi(w)-\varphi(v)|:w,v\in\Sigma^d(\mathcal{A}), \, w|_{\Lambda_n}=v|_{\Lambda_n} \}$.
\end{definition}

We say that a potential $\psi$ has {\it finite range} if there exists $n_0\in\NN$ such that $\delta_n(\psi)=0$, for all $n\geq n_0$. If a potential has finite range, then it is regular.

Next we define the pressure of an upper semi-continuous potential, the notion of an equilibrium measure and recall several results that characterize the equilibrium measures for a certain class of potentials.

\begin{definition}
The {\it pressure} of a upper semi-continuous potential $\varphi:\Sigma^d(\mathcal{A})\to\RR$ at inverse temperature $\beta$ is the value
\begin{displaymath}
P(\beta\varphi):=\sup_{\mu\in\mathcal{M}_\sigma(\Sigma^d(\mathcal{A}))}\left\{ h(\mu)-\int\beta\varphi d\mu \right\}.
\end{displaymath}
\end{definition}

\begin{definition}
\label{def.equilibrium.measure}
An {\it equilibrium measure} for a potential $\varphi$ at inverse temperature $\beta$ is a measure $\mu_{\beta\varphi}\in\mathcal{M}_\sigma(\Sigma^d(\mathcal{A}))$ such that
\begin{displaymath}
P(\beta\varphi) = h(\mu_{\beta\varphi})- \int\beta\varphi d\mu_{\beta\varphi}.
\end{displaymath}
\end{definition}

An important characterization for the set of equilibrium measures for a regular local potential is that it is exactly the set of invariant Gibbs measures. In order to state this result, we present one possible definition of Gibbs measures based on \cite{Keller}.

\begin{remark}
Here we will define all these notions and results for the full shift over a finite alphabet, but these definitions and results are also valid for a more general class of subshifts, for instance Muir~\cite{Muir} works with a countable alphabet in multidimensional subshifts and Israel~\cite{Israel} extended to general compact spin spaces and quantum systems for the full shift.
\end{remark}

Consider $\varphi$ a regular potential on $\Sigma^d(\mathcal{A})$ and denote
\[
\varphi_n:=\sum_{g\in\Lambda_n}\varphi\circ \sigma^g
\]
where $\Lambda_n=\llbracket-n,n\rrbracket^d$. We are interested in how $\psi_n(w)$ will change if we alter finitely many sites. For that, we will introduce, as in Keller~\cite{Keller}, a class of local homeomorphisms on $\Sigma^d(\mathcal{A})$.

\begin{definition}
Let $\varphi$ be a regular potential defined over $\Sigma^d(\mathcal{A})$. We denote by $\epsilon_n$ the set of all maps $\tau:\Sigma^d(\mathcal{A})\to\Sigma^d(\mathcal{A})$ such that
\[
(\tau(w))_i=
\left\{
\begin{array}{ll}
\tau_i(w_i), & i\in\Lambda_n \\
w_i, & i\notin\Lambda_n \\
\end{array}
\right.
\]
where $\tau_i:\mathcal{A}\to\mathcal{A}$ are permutations in the state space. We denote by $\epsilon:=\bigcup_{n>0}\epsilon_n$ the set of all homeomorphisms in $\Sigma^d(\mathcal{A})$ that change only finitely many coordinates.
\end{definition}

\begin{lemma}(Keller~\cite{Keller})
Let $\varphi$ be a regular potential and $\tau\in\epsilon$. For $n>0$ define
\[
\Psi_\tau^n:\Sigma^d(\mathcal{A})\to\RR, \quad \Psi_\tau^n:=\varphi_n\circ\tau^{-1}-\varphi_n.
\]
Then the limit
\[
\Psi_\tau:=\lim_{n \to + \infty}\Psi_\tau^n
\]
exists uniformly on $\Sigma^d(\mathcal{A})$.
\end{lemma}

\begin{definition}
Let $\varphi$ be a regular local potential. We say that a probability measure $\mu\in\mathcal{M}_1(\Sigma^d(\mathcal{A}))$ is a {\it Gibbs measure for the potential $\varphi$} if
\[
\tau_{\ast} \mu = \mu \cdot e^{\Psi_\tau}
\]
for each $\tau\in\epsilon$.
\end{definition}

The previous definition goes back to Capocaccia~\cite{Capocaccia} and does not involve conditional measures as in a more classical definition of Gibbs measure \cite{Georgii,Ruelle2004}.

As said before, there are several characterizations for a Gibbs measure (see Georgii~\cite{Georgii} and Ruelle~\cite{Ruelle2004}) and several results for the equivalence between these definitions (see Kimura~\cite{Kimura} and Keller~\cite{Keller}) even for potentials defined over more general subshifts.

The next theorem from Keller~\cite{Keller} gives a important characterization of the set of invariant Gibbs measures for a regular local potential.

\begin{theorem}
\label{teo.equilibrium=gibbs.invariant}
Let $\Sigma^d(\mathcal{A})=\mathcal{A}^{\ZZ^d}$ be the full shift and $\varphi:\Sigma^d(\mathcal{A})\to\RR$ be a regular local potential. The set of equilibrium measures for $\varphi$ is nonempty, compact, convex subset of $\mathcal{M}_\sigma(\Sigma^d(\mathcal{A}))$ and every equilibrium measure is also a Gibbs invariant probability measure.
\end{theorem}

Given a potential $\beta\varphi$ at inverse temperature $\beta$ and $\varphi$ a regular local potential, the set of equilibrium measures is exactly the set of Gibbs invariant measures for $\beta\varphi$.

\section{Turing Machines and the Simulation Theorem}
\label{sec.Turing.machine}

We present here the basic concepts of a Turing machine and how we can characterize a language based on its computability. The automaton that we call Turing machine was first introduced by Alan Turing in 1936 and is similar to a finite automaton but with unlimited and unrestricted memory. This model works on an infinite tape and therefore has unlimited memory. There is a head of calculation which can read and write symbols on the tape and move over the tape, both forward and backward. We will introduce a formal definition of a Turing machine as in Sipser~\cite{Sipser}.

\begin{definition}
A Turing machine $\mathcal{M}$ is a $7-$tuple $(Q,\mathcal{A},\mathcal{T}, \delta, q_0, q_a, q_r)$, where
\begin{itemize}
	\item $Q$ is a finite set of states of the head of calculation;
		
	\item $\mathcal{A}$ is the input alphabet which does not contain the blank symbol $\sharp$;
		
	\item $\mathcal{T}$ is the tape alphabet which contains the blank symbol $\sharp$ and $\mathcal{A}\subseteq \mathcal{T}$;
	
	\item $\delta:Q\times\mathcal{T}\to Q\times\mathcal{T}\times\{-1,+1\}$ is the transition function; 
		
	\item $q_0$ is the initial state of the head of calculation;
		
	\item $q_a\in Q$ is the accept state; and 
	
	\item $q_r\in Q$ is the reject state. 
\end{itemize}
\end{definition}

The machine works on an infinite tape divided into discrete boxes on which the head will act. If we think of $\ZZ$ as a bi-infinite tape filled with symbols of $\mathcal{T}$, we can express the Turing machine $\mathcal{M}$ by describing the state of the head and in which box the head is.

We always start the calculation over a word defined on the alphabet $\mathcal{A}$ that will be written on the tape of the machine. The other boxes of the infinite tape are filled with the blank symbols $\sharp$. The head will start on the leftmost symbol of the word with the initial state $q_0$. At each step of its calculation the head acts (read/write) only on the box where the head is located. Based on the symbol that the head reads and the state of the head, the transition function will give us which symbol the head must write in the box, the new state of the head and in which direction the head should move, $-1$ if it should move for the left box or $+1$ if it should move for the right box. It is possible to define the transition function with the possibility of the head staying in the same box after a calculation, but the definitions are equivalent.

One way of representing the transition function is by a {\it directed graph} where each node represents a state of the head of calculation and the arrows are tagged with the rules of the transition function. See the transition represented below.

\begin{figure}[!htpb]
	\centering
	\psfrag{0}{$q_m$}
	\psfrag{1}{$q_n$}
	\psfrag{a}{$x\to y, +1$}
	\psfrag{c}{$y\to y, +1$}
	\includegraphics[width=0.3\linewidth]{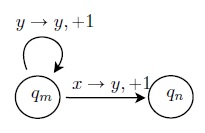}
	\caption{Directed graph representing two rules of some transition function $\delta$.}
	\label{fig:diagrama3}
\end{figure}

If the head of calculation is in the state $q_m$ and it reads the symbol $x$, then the head replaces this symbol by $y$, change of state to $q_n$ and move to the box to the right. If instead the head is in the state $q_m$ and reads the symbol $y$, then the head keeps the symbol $y$ in that box, does not change the state and moves to the box on the right.

The calculation of a Turing machine stops when the head reaches the accept state $q_a$ or the reject state $q_r$. If the machine never reaches one of these states the calculation will never stop. As said before, the calculation of a Turing machine starts over a finite word $w$ defined over the alphabet $\mathcal{A}$ that is written over the tape. If the machine reaches the accept state after a number of valid transitions, we say that the initial word is accepted by this Turing machine. A set of words $L$, also called {\it language}, is recognized by a Turing machine if the machine reaches the accept state for each word in this set and never reaches the accept state if the word is not in $L$ (the machine can reach a reject state or go into a infinite loop).

{\color{red}
\begin{definition}
\label{def.recursive-enumerable}
A set $L$ of words over an alphabet $\mathcal{A}$ is called {\it recursive} if there is a Turing machine that recognizes it. A set $L$ of words over an alphabet $\mathcal{A}$ is called {\it recursively enumerable} if there is a Turing machine that stops its calculation only on words of $L$.
\end{definition}
}

As said before the machine can also reach the reject state or enter in an infinite loop that never stops. There is a special classification for the set of words for which it is possible to define a Turing machine that never enters in a infinite loop, that is, for each finite initial word the machine always reaches $q_a$ or $q_r$. {\color{red} In this case we say that this Turing machine {\it decides} or, most popularly found in the literature, {\it recognizes} the language $L$.}

{\color{red}  These two concepts of recognizability and recursive enumerability, although seemingly equivalent, are two different notions. There are certain languages that only can be enumerate by a Turing machine. Now we present an example presented in \cite{Barbieri_thesis} of a Turing machine that recognizes (and also enumerates) a language defined over the alphabet $\mathcal{A}=\{a,b\}$.}

\begin{example}
This machine stops for every word that we write on the tape and it tells us whether such word belongs or not to the language $\mathcal{L}=\{a^nb^n; n\in \NN\}$. The input alphabet is $\mathcal{A}=\{a,b\}$ and the tape alphabet is $\mathcal{T}=\{a,b,\sharp\}$, where $\sharp$ is the blank symbol. We start with the word to be evaluated written on a bi-infinite tape filled with black symbols $\sharp$ and we set the head of calculation on the state $q_0$ on the leftmost symbol of the word. This Turing machine has 9 states $Q=\{q_0,q_1,q_2,q_3,q_4,q_5,q_6,q_a,q_r\}$ and the transition function $\delta:Q\times \mathcal{T}\rightarrow Q\times \mathcal{T}\times\{-1,+1\}$ is represented by the directed graph in Figure~\ref{fig:diagrama2b}.

\begin{figure}[!htpb]
	\centering
	\psfrag{0}{$q_0$}
	\psfrag{1}{$q_1$}
	\psfrag{2}{$q_2$}
	\psfrag{3}{$q_3$}
	\psfrag{4}{$q_4$}
	\psfrag{5}{$q_5$}
	\psfrag{6}{$q_6$}
	\psfrag{7}{$q_a$}
	\psfrag{8}{$q_r$}
	\psfrag{a}{$a\rightarrow \sharp,+1$}
	\psfrag{b}{$a\rightarrow a,+1$}
	\psfrag{c}{$b\rightarrow b,1$}
	\psfrag{d}{$b\rightarrow b,1$}
	\psfrag{e}{$\sharp\rightarrow\sharp,-1$}
	\psfrag{f}{$b\rightarrow\sharp,-1$}
	\psfrag{g}{$b\rightarrow b,-1$}
	\psfrag{h}{$b\rightarrow b,-1$}
	\psfrag{i}{$a\rightarrow a,-1$}
	\psfrag{j}{$a\rightarrow a,-1$}
	\psfrag{k}{$\sharp\rightarrow\sharp,+1$}
	\psfrag{l}{$\sharp$}
	\psfrag{m}{$b$}
	\psfrag{n}{$\sharp$}
	\psfrag{o}{$a$}
	\psfrag{p}{$a$}
	\psfrag{q}{$\sharp$}
	\includegraphics[width=\linewidth]{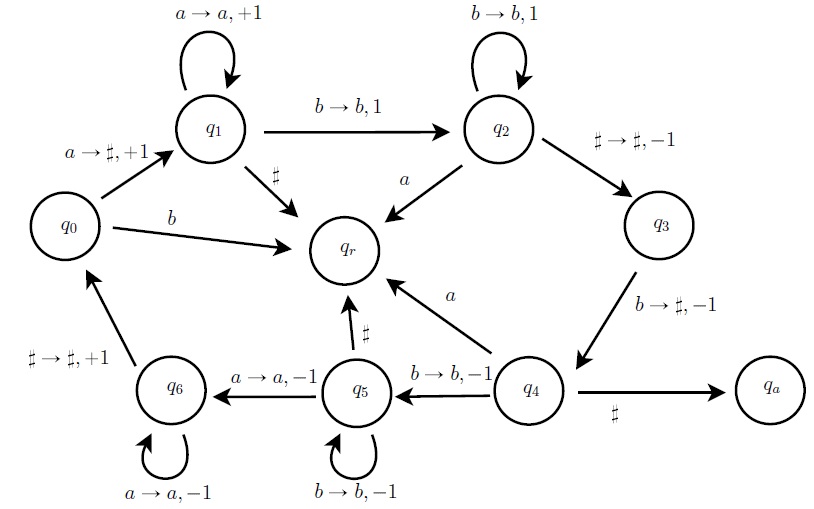}
	\caption{Directed graph representing the transition function for the Turing machine that decides the language $a^n b^n$.}
	\label{fig:diagrama2b}
\end{figure}

We are representing the accept state by $q_a$ and the reject state by $q_r$. Note that the transition function is not defined for every possible pair in $Q\times\mathcal{T}$ because this configuration never occurs in the calculation process. Another important aspect is that when the transition function goes to $q_a$ or $q_r$, we are not defining the symbol substitution or the move that the head should do, because it is irrelevant since the calculation will stop after this iteration.
	
Now we give a summary of the role played by each of the eight states that the machine can reach:
	
\begin{description}
	\item[$q_0$:] This state marks the beginning of the calculation. The head of the machine begins the calculation on the leftmost letter of the  word written on the tape. If the head reads the symbol $a$ then the head replaces the symbol by a blank symbol, moves to the right and also changes the state. If the head reads a symbol $b$ then the head of the machine goes to the reject state and the computation stops, which means that the word written on the tape does not belongs to the language.
		
	\item[$q_1$:] In this state the head of the machine goes to the rightmost symbol $a$ of the word without changing the symbols or the state of the machine. When the machine finds the first symbol $b$ the head of the machine does not change the letter, but changes the state and moves to the right. In this state the machine goes to the reject state if the head reads the blank symbol, which means that the word written on the tape has only the symbol $a$.
		
	\item[$q_2$:] This state makes the head of the machine goes to the end of the word without changing the symbols $b$'s that are written on the tape. The head goes to the last symbol $b$ and then when it finds the first blank symbol this state makes the head go to the left, but not replace the blank symbol. If the head is in this state and finds a symbol $a$, it means that in the word written on the tape exists the subword $ba$ which is forbidden in the language $\mathcal{L}$, so the head goes to the reject state and the calculation stops.
		
	\item[$q_3$:] This state always appears on the head when it is on the last symbol $b$ of the finite word written on the tape of calculation. The symbol $b$ is replaced by a blank symbol and the head of calculation moves to the box on the left. The symbol $b$ is the only possibility for the head to read.
		
	\item[$q_4$:] In this state if the head of the machine reads the symbol $b$ it means that there exists still symbols written on the tape of calculation that are different from the blank symbol, then the head of the machine does not replace the symbol $b$, but moves to the left and changes the state. If the head of the machine in this state reads the blank symbol it means that now, on the tape of calculation, there are only blank symbols, which means that the machine has replaced all of the symbols $a$'s and $b$'s in the initial finite word written and the number of $a$'s and $b$'s are the same. In this case the machine changes to the accept state which means that the initial word written on the tape belongs to the language $\mathcal{L}$. The other possibility is that the head of the machine in this state reads the symbol $a$ which means that the number of symbol $a$'s is bigger than the number of symbol $b$'s and then the machine changes to the reject state.
		
	\item[$q_5$:] This state makes the head of the machine reach the symbol $a$ most to the right on the word written on the tape. The head on this state when placed on the symbol $b$, does not replace the symbol $b$ and only moves to the left without changing the state. When the head reaches one symbol $a$ the machine still moves to the left without replacing the letter, but it changes the state. If the head in this state reaches a blank symbol this means that on the tape of calculation there are only letters $b$'s which means that the number of symbol $b$'s on the initial word is bigger than the number of letters $a$'s. In this case the machine changes to the reject state which means that the machine recognizes that the initial word written on the tape does not belong to the language $\mathcal{L}$.
		
	\item[$q_6$:] This state makes that the head of the calculation go to the leftmost symbol not blank on the tape. If the head in this state reads the letter $a$, the head does not change the state but moves to the left. When the head reaches a blank symbol this means that the head reaches the beginning of the word that is now written on the tape. In this case the head does not replace the blank symbol, changes the state and moves to the right leaving the head on the leftmost symbol on the word that is written on the tape. In this state it is not possible that the head reads the letter $b$ because of the construction and the way that the previous calculations occur.
		
	\item[$q_a$:] This is the accept state, which means that if the head of the machine reaches this state then the initial word written on the tape belongs to the language $\mathcal{L}$.
		
	\item[$q_r$:] This is the reject state, which means that if the head of the machine reaches this state then the initial word written on the tape does not belong to the language $\mathcal{L}$
\end{description}
	
\end{example}

The name 'recursively enumerable' comes from a variation of the Turing machine presented that is called {\it enumerator}. We can think of it as a general Turing machine attached to a printer that prints some output words that the machine has written on its tape. An enumerator starts with a infinite tape filled with blank symbols. Each word that this machine prints belongs to a language, that is why we say that this machine enumerates.

\begin{proposition}
Given a set of words $L$ defined over an alphabet $\mathcal{A}$. The set $L$ is recursively enumerable if and only if there is a Turing machine that enumerates it.
\end{proposition}

The next example from \cite{Aubrun} shows a Turing machine that enumerates the language described in the previous example.

\begin{example}
\label{example-diagram}
We describe an example of a Turing machine that enumerates the language $L=\{a^n,b^n, n\in\NN\}$. The input alphabet is $\mathcal{A}=\{a,b\}$ and the tape alphabet is $\mathcal{T}=\{a,b,\sharp,||\}$. This machine has five possible states $Q=\{q_0,q_{a+},q_{b+},q_{b++},q_{||}\}$ and it never stops its calculation. The symbol $||$ helps the machine to know when it must print the word written on the tape. The transition function will be $\delta:Q\times\mathcal{T}\to Q\times \mathcal{T}\times\{-1,+1\}$ given by Figure~\ref{fig}.
	
\begin{figure}
	\centering
	\psfrag{1}{$q_{b+}$}
	\psfrag{2}{$q_{||}$}
	\psfrag{3}{$q_{a+}$}
	\psfrag{4}{$q_{b++}$}
	\psfrag{0}{$q_{0}$}
	\psfrag{a}{$\sharp\rightarrow a,+1$}
	\psfrag{b}{$\sharp\rightarrow b, +1$}
	\psfrag{c}{$\sharp\rightarrow ||,-1$}
	\psfrag{d}{$b\rightarrow b, -1$}
	\psfrag{e}{$a\rightarrow a, +1$}
	\psfrag{f}{$b\rightarrow a, +1$}
	\psfrag{g}{$b\rightarrow b,+1$}
	\psfrag{h}{$||\rightarrow b,+1$}
	\includegraphics[width=0.8\linewidth]{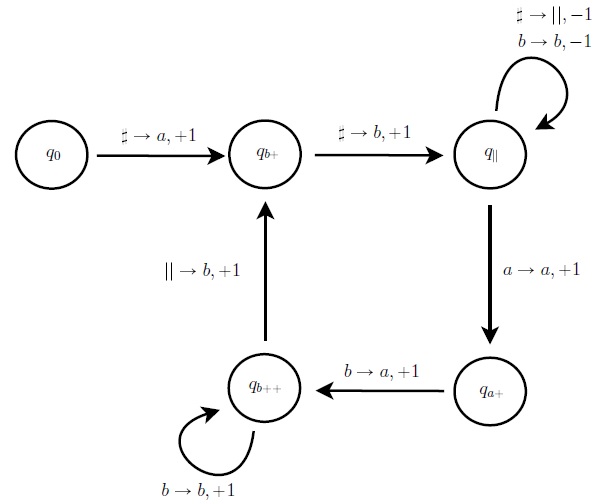}
	\caption{Directed graph of the transition function $\delta$ of the enumerator for the language $a^nb^n$.}
	\label{fig}
\end{figure}
	
The following is a summary of the role played by each of the five states that the machine can reach:
	
\begin{description}
	\item[$q_0$:] This state begins the work of the machine. In our case it always occurs in the bi-infinite tape filled with the blank symbol. It marks the start of the calculation of the machine by replacing the blank symbol by $a$ and moving the head to the right.
		
	\item[$q_{b+}$:] In this state the machine replaces the blank symbol by a letter $b$. This occurs after the head of the machine arrives at the end of the word that is written on the tape of calculation. This symbol $b$ will be the rightmost $b$ required to achieve the same number of letters $b$'s and letters $a$'s in the word written on the tape.
		
	\item[$q_{||}$:] When the machine has this state and reads the blank symbol, that is $(q_{||},\sharp)$, the machine prints the word written on the tape because it will be of the form $a^nb^n$. Besides that, this states is also responsible to return the head of calculation to the rightmost symbol $a$ on the tape. The head changes the blank symbol by a marker $||$ and moves to the left. The head goes to the left without making any changes until it achieves the rightmost symbol $a$ on the tape. The machine does not replace the symbol $a$, but it changes the state and moves to the right, leaving the head over the leftmost symbol $b$ written on the tape.
		
	\item[$q_{a+}$:] This state is responsible for adding a new symbol $a$ into the word written on the tape. It is the beginning of several changes to achieve the next word in the language $a^nb^n$. The head in this state always reads the symbol $b$. It changes to an $a$, it changes the state and it moves to the right.
		
	\item[$q_{b++}$:]  In this state the head of the machine goes to the end of the word written on the tape without making any changes, that is, the head goes to the marker $||$ after all the symbols $b$'s that compose the word on the tape. The head replaces it by a symbol $b$, it moves to the right and it changes the state.
\end{description}
	
The action of this Turing machine can also be described by a space-time diagram. The horizontal direction stands for the tape on which the machine works and the vertical direction for the time evolution of the machine.

\begin{table}
\begin{adjustbox}{width=\columnwidth,center}
\begin{tabular}{c|c|c|c|c|c|c|c|c|c|c|c}
$\cdots$ & $\cdots$ &$\cdots$ &$\cdots$ &$\cdots$ &$\cdots$ &$\cdots$ &$\cdots$ &$\cdots$ &$\cdots$ &$\cdots$ &$\cdots$ \\ \hline
$\cdots$ & $\sharp$ &$a$ &$a$ &$a$ &$a$ &$b$ &$(q_{b++},b)$ &$||$ &$\sharp$ &$\sharp$ &$\cdots$ \\ \hline
$\cdots$ & $\sharp$ &$a$ &$a$ &$a$ &$a$ &$(q_{b++},b)$ &$b$ &$||$ &$\sharp$ &$\sharp$ &$\cdots$ \\ \hline
$\cdots$ & $\sharp$ &$a$ &$a$ &$a$ &$(q_{a+},b)$ &$b$ &$b$ &$||$ &$\sharp$ &$\sharp$ &$\cdots$ \\ \hline
$\cdots$ & $\sharp$ &$a$ &$a$ &$(q_{||},a)$ &$b$ &$b$ &$b$ &$||$ &$\sharp$ &$\sharp$ &$\cdots$ \\ \hline
$\cdots$ & $\sharp$ &$a$ &$a$ &$a$ &$(q_{||},b)$ &$b$ &$b$ &$||$ &$\sharp$ &$\sharp$ &$\cdots$ \\ \hline
$\cdots$ & $\sharp$ &$a$ &$a$ &$a$ &$b$ &$(q_{||},b)$ &$b$ &$||$ &$\sharp$ &$\sharp$ &$\cdots$ \\ \hline
$\cdots$ & $\sharp$ &$a$ &$a$ &$a$ &$b$ &$b$ &$(q_{||},b)$ &$||$ &$\sharp$ &$\sharp$ &$\cdots$ \\ \hline
$\cdots$ & $\sharp$ &$a$ &$a$ &$a$ &$b$ &$b$ &$b$ &$(q_{||},\sharp)$ &$\sharp$ &$\sharp$ &$\cdots$ \\ \hline
$\cdots$ & $\sharp$ &$a$ &$a$ &$a$ &$b$ &$b$ &$(q_{b+},\sharp)$ &$\sharp$ &$\sharp$ &$\sharp$ &$\cdots$ \\ \hline
$\cdots$ & $\sharp$ &$a$ &$a$ &$a$ &$b$ &$(q_{b++},||)$ &$\sharp$ &$\sharp$ &$\sharp$ &$\sharp$ &$\cdots$ \\ \hline
$\cdots$ & $\sharp$ &$a$ &$a$ &$a$ &$(q_{b++},b)$ &$||$ &$\sharp$ &$\sharp$ &$\sharp$ &$\sharp$ &$\cdots$ \\ \hline
$\cdots$ & $\sharp$ &$a$ &$a$ &$(q_{a+},b)$ &$b$ &$||$ &$\sharp$ &$\sharp$ &$\sharp$ &$\sharp$ &$\cdots$ \\ \hline
$\cdots$ & $\sharp$ &$a$ &$(q_{||},a)$ &$b$ &$b$ &$||$ &$\sharp$ &$\sharp$ &$\sharp$ &$\sharp$ &$\cdots$ \\ \hline
$\cdots$ & $\sharp$ &$a$ &$a$ &$(q_{||},b)$ &$b$ &$||$ &$\sharp$ &$\sharp$ &$\sharp$ &$\sharp$ &$\cdots$ \\ \hline
$\cdots$ & $\sharp$ &$a$ &$a$ &$b$ &$(q_{||},b)$ &$||$ &$\sharp$ &$\sharp$ &$\sharp$ &$\sharp$ &$\cdots$ \\ \hline
$\cdots$ & $\sharp$ &$a$ &$a$ &$b$ &$b$ &$(q_{||},\sharp)$ &$\sharp$ &$\sharp$ &$\sharp$ &$\sharp$ &$\cdots$ \\ \hline
$\cdots$ & $\sharp$ &$a$ &$a$ &$b$ &$(q_{b+},\sharp)$ &$\sharp$ &$\sharp$ &$\sharp$ &$\sharp$ &$\sharp$ &$\cdots$ \\ \hline
$\cdots$ & $\sharp$ &$a$ &$a$ &$(q_{b++},||)$ &$\sharp$ &$\sharp$ &$\sharp$ &$\sharp$ &$\sharp$ &$\sharp$ &$\cdots$ \\ \hline
$\cdots$ & $\sharp$ &$a$ &$(q_{a+},b)$ &$||$ &$\sharp$ &$\sharp$ &$\sharp$ &$\sharp$ &$\sharp$ &$\sharp$ &$\cdots$ \\ \hline
$\cdots$ & $\sharp$ &$(q_{||},a)$ &$b$ &$||$ &$\sharp$ &$\sharp$ &$\sharp$ &$\sharp$ &$\sharp$ &$\sharp$ &$\cdots$ \\ \hline
$\cdots$ & $\sharp$ &$a$ &$(q_{||},b)$ &$||$ &$\sharp$ &$\sharp$ &$\sharp$ &$\sharp$ &$\sharp$ &$\sharp$ &$\cdots$ \\ \hline
$\cdots$ & $\sharp$ &$a$ &$b$ &$(q_{||},\sharp)$ &$\sharp$ &$\sharp$ &$\sharp$ &$\sharp$ &$\sharp$ &$\sharp$ &$\cdots$ \\ \hline
$\cdots$ & $\sharp$ &$a$ &$(q_{b+},\sharp)$ &$\sharp$ &$\sharp$ &$\sharp$ &$\sharp$ &$\sharp$ &$\sharp$ &$\sharp$ &$\cdots$ \\ \hline
$\cdots$ & $\sharp$ &$(q_0,\sharp)$ &$\sharp$ &$\sharp$ &$\sharp$ &$\sharp$ &$\sharp$ &$\sharp$ &$\sharp$ &$\sharp$ &$\cdots$ \\ \hline
\end{tabular}
\end{adjustbox}
\end{table}
\end{example}

The calculation of a Turing machine, that is, the set of rules defined by the transition function can be represented by a set of bidimensional patterns as proposed in \cite{Berger}. For example, consider the Turing machine presented in the last example and the transition function when the head of the machine is in the state $q_{||}$ and reads the symbol $a$. In this case the head of the machine does not change the symbol $a$ written on the tape, it changes its state to $q_{a+}$ and it moves to the right. This action can be represented by the following set of $3\times 2$ blocks or tiles described as below
\begin{displaymath}
\begin{tabular}{|c|c|c|}
\hline
$s_1$ & $a$ & $(q_{a+},s_3)$ \\ \hline
$s_1$ & $(q_{||},a)$ & $s_3$ \\ \hline
\end{tabular}
\quad
\begin{tabular}{|c|c|c|}
\hline
$a$ & $(q_{a+},s_2)$ & $s_3$ \\ \hline
$(q_{||},a)$ & $s_2$ & $s_3$ \\ \hline
\end{tabular}
\end{displaymath}
\begin{displaymath}
\begin{tabular}{|c|c|c|}
\hline
$(q_{a+},s_1)$ & $s_2$ & $s_3$ \\ \hline
$s_1$ & $s_2$ & $s_3$ \\ \hline
\end{tabular}
\quad
\begin{tabular}{|c|c|c|}
\hline
$s_1$ & $s_2$ & $a$ \\ \hline
$s_1$ & $s_2$ & $(q_{||},a)$ \\ \hline
\end{tabular}
\end{displaymath}
where $s_1,s_2,s_3\in\mathcal{T}$ are the symbols that have previously been written on the tape. These four patterns describe all the possible $3\times 2$ patterns that can be found in a bidimensional representation of this Turing machine for the rule $\delta(q_{||},a)=(q_{a+},a,+1)$. We can do this representation for each rule of the transition function. Since there is a finite number of rules, the set that describes all the possible $3\times 2$ patterns is also finite. Note that we have to include the pattern
\begin{displaymath}
\begin{tabular}{|c|c|c|}
\hline
$s_1$ & $s_2$ & $s_3$ \\ \hline
$s_1$ & $s_2$ & $s_3$ \\ \hline
\end{tabular}
\end{displaymath}
where the head of the Turing machine does not appear in this window that we are considering.

The set of all possible patterns $3\times 2$ in the alphabet
\begin{displaymath}
\dis \mathcal{T}\cup \left(Q\times \mathcal{T} \right)\cup \left(Q\times \mathcal{T}\times \{-1,+1\} \right)
\end{displaymath}
is finite. Since we are able to describe the language with patterns of the form $3\times 2$, we can take the complementary set from all the possible $3\times 2$ patterns and denote it as the set of forbidden patterns. Therefore, it is always possible to describe the calculation of a Turing machine by a SFT.

Based on the computability of a set of forbidden words, we can define another important class of subshifts.

\begin{definition}
\label{def.effectively-closed-subshift}
We say that a subshift $X\subset \mathcal{A}^\ZZ$ is an {\it effectively closed subshift} if there exists a recursively enumerable set of words $\mathcal{F}$ such that $X = \Sigma^d(\mathcal{A},\mathcal{F})$, that is, the set of forbidden words for the subshift $X$ can be recognized by a Turing machine.
\end{definition}

Here we define this class of subshifts only for one-dimensional subshifts, but it is possible to define the same class for multidimensional subshifts. In our main construction we describe a one-dimensional effectively closed subshift by an iteration process that builds the language of the subshift.

\section{The Aubrun-Sablik simulation theorem}

The simulation theorem in Aubrun-Sablik~\cite{AS} allows us to represent a one-dimensional effectively closed subshift as a subaction of a bidimensional SFT. We introduce some operations in subshifts as defined in \cite{Aubrun} so that we can give an idea of the construction proposed by Aubrun-Sablik~\cite{AS}.

Let $\mathcal{A}$ and $\mathcal{B}$ be two finite alphabets and $X_1\subseteq\Sigma^d(\mathcal{A})$ and $X_2\subseteq\Sigma^d(\mathcal{B})$ be two subshifts of the same dimension $d$. If we consider $x_1\in X_1$ and $x_2\in X_2$ two configurations in each subshift we define
\begin{displaymath}
x_1\times x_2=y\in \Sigma^d(\mathcal{A}\times\mathcal{B})
\end{displaymath}
such that
\begin{displaymath}
y=(y_j)_{j\in\ZZ^d} \mbox{ where } y_j=\left((x_1)_j,(x_2)_j\right)\in\mathcal{A}\times\mathcal{B}.
\end{displaymath}

\begin{definition}
Let be $X_1\subseteq\Sigma^d(\mathcal{A})$ and $X_2\subseteq\Sigma^d(\mathcal{B})$. We define the {\it product of $X_1$ and $X_2$} as the subshift $(X_1\times X_2)\subseteq\Sigma^d(\mathcal{A}\times\mathcal{B})$
\[
X_1\times X_2 = \left\{x_1\times x_2: x_i\in X_i, i=1,2 \right\}.
\]
\end{definition}
Note that the new alphabet is a product alphabet $\mathcal{A}\times\mathcal{B}$ of the two previous alphabets but the dimension of the subshift remains the same.

\begin{definition}
A {\it morphism} $\pi:\Sigma^d(\mathcal{A})\rightarrow\Sigma^d(\mathcal{B})$ is a continuous function which commutes with the shift action, that is,
\[\sigma^u\circ\pi=\pi\circ\sigma^u, \quad \forall u\in\ZZ^d.\]
\end{definition}

Hedlund~\cite{Hedlund} proved that such morphisms are block factors, that is, there exists a finite $U\subset \ZZ^d$ that we call {\it neighborhood} and there exists a function $\overline{\pi}$ such that
\[
\begin{array}{rcl}
\overline{\pi}:\mathcal{A}^U & \rightarrow & \mathcal{B} \\
(w_i)_{i\in\ZZ^d} & \mapsto & \dis \overline{\pi}(w)_i=\overline{\pi}(\sigma^i(x)|_U), \quad \forall i\in\ZZ^d. \\
\end{array}
\]

\begin{definition}
Let $\pi:\Sigma^d(\mathcal{A})\rightarrow\Sigma^d(\mathcal{B})$ be a morphism and $X\subseteq\Sigma^d(\mathcal{A})$ be a subshift. We define the {\it topological factor of the subshift $X$ by $\pi$} as the subshift $X_\pi\subseteq\Sigma^d(\mathcal{B})$ such that
\begin{displaymath}
X_\pi=\left\{ y\in \Sigma^d(\mathcal{B}): \exists x\in X \mbox{ such that } \pi(x)=y \right\}.
\end{displaymath}
\end{definition}

\begin{example}
Consider two alphabets $\mathcal{A}=\{0,1,2\}$ and $\mathcal{B}=\{0,2\}$ and define $X=\Sigma^1(\mathcal{A},\mathcal{F})$ where $\mathcal{F}={\{00,11,02,21\}}$. Let $\overline{\pi}:\mathcal{A} \rightarrow \mathcal{B}$ be a one-to-one block defined as
\[
\overline{\pi}(0)=\overline{\pi}(1)=0 \quad\mbox{and}\quad \overline{\pi}(2)=2.
\]

We can define a morphism $\pi$ as
\[
\begin{array}{rcl}
\pi:\Sigma^1(\mathcal{A}) & \to & \Sigma^1(\mathcal{B}) \\
(x_i)_{i\in\ZZ} & \mapsto & (y_i)_{i\in\ZZ}=(\overline{\pi}(x_i))_{i\in\ZZ}.\\
\end{array}
\]

Thus the topological factor of the subshift $X$ by $\pi$ is
\begin{displaymath}
X_\pi=\left\{x\in\Sigma^1(\mathcal{B}): \mbox{ finite blocks of consecutive 0's are of even length } \right\}
\end{displaymath}
which is called the even shift. This subshift is not a subshift of finite type because we cannot represent the set of forbidden patterns by a finite number of patterns, since one needs to exclude all arbitrarily large blocks of consecutive 0's of odd lengths to describe it.
\end{example}

\begin{remark}
A {\it sofic subshift} is a factor of a subshift of finite type. The class of sofic subshifts is bigger than the class of subshifts of finite type and there exists several representations for a sofic subshift, see \cite{LM}.
\end{remark}

The following definitions of a projective subaction and extension can be generalized for any subgroup as in \cite{Aubrun, Hochman}, but for the purpose of our construction the projective $\ZZ$-subaction and extension by duplication are enough.

\begin{definition}
Let $X\subseteq\Sigma^2(\mathcal{A})$ be a bidimensional subshift defined over the alphabet $\mathcal{A}$. We define the {\it projective $\ZZ$-subaction} as the one-dimensional subshift $Y$ given by
\[
Y=\{y\in\Sigma^1(\mathcal{A}):\exists x\in X, s.t.\, x|_{\ZZ\times\{0\}}=y \},
\]
that is, we are only considering the $e_1=(1,0)$-action on the subshift $X$.
\end{definition}

\begin{definition}
\label{def.extensionbyduplication}
Let $X\subseteq\Sigma^1(\mathcal{A})$ be a subshift. We define the {\it extension by duplication} of the subshift $X$ to be the bidimensional subshift $\overline{X}\subseteq\Sigma^2(\mathcal{A})$ given as
\[
\overline{X}:=\left\{\overline{x}\in\Sigma^2(\mathcal{A}): x|_{\ZZ\times\{0\}}\in X \, \mbox{and}\, \overline{x}|_{(i,j)}=\overline{x}_{(i,j+1)}, \forall(i,j)\in\ZZ^2 \right\}.
\]
\end{definition}

\begin{theorem}[Aubrun and Sablik~\cite{AS}, Durand Romaschenko and Shen~\cite{DRS}]
\label{tAS}
For every effectively closed $\ZZ$-subshift $Z\subseteq\Sigma^1(\mathcal{A})$ there exists an alphabet $\mathcal{B}$, a $\ZZ^2$-subshift of finite type $X\subseteq\Sigma^2(\mathcal{B})$ and a morphism $\pi: \Sigma^2(\mathcal{B}) \to \Sigma^2(\mathcal{A})$ so that
\begin{enumerate}
	\item The topological entropy of $X$ is zero.
	\item The action of $e_2 = (0,1)$ on $X_\pi\subseteq \Sigma^2(\mathcal{A})$ is trivial, that is, the restriction of the action of the subgroup $\{0\} \times \mathbb{Z}$ is the identity on $X_\pi$.
	\item The projective $\ZZ$-subaction of $X_\pi$ is equal to $Z$, that is, the one-dimensional effectively closed subshift $Z$ can be seen as a $\ZZ$-subaction of the topological projection of a bidimensional SFT $X$.
\end{enumerate}
\end{theorem}

The proof of this theorem is constructive and it uses several different elements to construct the final subshift. Among the techniques that they use are the representation of Turing machines via a space-time diagram as in the Example~\ref{example-diagram} as proposed by Berger~\cite{Berger} and the substitution theorem by Mozes~\cite{Mozes}. The final subshift is built as four different layers with four different alphabets that are combined in order to form a really large alphabet in which it is possible to describe a finite set of forbidden patterns that defines a subshift that simulates our first subshift.

As said before, the subshift of finite type $X$ in the Aubrun-Sablik construction~\cite{AS} is composed of four layers, that is, it is a subshift of a product of four subshifts of finite type given by a finite number of forbidden patterns which impose conditions on how the layers superpose. See Figure 14 of~\cite{AS}. The layers are:

\begin{enumerate}
	\item \texttt{Layer} 1: The set of all configurations $x\in \mathcal{A}^{\ZZ^2}$ obtained by the extension by duplication as in Definition~\ref{def.extensionbyduplication}.
	\item \texttt{Layer} 2: $\textbf{T}_{\texttt{Grid}}$ A subshift of finite type extension of a sofic subshift which is generated by the substitution given in Figure 3 of~\cite{AS}. The sofic subshift induces infinite vertical ``strips'' of computation which are of width $2^n$ for every $n \in \NN$ and occur with bounded gaps (horizontally) in any configuration.
	\item \texttt{Layer} 3: $\mathcal{M}_{\texttt{Forbid}}$ A subshift of finite type given by Wang tiles which replicates the space-time diagram of a Turing machine which enumerates all forbidden patterns of $X$ and communicates this information to the fourth layer.
	\item \texttt{Layer} 4: $\mathcal{M}_{\texttt{Search}}$ A subshift of finite type given by Wang tiles which simulates a Turing machine which serves the purpose of checking whether the patterns enumerated by the third layer appear in the first layer.  ``responsibility zone'' which is determined by the hierarchical structure of \texttt{Layer} 2.
\end{enumerate}

The rules between the four layers described in~\cite{AS} force the Turing machine space-time diagrams to occur in every strip, and to restart their computation after an exponential number of steps. This ensures that every configuration restarts the computation everywhere, and that every forbidden pattern is written on the tape by the Turing machine in every large enough strip. The fourth layer searches for occurrences of the forbidden patterns in the first layer and thus discards any configuration in the first layer where one of these patterns occurs.

Based on their construction and the objects that we will define later, it will be possible to have some important estimates.

\chapter{Main Construction}

In this chapter we present the main construction that allows us to define our locally constant potential. First we define a one-dimensional effectively closed subshift generated by an iteration process that defines the language of this subshift. We prove that this subshift is in fact effectively closed. We prove also some important properties. Next we apply the simulation theorem of Aubrun-Sablik~\cite{AS} in order to get a bidimensional SFT that simulates our initial subshift. We also prove some properties for this subshift and define a new coloring of this subshift.

\section{One-dimensional effectively closed subshift}

Now we present a general lemma that we use in our construction. It gives us certain properties based on how we define the iteration process that defines our one-dimensional subshift. See Definition~\ref{def.concatenatedsubshift} for concatenated subshifts.

\begin{lemma}
	\label{lemma.onedimensional}
	Let $\mathcal{A}$ be a finite alphabet. Let $(\ell_k)_{k\geq0}$ be a strictly increasing sequence of integers, and $(L_k)_{k\geq0}$ be a sequence of dictionaries of size $(\ell_k)_{k\geq0}$ over the alphabet $\mathcal{A}$, say $L_k \subseteq  \mathcal{A}^{\llbracket 1, \ell_k \rrbracket}$. We assume that, for every $k\geq0$, every word in $L_{k+1}$ is the concatenation of words of $L_k$. Then 
	\begin{enumerate}
		\item$ \forall\,  k\geq0, \ \langle L_{k+1} \rangle \subseteq \langle L_{k} \rangle$,
		\item $X  := \bigcap_{k\geq0} \langle L_k \rangle  = \Sigma^1(\mathcal{A},\mathcal{F})$ where $\mathcal{F} = \bigsqcup_{k\geq0} \mathcal{F}_k$  and $\mathcal{F}_k$ is the set of words of length $\ell_k$  that are not subwords of the concatenation of two words of $L_{k}$.
	\end{enumerate}
	If we assume in addition that  every concatenation of two words in $L_k$ is a subword of the concatenation of two words of $L_{k+1}$, then
	\begin{enumerate}
		\addtocounter{enumi}{2}
		\item for every $n\geq0$, the concatenation of two words of $L_n$ is a word of the language of $X$.
	\end{enumerate}
\end{lemma}

\begin{proof}
	For this proof we use the following notation: for each $k\geq 0$ and $i\in\ZZ$ we denote $E_k(i)\subset \ZZ$ as the set
	\begin{displaymath}
	\dis E_k(i):=\llbracket i, i+\ell_k-1\rrbracket\subset\ZZ.
	\end{displaymath}
	
	Consider $x\in\langle L_{k+1}\rangle$. By definition there exists $j\in[1,\ell_{k+1}]$ such that
	\begin{displaymath}
	x|_{E_{k+1}(j+1+i\ell_{k+1})}\in L_{k+1}, \quad \forall i\in\ZZ
	\end{displaymath}
	that is, $x$ can be seen as an infinite concatenation of words in $L_{k+1}$. By our assumptions every word in $L_{k+1}$ is a concatenation of words in $L_k$. Then $x\in\langle L_k \rangle$ and that means $\langle L_{k+1} \rangle\subseteq\langle L_k \rangle$.
	
	Now we prove that $X=\Sigma^1(\mathcal{A},\mathcal{F})$ where $\mathcal{F}$ is the set of words of length $\ell_k$, $k\geq0$, that are not subwords of the concatenation of two words of $L_{k}$. For a fixed $k\geq 0$, denote $\mathcal{F}_k$ the set of words of length $\ell_k$ that are not subwords of the concatenation of two words of $L_{k}$. In this case the set $\mathcal{F}_k$ is finite and if $\Sigma^1(\mathcal{A},\mathcal{F}_k)$ is the SFT generated by the set of forbidden words $\mathcal{F}_k$ it is clear that $\langle L_k\rangle\subseteq \Sigma^1(\mathcal{A},\mathcal{F}_k)$. By our assumptions $\langle L_{k+1} \rangle \subseteq \langle L_k \rangle$ for every $k\geq 0$, thus
	\begin{displaymath}
	\dis \bigcap_{i\geq k} \langle L_i\rangle \subset \Sigma^d(\mathcal{A},\mathcal{F}_k).
	\end{displaymath}
	Therefore
	\begin{displaymath}
	\dis X = \bigcap_{k\geq0} \langle L_k \rangle \subseteq \bigcap_{k\geq 0} \Sigma^1(\mathcal{A},\mathcal{F}_k)=\Sigma^1(\mathcal{A},\mathcal{F}).
	\end{displaymath}
	
	For every $k\geq0$, define the interval
	\begin{displaymath}
	\dis I_k := \Big\llbracket 1 - \Big\lfloor \frac{\ell_k}{2} \Big\rfloor,\ell_k- \Big\lfloor \frac{\ell_k}{2} \Big\rfloor  \Big\rrbracket.
	\end{displaymath}
	If we consider $x \in \Sigma^1(\mathcal{A},\mathcal{F})$, then $x|_{I_k}$ is a subword of length $\ell_k$ of the concatenation of two words of $L_k$. For every $k\in\NN$ we can assure that there exists a configuration $y^{k} \in \langle L_k \rangle$ such that $x|_{I_k} = y^{k}|_{I_k}$. We may take a subsequence of indices $k$ such that $(y^k)_{k\geq0}$ converges to some $y \in \mathcal{A}^{\mathbb{Z}}$. Since $y^k \in \langle L_j \rangle$ for every $k \geq j$, by taking the limit in $k$ we obtain $y \in \langle L_j \rangle$, for every $j\geq0$,  thus $y \in X$. For every $k\geq j$, as $I_j \subseteq I_k$, we have  $x|_{I_j}= y^k|_{I_j}$. Since $(y^k)_{k\geq0}$ converges to $y$, $x|_{I_j}=y|_{I_j}$ for every $j\geq0$, thus $x=y \in X$. Therefore $X=\Sigma^1(\mathcal{A},\mathcal{F})$.
	
	Consider two words $u_k,v_k\in L_k$. There exists a configuration $x^k\in\langle L_k\rangle$ such that
	\begin{displaymath}
	\dis x^k|_{\llbracket -\ell_k,\ell_k-1 \rrbracket}=u_kv_k.
	\end{displaymath}
	If the concatenation $u_kv_k$ can be found in a word of $u_{k+1}\in L_{k+1}$, then it is enough to assure there exists a configuration $x\in X$ that $x|_{\llbracket -\ell_k,\ell_k-1 \rrbracket}=u_kv_k$ and therefore $u_kv_k\in\mathcal{L}(X)$.
	
	If $u_kv_k$ is not a subword of a word in $L_{k+1}$, then by our assumptions the concatenation $u_kv_k$ can be seen as a subword of a concatenation of two words in $L_{k+1}$, that is, there exists $u_{k+1},v_{k+1}\in L_{k+1}$ such that $u_kv_k\sqsubset u_{k+1}v_{k+1}$. We can assure again there exists $x^{k+1}\in\langle L_{k+1}\rangle$ such that 
	\begin{displaymath}
	\dis x^{k+1}|_{\llbracket -\ell_{k+1},\ell_{k+1}-1 \rrbracket}=u_{k+1}v_{k+1},
	\end{displaymath}
	and therefore the word $u_kv_k$ appears in the configuration $x^{k+1}$. Hence we assure that for every $j\geq k$ we can find a configuration $x^j\in\langle L_j\rangle$ and two words $u_j,v_j\in L_j$ such that $x^{j}|_{\llbracket -\ell_{j},\ell_{j}-1 \rrbracket}=u_{j}v_{j}$ and $u_kv_k\sqsubset u_jv_j$. We may take a subsequence of indexes $j$ such that $x^j$ converges to some $x\in X$. As we have $\dis \lim_{k\rightarrow+\infty}\ell_k=+\infty$ we obtain a configuration $x\in X$ such that $u_kv_k\sqsubset x\in X$ and therefore $u_kv_k\in \mathcal{L}(X)$.
	
\end{proof}



First we describe a one-dimensional construction that satisfies all of our previous hypotheses and from there we describe our bidimensional elements. We use the notation with a marker $\sim$ for the one-dimensional elements. Consider an alphabet $\tilde{A}=\{0,1,2\}$, a sequence of integers $\ell_k$, sets of blocks $\tilde{A}_k,\tilde{B}_k\subset\tilde{\mathcal{A}}^{\ell_k}$ (or $\tilde{\mathcal{A}}^{\llbracket1,\ell_k\rrbracket}$) and two auxiliary sequences of integers $(N_k)_{k\geq 0}$ and $(N_k')_{k\geq 0}$. We impose assumptions on these sequences in order to properly build our example. We assume that $N'_k \geq 4$ and $N_k$ is a multiple of $N'_k$ for each $k\geq0$. 

\begin{notation}
	\label{notation:BaseLanguageBis}
	For each $k\geq 0$ the sets $\tilde{A}_k$ and $\tilde{B}_k$ will be
	\begin{displaymath}
	\dis \tilde{A}_k=\{a_k,1^{\ell_k}\} \quad\tilde{B}_k=\{b_k,2^{\ell_k}\},
	\end{displaymath}
	where $a_k,b_k\in\tilde{\mathcal{A}}^{\llbracket1,\ell_k\rrbracket}$. We define these blocks by an iteration process described below.
	
	Start with $\ell_0=2$, $a_0=01$ and $b_0=02$, then we have
	\begin{displaymath}
	\tilde{A}_0=\{01,11\} \quad \mbox{and} \quad \tilde{B}_0=\{02,22\}.
	\end{displaymath}
	
	If $k\geq 1$ is odd we define
	\begin{equation}
	\label{rule.odd}
	\begin{array}{l}
	\dis a_k=\underbrace{a_{k-1}a_{k-1}\cdots a_{k-1}}_{N_k\mbox{-times}} \mbox{ and } \\
	\\
	\dis b_k=b_{k-1}2^{(N_k-2)\ell_{k-1}}b_{k-1};
	\end{array}
	\end{equation}
	and if $k\geq 2$ is even we define
	\begin{equation}
	\label{rule.even}
	\begin{array}{l}
	\dis a_k=a_{k-1}1^{(N_k-2)\ell_{k-1}}a_{k-1} \mbox{ and } \\
	\\
	\dis b_k=\underbrace{b_{k-1}b_{k-1}\cdots b_{k-1}}_{N_k\mbox{-times}}.
	\end{array}
	\end{equation}
\end{notation}

In our iteration process, for every $k\geq 0$, the sets $\tilde{A}_k$ and $\tilde{B}_k$ are formed by two blocks of length $\ell_k$ and we always have $1^{\ell_k}\in \tilde{A}_k$ and $2^{\ell_k}\in \tilde{B}_k$. The length of the blocks at each stage is given by
\begin{displaymath}
\dis \ell_k=N_k\ell_{k-1}.
\end{displaymath}

\begin{notation}
	\label{notation:BaseLanguageTer}
	Now we define the sub-dictionaries $\tilde A'_k$ and $\tilde B'_k$ which are made of subwords of length $\ell'_k = N_k'\cdot \ell_{k-1}$ that are either initial or terminal words of a word in $\tilde A_k$ and $\tilde B_k$. Formally,
	\begin{enumerate}
		\item if $k$ is odd, $\tilde A'_k = \{ a'_k, 1^{\ell'_k}\}$, $\tilde B'_k = \{b'_k,b''_k,2^{\ell'_k}\}$,
		\begin{equation}
		\label{rule.odd.prime}
		\begin{array}{l}
		\dis a'_k := a_{k-1}a_{k-1} \cdots a_{k-1}, \quad \text{$N'_k$ times}, \\ 
		\dis b'_k := b_{k-1} 2^{(N'_k-1)\ell_{k-1}} \mbox{ and } \\
		\dis b''_k := 2^{(N'_k-1)\ell_{k-1}}b_{k-1}; \\
		\end{array}
		\end{equation}
		
		\item if $k$ is even, $\tilde A'_k = \{ a'_k,a''_k,1^{\ell'_k}\}$, $\tilde B'_k = \{ b'_k,2^{\ell'_k}\}$,
		\begin{equation}
		\label{rule.even.prime}
		\begin{array}{l}
		\dis a'_k := a_{k-1} 1^{(N'_k-1)\ell_{k-1}}, \\
		\dis a''_k = 1^{(N'_k-1)\ell_{k-1}}a_{k-1} \mbox{ and }\\
		\dis b_k := b_{k-1}b_{k-1} \cdots b_{k-1}, \quad \text{$N'_k$ times}. \\
		\end{array}
		\end{equation}
	\end{enumerate}
	Notice that, as $N_k$ is a multiple of $N'_k$, we  have $\langle \tilde A_k \rangle \subset \langle \tilde A'_k \rangle$ and $\langle \tilde B_k \rangle \subset \langle \tilde B'_k \rangle$.
\end{notation}

\begin{remark}
For each $k\in\NN$, we denote the block of $\ell_k$ consecutive $1$'s by $1_k:=1^{\ell_k}$ and, in a similar fashion $2_k:=2^{\ell_k}$.
\end{remark}

The frequency of the symbol $0$ in any word $\tilde w \in \tilde{\mathcal{A}}^{\llbracket 1, \ell_k \rrbracket}$ of length $\ell_k$ is denoted by
\begin{equation}
\label{def.frequence}
f_k(\tilde w ) := \frac{1}{\ell_k} \Card \big(\{ i\in \llbracket 1, \ell_k \rrbracket : \tilde w(i) = 0 \} \big).
\end{equation}
We denote in the same fashion the frequency of the symbol $0$ in words $\tilde w \in \tilde{\mathcal{A}}^{\llbracket 1, \ell_k' \rrbracket}$ as
\begin{displaymath}
f_k'(\tilde w ) := \frac{1}{\ell_k'} \Card \big(\{ i\in \llbracket 1, \ell_k' \rrbracket : \tilde w(i) = 0 \} \big).
\end{displaymath}

Let $f_k^A$, $f_k^B$ (resp. $f'_k{}^A$, $f'_k{}^B$) be the largest frequency of the symbol $0$ in the words of $\tilde A_k$, $\tilde B_k$ (resp. $\tilde A'_k$, $\tilde B'_k$).

\begin{lemma}
	\label{lemma.frequence.max}
	Let $\tilde A_k$ and $\tilde B_k$ be the two languages defined in Notation~\ref{notation:BaseLanguageBis}, $\tilde A'_k$ and $\tilde B'_k$ those defined in Notation~\ref{notation:BaseLanguageTer}. Then
	\begin{enumerate}
		\item if $k \geq1$ is odd,  then
		\begin{gather*}
		\renewcommand{\arraystretch}{1.2}
		\left\{\begin{array}{l}
		\dis f'_k{}^A = f_k^A = f^A_{k-1}, \ \ f^B_k = \frac{2}{N_k}f^B_{k-1}, \ \ f'_k{}^B = \frac{1}{N'_k}f^B_{k-1}, \\
		\dis f_k^A = \prod_{i=1}^{(k+1)/2} \left( \frac{2}{N_{2i-2}} \right) f_0^A,  \quad f_k^B = \prod_{i=1}^{(k+1)/2} \left( \frac{2}{N_{2i-1}} \right) f_0^B,
		\end{array}\right.
		\end{gather*}
		with $N_0 = 2$;
		\item if $k\geq$ is even, then
		\begin{gather*}
		\renewcommand{\arraystretch}{1.2}
		\left\{\begin{array}{l}
		\dis f^A_k = \frac{2}{N_k}f_{k-1}^A, \ \ f'_k{}^A = \frac{1}{N'_k} f_{k-1}^A, \ \ f'_k{}^B =  f^B_k = f^B_{k-1}, \\
		\dis f_k^A = \prod_{i=1}^{k/2} \left( \frac{2}{N_{2i}} \right) f_0^A, \quad f_k^B = \prod_{i=1}^{k/2} \left( \frac{2}{N_{2i-1}} \right) f_0^B.
		\end{array}\right.
		\end{gather*}
	\end{enumerate}
\end{lemma}

Consider $\tilde{L}_k:=\tilde{A}_k\bigsqcup \tilde{B}_k$ (resp. $\tilde{L}_k':=\tilde{A}_k'\bigsqcup \tilde{B}_k'$). We will say that two words $a,b \in \tilde{\mathcal{A}}^\ell$ overlap if there exists a non-trivial shift $0 <s < \ell$ such that the terminal segment of  length $s$ of the word $a$ coincides with the initial segment of the word $b$ of the same length, or vice-versa by permuting $a$ and $b$. Note that we exclude the overlapping where $a$ and $b$ coincide.

The next three lemmas are technical lemmas that concern some important properties about the possible types of overlapping in the objects that we described before. The first one ensures that there is no possible overlapping between two words one of $\tilde A_k$ and the other one from $\tilde B_k$ (resp. $\tilde A_k'$ and $\tilde B_k'$). The next two lemmas characterize the possible overlaps between any two words at each stage $k$ of the iteration process.

\begin{lemma}
	\label{lemma:NoOverlappingAB}
	In our construction described above, a word from $\tilde A_k'$ and a word from $\tilde B_k'$ never overlap, neither can a word from $\tilde A_k$ and a word from $\tilde B_k$ overlap.
\end{lemma}

\begin{proof}
	Every word in $\tilde A_k'$ ends with the symbol $1$ which does not appear in any word in $\tilde B_k'$. Conversely, every word in $\tilde B_k'$ ends with the symbol $2$ that does not appear in any word in $\tilde A_k'$. The same argument is valid for the words in $\tilde A_k$ and $\tilde B_k$.
\end{proof}

The next lemma is formulated for the case $k$ even, but a similar lemma holds for the case $k$ odd. First we need to fix some notations. Consider $k\geq 1$ an even integer and the even rules described in (\ref{rule.even}) and (\ref{rule.even.prime}). We denote the {\bf initial segment} of length $\ell_{k-1}$ of $a_k$ and $a'_k$ by $a_{k-1}^I$; the {\bf terminal segment} of length $\ell_{k-1}$ of $a_k$ and $a''_k$ by $a_{k-1}^T$; and the remaining segment $1^{(N'_k-1)\ell_{k-1}}$ that we call {\bf marker}. We can represent
\begin{displaymath}
a_k = \underbrace{a_{k-1}}_{a_{k-1}^I} 1^{(N_k-2)\ell_{k-1}} \underbrace{a_{k-1}}_{a_{k-1}^T},
\end{displaymath}
\begin{displaymath}
a_k' = \underbrace{a_{k-1}}_{a_{k-1}^I} \underbrace{1^{(N_k'-1)\ell_{k-1}}}_{\mbox{marker}} \quad \mbox{ and }\quad a_k'' = \underbrace{1^{(N_k'-1)\ell_{k-1}}}_{\mbox{marker}} \underbrace{a_{k-1}}_{a_{k-1}^T}.
\end{displaymath}

We define similarly the initial and terminal segments of $b'_k$ and denoted as $b^I_{k-1}$ and $b^T_{k-1}$, respectively, as shown below
\begin{displaymath}
b'_k = \underbrace{b_{k-1}}_{b_{k-1}^I} b_{k-1}^{(N'_k-2)}\underbrace{b_{k-1}}_{b_{k-1}^T}.
\end{displaymath}
Note that $a_{k-1}^I=a_{k-1}^T=a_{k-1}$ and $b_{k-1}^I=b_{k-1}^T=b_{k-1}$.

\begin{lemma}
	\label{lemma:OverlappingWords}
	Let $k\geq1$ be even, $a_k\in \tilde A_k$ and $b_k\in \tilde{B}_k$ as described in (\ref{rule.even}). Then
	\begin{enumerate}
		\item two words of the same type $a_k$ can only overlap on their initial and terminal segment, that is, $a_{k-1}^I$ of one of the two words overlaps $a_{k-1}^T$ of the other word $a_k$;
		
		\item on the other hand, two words of the same type $b_k$ can overlap exactly on a multiple of $b_{k-1}$ or they have an overlap of length $\ell_{k-2}$ between $b_{k-1}^I$ and $b_{k-1}^T$.
	\end{enumerate}
\end{lemma}

\begin{proof}
	\begin{enumerate}
		\item We consider a non-trivial shift $0 < s < \ell_{k}$ and a word $w \in \tilde{\mathcal{A}}^{\llbracket 1, s+\ell_k \rrbracket}$ made of two overlapping $a_k$:
		\[
		a_k = w|_{\llbracket 1, \ell_k \llbracket}, \quad \tilde a_k := w|_{s+\llbracket 1,\ell_k \rrbracket}, \quad \forall\, i \in \llbracket1,\ell_k\rrbracket, \ \tilde a_k(s+i) = a_k(i).
		\]
		
		We assume first that $0 < s < \ell_{k-1}$. Then on the one hand $a_{k-1}^T$ of $a_k$ starts with the symbol $0$ at the index $i=(N_k-1)\ell_{k-1}+1$. On the other hand the symbol $1$ appears in $\tilde a_k$ at the indices in the range $\llbracket \tilde i, \tilde j \rrbracket := \llbracket s+\ell_{k-1}+1,s+(N_k-1)\ell_{k-1} \rrbracket$. Since $i \in \llbracket \tilde i, \tilde j \rrbracket$ we obtain a contradiction. 
		
		We assume next that $\ell_{k-1} \leq s < (N_k-1)\ell_{k-1}$. Then on the one hand the symbol $1$ appears in $a_k$ at the indices in the range $\llbracket \tilde i, \tilde j \rrbracket := \llbracket \ell_{k-1}+1,(N_k-1)\ell_{k-1}\rrbracket$. On the other hand $\tilde a_k$ starts with the symbol $0$ at the index $i=s+1$. We obtain a contradiction.
		
		We conclude that $s$ should satisfy $s \geq (N_k-1)\ell_{k-1}$: two words of the form $a_k$ can only overlap  on their initial and terminal segments.
		
		\item We notice that $k-1$ is odd and $b_{k-1}$ has the same structure as $a_k$ in the first item. Two words of the form $b_{k-1}$ only overlap on their initial and terminal segments.  Then $b_{k-1}$ cannot be a subword of the concatenation $c=b_{k-1}b_{k-1}$ of two words $b_{k-1}$ unless $b_{k-1}$ coincides with the first or the last $b_{k-1}$ in $c$. If $b_k$ and $\tilde b_k$ overlap, either $\tilde b_k$ has been shifted by a multiple of $\ell_{k-1}$, $s \in \{ \ell_{k-1}, 2 \ell_{k-1}, \ldots, (N'_k-1)\ell_{k-1} \}$. Note that $k-1$ is an odd number, then $b_{k-1}$ has the same behavior as $a_k$ described in the previous item. Therefore, it is only possible to have an overlap of a word $b_{k-2}$ of length $\ell_{k-2}$ between $b_{k-1}^T$ and $\tilde b_{k-1}^I$.
	\end{enumerate}
\end{proof}

\begin{lemma}
	\label{lemma:OverlappingWordsPrime}
	Let $k\geq1$ be an even integer and $a_k'$ and $a_k''$ as described in (\ref{rule.even.prime}). Then the following holds:
	\begin{enumerate}	
		\item two words of the same form $a'_k$ never overlap; the same is true for two words of the same form $a_k''$;
		
		\item two words $a'_k$ and $a''_k$ overlap if and only if they overlap either partially on their marker or partially on their initial and terminal segments, respectively. 
	\end{enumerate}
\end{lemma}

\begin{proof}
	\begin{enumerate}
		\item We consider a non trivial shift $0 < s < \ell'_k$ and two overlapping words of the form $a'_k$ shifted by $s$. Let be $w \in {\tilde{\mathcal{A}}}^{\llbracket 1, s+\ell'_k \rrbracket}$ such that
		\begin{displaymath}
		\dis a'_k = w|_{\llbracket 1, \ell'_k \rrbracket}, \quad \tilde a'_k := w|_{s+\llbracket 1,\ell'_k \rrbracket}, \quad \forall\, i \in \llbracket1,\ell'_k\rrbracket, \ \tilde a'_k(s+i) = a'_k(i).
		\end{displaymath}
		
		We assume first that $\ell_{k-1} \leq s < \ell'_{k}$. On the one hand, $\tilde a'_k$ starts with the symbol $0$, $w(s+1)=0$; on the other hand, $w|_{\llbracket \ell_{k-1}+1, \ell'_k \rrbracket}$ contains only the symbol $1$. Since $s+1 \in \llbracket \ell_{k-1}+1, \ell'_k \rrbracket$ we obtain a contradiction. 
		
		We assume next that $0 < s < \ell_{k-1}$. We observe that $k-1$ is odd and the two initial segments $a^I_{k-1}$ of $a'_k$ and $\tilde a'_k$ are of the same form as $b_k$ in the second item. They overlap on a multiple of words of the form $a_{k-2}$ or at their initial and terminal segments. Necessarily $s \geq l_{k-2}\geq2$. On the one hand, the initial segment of $\tilde a'_k$ ends with the symbols $01$, $w(s+\ell_{k-1}-1)=0$, on the other hand, $w|_{\llbracket\ell_{k-1}+1,\ell'_k \rrbracket}$ contains only the symbol $1$. Since $s+\ell_{k-1}-1 \in \llbracket \ell_{k-1}+1, \ell'_k \rrbracket$ we obtain a contradiction. 
		
		A similar proof works for $a''_k$ instead of $a'_k$.
		
		\item We divided our discussion in two cases. We consider first the case,
		\begin{displaymath}
		\dis a'_k = w|_{\llbracket 1, \ell'_k \rrbracket}, \quad \tilde a''_k := w|_{s+\llbracket 1,\ell'_k \rrbracket}, \quad \forall\, i \in \llbracket1,\ell'_k\rrbracket, \ \tilde a''_k(s+i) = a''_k(i).
		\end{displaymath}
		We assume that $0 < s < \ell_{k-1}$. The terminal segment of $\tilde a''_k$ is a word like $a_{k-1}$ and then it starts with the symbol $0$ which appears in $w$ at the index $s+(N'_k-1)\ell_{k-1} \in \llbracket \ell_{k-1},\ell'_k\rrbracket$. On the other hand $w|_{\llbracket \ell_{k-1}, \ell'_k \rrbracket}$ contains only the symbol $1$. We obtain a contradiction, then necessarily $\ell_k \leq s$ and the two words $a'_k$ and $a''_k$ overlap (partially or completely) on their markers.
		
		We consider next the case,
		\begin{displaymath}
		\dis a''_k = w|_{\llbracket 1, \ell'_k \rrbracket}, \quad \tilde a'_k := w|_{s+\llbracket 1,\ell'_k \rrbracket}, \quad \forall\, i \in \llbracket1,\ell'_k\rrbracket, \ \tilde a'_k(s+i) = a'_k(i).
		\end{displaymath}
		Assume that $0 < s < (N'_k-1)\ell_{k-1}$. The initial segment of $\tilde a'_k$ starts with the symbol $0$ which is located at the index $s+1 \in \llbracket 1, (N'_k-1)\ell_{k-1} \rrbracket$ in $w$. On the other hand $w|_{\llbracket 1,(N'_k-1)\ell_{k-1}\rrbracket}$ is the marker of $a''_k$ and contains only the symbol $1$. We obtain a contradiction, then it is only possible to have $s \geq  (N'_k-1)\ell_{k-1}$, which means that the terminal segment of $a''_k$ overlaps with the initial segment of $a'_k$. Both segments are copies of $a_{k-1}$ and as we consider $k\geq 2$ even, $k-1$ is odd and $a_{k-1}$ has the same behavior described in Lemma~\ref{lemma:OverlappingWords} item 2. Therefore the possible overlap can occur (partially or completely) on their initial and terminal segments by the rules described as in Lemma~\ref{lemma:OverlappingWords} item 2.
	\end{enumerate}
\end{proof}

As defined in (\ref{eq.concatenated-subshift.Lk}) we consider for each $k\geq 0$ the concatenated subshifts generated by the sets $\tilde{L}_k$, $\tilde A_k$ and $\tilde B_k$ that are denoted as $\langle\tilde{L}_k\rangle$, $\langle\tilde A_k\rangle$ and $\langle\tilde B_k\rangle$, respectively.

By the definition of these subshifts we have that for each $k\geq 0$
\begin{displaymath}
\dis \langle\tilde A_k\rangle\subseteq \langle\tilde A_{k+1}\rangle, \quad \langle\tilde B_k\rangle\subseteq \langle\tilde B_{k+1}\rangle
\end{displaymath}
and
\begin{displaymath}
\dis \langle \tilde{L}_{k+1}\rangle\subseteq \langle \tilde{L}_k\rangle.
\end{displaymath}

\begin{lemma}
\label{lemma.LksubsetLkprime}
Consider the iteration process described in Notation~\ref{notation:BaseLanguageBis} and Notation~\ref{notation:BaseLanguageTer}. If we denote $\tilde L_k=\tilde A_k\bigsqcup \tilde B_k$ and $\tilde L_k'=\tilde A_k'\bigsqcup \tilde B_k'$ for each $k\in\RR$, then
\[ \langle\tilde{L}_k\rangle \subseteq \langle\tilde{L}_k'\rangle. \]
\end{lemma}

\begin{proof}
If we consider the iteration process described in Notation~\ref{notation:BaseLanguageBis} and Notation~\ref{notation:BaseLanguageTer}, then $N'_k$ divides $N_k$. More than that, every word of $\tilde A_k$, $\tilde B_k$ is obtained as concatenation of words of $\tilde A'_k$, $\tilde B'_k$ respectively. Therefore, the concatenated subshift $\langle\tilde{L}_k\rangle$ is a subset of $\langle\tilde{L}_k'\rangle$, since every pattern in $\tilde{L}_k'$ is a subpattern in $\tilde{L}_k$.
\end{proof}

We consider
\begin{equation}
\label{eq.Xtilde}
\tilde{X}:=\bigcap_{k\in\NN}\langle \tilde{L}_k\rangle.
\end{equation}
The construction presented here satisfies all the hypotheses of Lemma~\ref{lemma.onedimensional}, therefore $\tilde{X}=\Sigma^1(\tilde{\mathcal{A}},\overline{\mathcal{F}})$ is the subshift generated by the set of forbidden words $\overline{\mathcal{F}}= \bigsqcup_{k\geq0} \tilde{\mathcal{F}}(\ell_k)$, where $\tilde{\mathcal{F}}(\ell_k)$ is the set of words of length $\ell_k$  that are not subwords of the concatenation of two words of $\tilde{L}_{k}$.

From now on we give a specialized algorithm which produces our auxiliary sequences ($N_k$, $\ell_k$, $N_k'$ and $\ell_k'$) and also the choice of $\beta_k$ for each $k$. We introduce two integer numbers $\rho_k^A$ and $\rho_k^B$ that count  the number of symbols $0$ in the words $a_k$ and $b_k$
\[
\rho_k^A := \ell_k f_k^A, \quad \rho_k^B := \ell_k f_k^B.
\]

\begin{definition}[The recursive sequences]
	\label{notation}
	We define the partial recursive function $S : \mathbb{N}^4 \to \mathbb{N}^4$
	\[
	(\ell_k,\beta_k,\rho_k^A,\rho_k^B) = S(\ell_{k-1},\beta_{k-1},\rho_{k-1}^A, \rho_{k-1}^B).
	\] 
	satisfying $\ell_0=2$, $\beta_0=0$, $\rho_0^A = \rho_0^B = 1$ and defined such that the following holds:
	
	In the case $k$ is even:
	\begin{enumerate}
		\item $\displaystyle N'_k :=  \Big\lceil \frac{k \rho_{k-1}^A}{\rho_{k-1}^B} \Big\rceil, \ \ell'_k = N'_k \ell_{k-1}$,
		\item $\displaystyle \beta_k := \Big\lceil\frac{\ell_{k-1}^2 2^{k \ell'_k}}{(\rho_{k-1}^B)^2}\Big\rceil$,
		\item $\displaystyle N_k := N'_k \Big\lceil \frac{k\beta_k}{N'_k \rho_{k-1}^B} \Big\rceil, \ \ell_k = N_k \ell_{k-1}$,
		\item $\displaystyle \rho^A_k =2\rho_{k-1}^A,  \   \rho^B_k =N_k  \rho^B_{k-1}$,
	\end{enumerate}
	
	In the case $k$ is odd: 
	\begin{enumerate}
		\addtocounter{enumi}{4}
		\item $(\ell_k,\beta_k,\rho_k^A,\rho_k^B)$ are computed as before with $A$ and $B$ permuted.
	\end{enumerate}
\end{definition}

The following proposition assures there exists a Turing machine that enumerates all the forbidden patterns of $\tilde{X}$, which means that $\tilde{X}$ is an effectively closed subshift. More than that, this Turing machine can be constructed such that it enumerates the forbidden words in increasing length, it gives an exponential upper bound for the number of steps to enumerate every forbidden word up to a given length and it also gives a trivial reconstruction function ($R(n)=n$) that will be defined later (Definition~\ref{def.reconstructionfunction}).

\begin{proposition}
\label{proposition:algorithm}
Let $\tilde{X}$ be the subshift defined as in (\ref{eq.Xtilde}). Let $\tilde{\mathcal{F}}:=\bigsqcup_{n\in\NN}\tilde{\mathcal{F}}(n)$ where $\tilde{\mathcal{F}}(n)$ is the set of words of length $n$ that are not sub-words of the concatenation of two words of $\tilde L_k$ for some $k\geq0$ such that $\ell_k \geq n$. 
	
Then the following holds:
\begin{enumerate}
	\item $\tilde{X}=\Sigma^1(\tilde{\mathcal{A}},\tilde{\mathcal{F}})$.
	
	\item For every $n\geq0$, there exist unique integers $k\geq1$ and $p\geq2$ satisfying   
	\[
	\ell_{k-1} < n \leq  \ell_k \ \ \text{and} \ \ (p-1)\ell_{k-1} < n \leq p \ell_{k-1}.
	\]
	We denote $\tilde{\mathcal{F}}'(n)$ as the set of words of length $n$ that are not sub-words of any word of the form $\overrightarrow{w_1}\overleftarrow{w_2}$ where $\overrightarrow{w_1}$ is a terminal segment of $w_1$ of length $(p+1)\ell_{k-1}$, $\overleftarrow{w_2}$ is an initial segment of $w_2$  of length $(p+1)\ell_{k-1}$, and $w_1$ or $w_2$ are  either one of the words $a_k,b_k,1_k,2_k$. Then
	\[\tilde{\mathcal{F}}'(n)=\tilde{\mathcal{F}}(n).\]
	
	\item There exists a Turing machine $\mathcal{M}$ that enumerates all patterns of $\tilde{\mathcal{F}}$ in increasing order (words of $\tilde{\mathcal{F}}(n)$ are enumerated before those in $\tilde{\mathcal{F}}(n+1)$). If we denote by $\tau\colon \NN \to \NN$ the function $\tau(n)$ that counts the number of steps that $\mathcal{M}$ takes to enumerate all patterns of $\tilde{\mathcal{F}}$ up to length $n$, then $\tau(n)\leq P(n)|\tilde{\mathcal{A}}|^n$, for some polynomial $P(n)$.
\end{enumerate}
\end{proposition}

The proof for the previous proposition is in Appendix~\ref{appendix}.

The next lemma gives that the sets $\mathcal{L}(\langle\tilde A_k\rangle,\ell_k)$ and $\mathcal{L}(\langle\tilde B_k\rangle,\ell_k)$ can be seen as the set of all possible words of length $\ell_k$ that can be seen as a subword of a concatenation of two words of $\tilde{A}_k$ and $\tilde{B}_k$, respectively.

\begin{lemma}
	Given our construction of $\tilde{A}_k$ and $\tilde{B}_k$ we have that for each $k\geq 0$
	\begin{equation}
	\label{psiA}
	\dis \mathcal{L}(\langle\tilde A_k\rangle,\ell_k)=\left\{w\in\tilde{\mathcal{A}}^{\llbracket1,\ell_k\rrbracket}: \exists a_1,a_2\in \tilde{A}_k \mbox{ such that }  w\sqsubset a_1a_2  \right\}
	\end{equation}
	and
	\begin{equation}
	\label{psiB}
	\dis \mathcal{L}(\langle\tilde B_k\rangle,\ell_k)=\left\{w\in\tilde{\mathcal{A}}^{\llbracket1,\ell_k\rrbracket}: \exists b_1,b_2\in \tilde{B}_k \mbox{ such that }  w\sqsubset b_1b_2  \right\}.
	\end{equation}
\end{lemma}

\section{Bidimensional SFT}

We can apply the construction of Aubrun-Sablik to our one-dimensional effectively closed subshift $\tilde{X}=\Sigma^1(\tilde{\mathcal{A}},\tilde{\mathcal{F}})$ and obtain a bidimensional SFT $\hat{X}\subseteq \Sigma^2(\hat{\mathcal{A}})$ defined over an alphabet $\hat{\mathcal{A}}=\tilde{\mathcal{A}}\times\mathcal{C}$. We are using the symbol $\wedge$ over the objects that are defined for the SFT generated by the Theorem~\ref{tAS}. Let $\hat{\mathcal{F}} \subseteq \mathcal{A}^{\llbracket 1,D \rrbracket^2}$ be a finite set of forbidden patterns such that
\begin{equation}
\dis \hat{X}:= \Sigma^2(\hat{\mathcal{A}},\hat{\mathcal{F}})
\end{equation}
as the corresponding subshift generated by $\hat{\mathcal{F}}$.

\begin{definition}
	\label{definition:BaseLanguages}
	Let ${\mathcal{V}}_*$ be the set of forbidden patterns in $\Sigma^2(\tilde{\mathcal{A}})$ that are not vertically aligned, that is,
	\begin{displaymath}
	{\mathcal{V}}_* := \{ p \in \tilde{\mathcal{A}}^{\{1\} \times \llbracket 1,2\rrbracket} : p(1,1) \not= p(1,2) \}.
	\end{displaymath}
\end{definition}

	Let $\overline{\pi}:\hat{\mathcal{A}}\to\tilde{\mathcal{A}}$ defined as
	\begin{equation}
	\label{eq.overlinepi}
	\left\{
	\begin{array}{rcl}
	\overline{\pi}:\hat{\mathcal{A}}=\tilde{\mathcal{A}}\times\mathcal{C} & \to & \tilde{\mathcal{A}} \\
	(a,c) & \mapsto & \overline{\pi}(a,c)=a; \\
	\end{array}
	\right.
	\end{equation}
	and let $\pi :\Sigma^2(\hat{\mathcal{A}},\hat{\mathcal{F}}) \to \Sigma^2(\tilde{\mathcal{A}})$ be the projection defined as
	\begin{equation}
	\label{eq.pi}
	\left\{
	\begin{array}{rcl}
	\pi:\hat{X}=\Sigma^2(\tilde{\mathcal{A}},\hat{\mathcal{F}}) & \to & \Sigma^2(\tilde{\mathcal{A}}) \\
	x & \mapsto & \pi(x)=\left(\overline{\pi}(x_{(i,j)})\right)_{(i,j)\in\ZZ^2}.
	\end{array}
	\right.
	\end{equation}
	
	We denote
	\[
	\hat{X}_\pi:=\left\{\pi(x): x\in \hat{X}  \right\}.
	\]
	Note that $\tilde X_\pi\subseteq\Sigma^2(\tilde{\mathcal{A}},\mathcal{V}_*)$ since $\hat{\mathcal{F}}$ contains all the patterns that are not vertically aligned.
	
	\begin{remark}
	Here we always use the expression "vertically aligned" to express the vertical alignment over the the first coordinate of $\hat{\mathcal{A}}$, that is, over the one-dimensional alphabet $\tilde{A}$.
	\end{remark}
	
	By Theorem~\ref{tAS}, the projective $\ZZ$-subaction of $\hat{X}_\pi$ is equal to $\tilde X$, which means that
	\begin{displaymath}
	\dis \hat{X}_\pi = \{ x \in \Sigma^2(\tilde{\mathcal{A}},\mathcal{V}_*) : x|_{\ZZ\times\{0\}}  \in \tilde X \}.
	\end{displaymath}
	
	\begin{definition}
	We define $\tilde A'_{k*}\subseteq\tilde{\mathcal{A}}^{\llbracket 1,\ell'_k\rrbracket^2}$ as the bidimensional dictionary of linear size $\ell'_k$ of vertically aligned patterns that project onto $\tilde A'_k$, formally defined as
	\begin{displaymath}
	\dis \tilde A'_{k*} := \big\{p \in \tilde{\mathcal{A}}^{\llbracket 1, \ell'_k \rrbracket^2} : \exists \tilde p \in  \tilde A'_k, \ \text{s.t.} \ \ \forall, (i,j) \in\llbracket 1, \ell'_k \rrbracket^2, \ p(i,j) = \tilde p(i) \big\}.
	\end{displaymath}
	$\tilde B'_{k*}\subseteq\tilde{\mathcal{A}}^{\llbracket 1,\ell'_k\rrbracket^2}$ is defined similarly. We use the notation $\pi_*: \tilde A'_{k*} \to \tilde A'_k$ (resp. $\pi_*: \tilde B'_{k*} \to \tilde B'_k$) to represent the projection of a square pattern $p\in\tilde A'_{k*}$ (resp. $\tilde B'_{k*}$) to its word $\tilde{p}\in\tilde A'_k$ (resp. $\tilde B'_k$) that defines it.
	\end{definition}

We consider a large pattern $p \in \tilde{\mathcal{A}}^{\llbracket 1,n\rrbracket^2}$ and translates $u$ of small squares of size $2\ell'_k$ inside this pattern that are labeled by vertically aligned words of $\tilde A'_k$ or $\tilde B'_k$. Let $k \geq 2$,  $n > 2\ell'_k$, and  $p \in \tilde{\mathcal{A}}^{\llbracket 1,n\rrbracket^2}$. We denote 
\begin{equation}
\label{eq.I}
\dis I(p,\ell'_k):=\left\{u\in\llbracket 0,n-2\ell'_k\rrbracket^2:\sigma^u(p)|_{\llbracket 1,2\ell'_k\rrbracket^2}\in\mathcal{L}(\hat{X}_\pi, 2\ell'_k)\right\},
\end{equation}
\begin{equation}
\label{eq.IA}
\dis I^A(p,\ell'_k):=\left\{u\in\llbracket 0,n-\ell'_k\rrbracket^2:\sigma^u(p)|_{\llbracket 1,\ell'_k \rrbracket^2}\in\tilde A'_{k*}\right\}
\end{equation}
and
\begin{equation}
\label{eq.JA}
\dis J^A(p,\ell'_k):=\bigcup_{u\in I^A(p,\ell'_k)}\left(u+\llbracket 1,\ell'_k \rrbracket^2\right).
\end{equation}
We define $I^B(p,\ell'_k)$ and $J^B(p,\ell'_k)$ similarly with replacing $\tilde A_{k*}'$ for $\tilde B'_{k*}$ in (\ref{eq.IA}) and (\ref{eq.JA}), respectively.

\begin{lemma}
	\label{lemma:AdmissibilityIntermediateScale}
	Let $k \geq 2$,  $n > 2\ell'_k$,  $p \in \tilde{\mathcal{A}}^{\llbracket 1,n\rrbracket^2}$ and the sets defined above. We will denote $\tau'_k=:(\ell'_k,\ell'_k)\in\mathbb{N}^2$. Then $J^A(p,\ell'_k)\cap J^B(p,\ell'_k)=\emptyset$ and for each $u\in I(p,\ell'_k)$
	\begin{displaymath}
	u+ \tau'_k \in J^A(p,\ell'_k) \bigsqcup J^B(p,\ell'_k).
	\end{displaymath}
\end{lemma}

\begin{proof}
	The fact that $J^A(p,\ell'_k)$ and $J^B(p,\ell'_k)$ do not intersect is a consequence of Lemma \ref{lemma:NoOverlappingAB}. Let be $u \in I(p,\ell'_k)$ and  $w_* = \sigma^u(p)|_{\llbracket 1,2\ell'_k \rrbracket^2}$. There exists $w\in \mathcal{L}(\langle \tilde L_k \rangle,2\ell'_k)$  such that $w_*(i,j)=w(i)$ for all $(i,j) \in \llbracket 1,2\ell'_k \rrbracket^2$. By definition of $\langle \tilde L_k \rangle$, $w \sqsubset w_1w_2$ is a subword of the concatenation of two words of $\tilde L_k$. Note that, by Lemma~\ref{lemma.LksubsetLkprime} $\langle\tilde{L}_k\rangle\subseteq \langle\tilde{L}_k'\rangle$. Hence $\mathcal{L}(\langle\tilde{L}_k\rangle,2\ell_k')\subseteq \mathcal{L}(\langle\tilde{L}_k'\rangle,2\ell_k')$
	
	On the other hand, a word in $\tilde L_k$ is either a word of $\tilde A_k$ or a word of $\tilde B_k$. As $\langle \tilde A_k \rangle \subset \langle \tilde A'_k \rangle$ and $\langle \tilde B_k \rangle \subset \langle \tilde B'_k \rangle$, $w_1$ and $w_2$ are obtained as a concatenation of words of $\tilde A'_k$ or $\tilde B'_k$. There exists $0 \leq s < \ell'_k$ such that 
	\begin{displaymath}
	\dis \sigma^s(w)|_{\llbracket 1, \ell'_k \rrbracket} \in \tilde A'_k \bigsqcup \tilde B'_k.
	\end{displaymath}
	
	Then
	\begin{displaymath}
	u+(s,s) \in I^A(p,\ell'_k) \bigsqcup I^B(p,\ell'_k),
	\end{displaymath}
	and therefore
	\begin{displaymath}
	\dis u+ \tau'_k  \in J^A(p,\ell'_k) \bigsqcup J^B(p,\ell'_k).
	\end{displaymath}
\end{proof}

\begin{figure}[h]
	\centering
	\psfrag{n}{$n$}
	\psfrag{u}{$u$}
	\psfrag{2}{$2\ell_k$}
	\psfrag{1}{$u+(s,s)$}
	\includegraphics[width=0.7\linewidth]{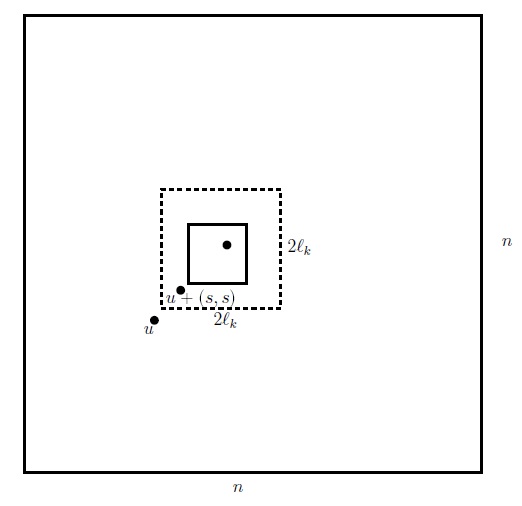}
	\caption{In the figure we are taking a square pattern $p\in \tilde{\mathcal{A}}^{\llbracket 0,n\rrbracket^2}$ shown as the biggest square. We are considering that $u\in I(p,\ell'_k)$ and therefore the patterns located in the dashed square of size $2\ell'_k$ belong to $\mathcal{L}(\hat{X}_\pi, 2\ell'_k)$. We know that the pattern located in the most inner box of size $\ell'_k$ belongs to $\tilde A'_k \bigsqcup \tilde B'_k$. The most inner dot represents $u+\tau'_k$.}
	\label{fig:diagram1008}
\end{figure}

\begin{lemma}
	\label{lemma:FrequencyOf0}
	Let $k\geq2$ be an even integer, $n>2\ell'_k$, and $ p \in \tilde{\mathcal{A}}^{\llbracket 1,n \rrbracket^2}$. Let $I^A(p,\ell'_k)$, $J^A(p,\ell'_k)$, $I^B(p,\ell'_k)$, $J^B(p,\ell'_k)$ be the sets defined in Lemma \ref{lemma:AdmissibilityIntermediateScale}. Define
	\begin{equation}
	\label{eq.KA}
	K^A(p,\ell'_k) = \{ v \in J^A(p,\ell'_k) : p(v) = 0 \}, \ \ K^B(p,\ell'_k) = \{ v \in J^B(p,\ell'_k) : p(v) = 0 \}.
	\end{equation}
	Then
	\begin{enumerate}
		\item $\displaystyle \Card(K^B(p,\ell'_k)) \leq \big(1-N_{k-1}^{-1}\big)^{-1} \Card(J^B(p,\ell'_k)) f_{k-1}^B$,
		\item $\displaystyle\Card(K^A(p,\ell'_k)) \leq \frac{2}{N'_k}\Card(J^A(p,\ell'_k))f_{k-1}^A$.
	\end{enumerate}
\end{lemma}

\begin{proof}
	Let $k\geq 2$ even, $n>2\ell_k'$ and a fixed $p\in\tilde{\mathcal{A}}^{\llbracket 1,n \rrbracket^2}$. To simplify the notations, we write $I^A= I^A(p,\ell'_k)$, $J^A=J^A(p,\ell'_k)$ and so on. As the symbol $0$ does not appear in the markers $1^{N'_k\ell_{k-1}}\in \tilde{A}_k'$ and $2^{N'_k\ell_{k-1}}\in\tilde{B}_k'$, we only need to consider the subset of $I^A$ (resp. $I^B$) that corresponds to the translates  $u \in \llbracket 0, n-\ell'_k \rrbracket^2$ and the subwords $w_*=\sigma^u(p)|_{\llbracket 1,\ell'_k\rrbracket^2}$ satisfying $\pi_*(w_*)\in\{a'_k,a''_k\}$ (resp. $\pi_*(w_*)=b'_k$).
	
	{\it Item 1.} We first enumerate $I^B = \{u_1,u_2,\ldots,u_H\}$.  Let be $u_h= (u_h^x,u_h^y) \in\mathbb{Z}^2$.  Let
	\begin{displaymath}
	\dis J^B:=\bigcup_{h=1}^H J_h\quad\text{where}\quad J_h:=u_h+\llbracket 1,\ell'_k\rrbracket^2, \ \  \pi_*(\sigma^{u_h}(p))|_{\llbracket1,\ell'_k\rrbracket^2}=b'_k,
	\end{displaymath}
	that is, we are only considering the $J_h$ elements of $J^B(p,\ell_k')$ such that the one-dimensional projection is the block $b_k'$. For each box $J_h$ we divide into $N_k'$ vertical strips of length $\ell_{k-1}$. Formally we have
	\begin{displaymath}
	\dis J_h=\bigcup_{i=1}^{N'_k}J_{h,i}\quad\text{where}\quad J_{h,i}:=u_h+\llbracket 1+(i-1)\ell_{k-1},i\ell_{k-1}\rrbracket\times\llbracket1,\ell'_k \rrbracket.
	\end{displaymath}
	
	We construct a partition of $J^B$  inductively by,
	\begin{gather*}
	J^B=\bigsqcup_{h=1}^H J_h^*,\quad J_1^*= J_1, \ \ \forall\, h \geq2, \ J_h^*:=J_{h}\setminus\left(J_1\cup\cdots\cup J_{h-1}\right).
	\end{gather*}
	Let
	\begin{displaymath}
	\dis K_h^*:=\{v\in J_h^*:p(v)=0\}.
	\end{displaymath}
	It will be enough to show that for every $h \in \llbracket 1,H \rrbracket$
	\begin{equation}
	\label{eq.Card.Kh.B}
	\dis \Card( K_h^*) \leq \big(1-N_{k-1}^{-1}\big)^{-1}  \Card(J_h^*) f_k^B,
	\end{equation}
	By definition of $u_h$, $\tilde w_h = \pi_*(p|_{(u_h+\llbracket 1, \ell'_k \rrbracket^2)})$ is a translate of $b'_k\in\tilde{\mathcal{A}}^{\ell_k'}$,
	\begin{displaymath}
	\dis \forall\,i,j\in \llbracket 1,\ell_k\rrbracket^2, \ \tilde{w}_h(u_h^x+i)=b'_k(i).
	\end{displaymath}
	Since $b'_k$ is made of $N'_k$ subwords of the form $b_{k-1}$, we denote by $\tilde{w}_{h,i}\in\tilde{\mathcal{A}}^{\ell_{k-1}}$, the successive subwords, $\forall\, 1\leq i\leq N'_k,$
	\begin{displaymath}
	\dis \tilde w_{h,i}:=\tilde w_h|_{(u_h^x+\llbracket 1+(i-1)\ell_{k-1},i\ell_{k-1}\rrbracket)} \mbox{ and } \sigma^{u_h^x+(i-1)\ell_{k-1}}(\tilde w _{h,i}) = b_{k-1}.
	\end{displaymath}
	
	We are considering a fixed $h$ and we show that $J_h^*$ is equal to a disjoint union of $N'_k$ vertical strips $(J_{h,i}^*)_{i=1}^{N'_k}$  of the following forms:
	\begin{itemize}
		\item the initial strip $J_{j,1}^*$,
		\begin{displaymath}
		u_h+\left(\llbracket 1+\ell_{k-2},\ell_{k-1}\rrbracket\times\llbracket c_{h,1}, d_{h,1} \rrbracket\right) \subseteq  J_{j,1}^* \subseteq \left( u_h+\llbracket 1,\ell_{k-1} \rrbracket \right) \times \llbracket a_{h,1}, b_{h,1} \rrbracket;
		\end{displaymath}
		\item  the intermediate strips, $J_{h,i}^*$, $1<i<N'_k$,
		\begin{displaymath}
		J_{h,i}^* = u_h+\left(\llbracket(i-1)\ell_{k-1}+1,i\ell_{k-1}\rrbracket\times\llbracket c_{h,i}, d_{h,i}\rrbracket\right); \mbox{and}
		\end{displaymath}
		\item the terminal strip $J_{h,N'_k}^*$,
		\begin{multline*}
		u_h+\left(\llbracket 1+(N_k-1)\ell_{k-1},\ell_k-\ell_{k-2} \rrbracket \times \llbracket c_{h,N_k}, d_{h,N_k} \rrbracket\right) \subseteq \\
		\subseteq  J_{h,N'_k}^* \subseteq 
		u_h+\left(\llbracket 1+(N'_k-1)\ell_{k-1},\ell'_k \rrbracket\times\llbracket a_{h,N'_k}, b_{h,N'_k} \rrbracket\right).
		\end{multline*}
	\end{itemize}
	Here for each $i\in\llbracket1,N_k'\rrbracket$, the values $1 \leq c_{h,i},d_{h,i}\leq \ell_k$ are integers that represent the vertical extent of each strip and it will be  possible that $c_{h,i} < d_{h,i}$ to denote an empty strip $J_{h,i}^*$.
	
	Indeed, for a fixed $1 \leq i \leq N'_k$, we first consider the previous $J_g$, $1 \leq g < h$, that intersects the strip $J_{h,i}$   so that the word $\tilde w_g$ overlaps $\tilde w_h$ on a multiple of $b_{k-1}$ (see item 2 of Lemma \ref{lemma:OverlappingWords}). Then $c_{h,i}$ is the largest upper level of those $J_g \cap J_{h,i}$, more precisely,
	\begin{equation}
	\label{eq.c_hi}
	\dis c_{h,i} = \max_g \big\{ u_g^y + \ell_k' + 1 :  u_g^y \leq u_h^y, \big(u_h^x+(i-1) \ell_{k-1}+\llbracket 1, \ell_{k-1} \rrbracket \big) \subseteq \big( u_g^x + \llbracket 1, \ell_k' \rrbracket \big) \big\}.
	\end{equation}
	and similarly $d_{h,i}$ is the smallest lower level of those $J_g \cap J_{h,i}$, formally we have
	\begin{equation}
	\label{eq.d_hi}
	\dis d_{h,i} =\min_g \big\{ u_g^y + 1 :  u_g^y \geq u_h^y, \big(u_h^x+(i-1) \ell_{k-1}+\llbracket 1, \ell_{k-1} \rrbracket \big) \subseteq \big( u_g^x + \llbracket 1, \ell_k' \rrbracket \big) \big\}.
	\end{equation}
	We have just constructed the intermediate strips $J_{h,i}^*$ for $1<i<N_k$.
	
	\begin{figure}[!h]
		\centering
		\psfrag{1}{$u_g^y+1$}
		\psfrag{2}{$u_g^y+\ell_k$}
		\psfrag{3}{$u_h^y+1$}
		\psfrag{4}{$u_h^y+\ell_k$}
		\psfrag{5}{$u_p^y+1$}
		\psfrag{6}{$u_p^y+\ell_k$}
		\psfrag{h}{$J_h$}
		\psfrag{i}{$J_{h,i}$}
		\psfrag{g2}{$J_p$}
		\psfrag{g}{$J_g$}
		\includegraphics[width=0.7\linewidth]{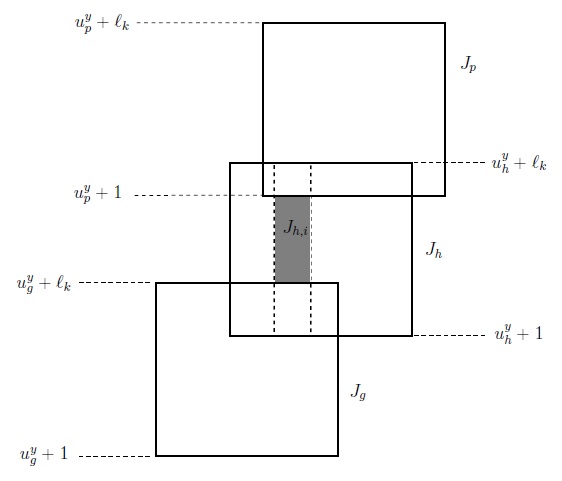}
		\caption{We are representing here the case where there is an intersection but the strip $J_{h,i}$ is not completely covered by the previous squares $J_g$. The squares $J_g$ and $J_p$ are already in the partition, then $J_{h,i}^*$ is only the highlighted gray area.}
		\label{fig:diagram01}
	\end{figure}
	
	We now construct the initial strip (the terminal strip is constructed similarly). We intersect the remaining $J_g$ with $J_{h,1}$. The terminal segment $b_{k-1}^T$ of $\tilde w_g$ overlaps the initial segment $b_{k-1}^I$ of $\tilde w_h$. Thanks to item 1 of Lemma \ref{lemma:OverlappingWords}, as $k-1$ is odd, $b_{k-1}$ has the same structure as $a_k$,  the overlapping can only happen at their end segments of the form $b_{k-2}$. We have just proved that $J_{h,1}^*$ contains a small strip $\big( u_h + \llbracket 1+\ell_{k-2},\ell_{k-1} \rrbracket \big) \times \llbracket c_{h,1}, d_{h,1} \rrbracket$ of base $b_{k-1}^I \setminus b_{k-2}$ and is included in a larger strip $\big( u_h+\llbracket 1,\ell_{k-1} \rrbracket \big) \times \llbracket c_{h,1},d_{h,1} \rrbracket$ of base $b_{k-1}$. For the initial and terminal strip the vertical extension ($\llbracket c_{h,1},d_{h,1}\rrbracket$ and $\llbracket c_{h,{N}_k'},d_{h,{N}_k'}\rrbracket$) of the elements $J_{h,1}^*$ and $J_{h,N_k'}^*$ are defined as in (\ref{eq.c_hi}) and (\ref{eq.d_hi}).
	
	\begin{figure}[!h]
		\centering
		\psfrag{i}{$J_{h,1}$}
		\psfrag{h}{$J_h$}
		\psfrag{g}{$J_g$}
		\includegraphics[width=0.55\linewidth]{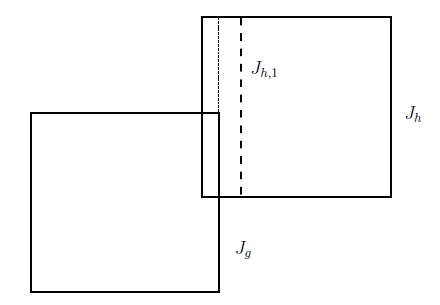}
		\caption{The strip of length $\ell_{k-1}-\ell_{k-2}$ is always contained in $J_{h,1}^*$.}
		\label{fig:diagram03}
	\end{figure}
	
	Let be $K_{h,i}^*:=\{v\in J_{h,i}^*:p_v=0\}$. We show that
	\begin{equation}
	\label{eq.K_hi}
	\forall\, 1 \leq i \leq N_k, \ \Card( K_{h,i}^*) \leq  \big(1-N_{k-1}^{-1}\big)^{-1}  \Card( J_{h,i}^*) f_k^B.
	\end{equation}
	
	For the intermediate strips $ J_{h,i}^*$, where $1 < i < N'_k$, we use the fact that $ J_{h,i}^*$ is a square strip of base $b_{k-1}$, and the fact that the frequency $f_{k-1}^B$ of the symbol $0$ in the word $b_{k-1}$ is identical to the frequency $f_k^B$ of the symbol $0$ in $b_k$. We have,
	\begin{displaymath}
	\Card(K_{h,i}^*) = \ell_{k-1}(d_{h,i}-c_{h,i}+1)  f_k^B =  \Card( J_{h,i}^*) f_k^B.
	\end{displaymath}
	
	For the initial strip $J_{h,1}^*$, we use the fact that $J_{h,1}^*$ resembles largely a square strip of base $b_{k-1}$. We have,
	\begin{align*}
	\Card(\ K_{h,i}^*) &\leq  \ell_{k-1} (d_{h,1}-c_{h,1}+1) f_{k}^B \\
	&\leq \frac{\ell_{k-1}}{\ell_{k-1}-\ell_{k-2}} (\ell_{k-1}-\ell_{k-2}) (d_{h,1}-c_{h,1}+1)f_{k}^B \\
	&\leq \big(1-N_{k-1}^{-1}\big)^{-1} \Card( J_{h,i}^* ) f_k^B.
	\end{align*}
	We have proved (\ref{eq.K_hi}) and by summing over $i \in \llbracket 1, N'_k \rrbracket$ we have proved (\ref{eq.Card.Kh.B}).
	
	{\it Item 2.} As before we will consider $I^A$ (defined in (\ref{eq.IA}), but only consider the translates $u \in \llbracket 0, n-\ell'_k \rrbracket^2$ such that $\pi_*(\sigma^u(p)|_{\llbracket 1,\ell'_k\rrbracket^2})\in\{a'_k,a''_k\}$. If $J_g\cap J_h\neq\emptyset$, the two projected words $\tilde w_g = \pi_*(\sigma^{u_g}(p)|_{\llbracket 1,\ell'_k \rrbracket^2})$ and $\tilde w_h = \pi_*(\sigma^{u_h}(p)|_{\llbracket 1,\ell'_k \rrbracket^2})$ may either coincide in three ways: $\tilde w_g = \tilde w_h$, so $u_g^x=u_h^x$; overlap partially on their markers or overlap on their initial and terminal segments as proved in Lemma~\ref{lemma:OverlappingWordsPrime}.
	
	We redefine again $I^A$ by clustering into a unique rectangle adjacent squares where the overlap occurs in the whole word, that is, we group the squares $J_g$ and $J_h$ that pairwise satisfy $J_g\cap J_h \not= \emptyset$, $u_g^x=u_h^x$, $\tilde w_g = \tilde w_h$,  $|u_g^y-u_h^y|< \ell'_k$. Then, after reindexing $I^A$, one obtains,
	\begin{displaymath}
	\dis J^A = \bigcup_{h=1}^H J_h, \quad J_h =  u_h + \left( \llbracket 1, \ell_{k-1} \rrbracket \times \llbracket 1, d_h \rrbracket \right),
	\end{displaymath}
	where $d_h$ is the final height of each rectangle obtained after the clustering. Thus $w_h^* = \sigma^{u_h}(p)|_{\llbracket 1, \ell'_k \rrbracket \times \llbracket 1, d_h \rrbracket}$ is a vertically aligned pattern whose projection $\tilde w_h = \pi_*(w_h^*)$ is a word of the form  $a'_k$ or $a''_k$, and so that $\tilde w_g$, $\tilde w_h$ never entirely coincide if $J_g \cap J_h \not= \emptyset$.
	
	\begin{figure}[!h]
		\centering
		\psfrag{n}{$n$}
		\includegraphics[width=0.8\linewidth]{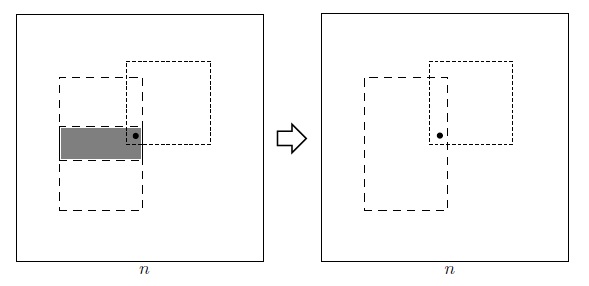}
		\caption{The biggest square has size $n$. On the left side the three squares of size $\ell_k'$ intersect each other and the black dot belongs to each of these squares. The highlighted gray area belongs to both of the vertically aligned squares. After the clustering, on the right side, the two previous dashed squares emerge into one box of size $\ell_k'\times d_h$ and thus the point represented in the figure only belong to two boxes.}
		\label{fig:diagram2408}
	\end{figure}
	
	We now show that an index $v=(v^x,v^y) \in J^A$ may belong to at most two rectangles $J_g$ and $J_h$. Indeed, by construction, $u_g^x \not= u_h^x$, if $v^x$ belongs to two overlapping words of the form $a'_k,a''_k$,  then $v^x$ belongs to either  the intersection of the two markers $1^{(N'_k-1)\ell_{k-1}}$  or the intersection of the terminal segment $a_{k-1}^T$ of $a''_k$ and  the initial segment  $a_{k-1}^I$ of $a'_k$. In both cases described in Lemma~\ref{lemma:OverlappingWordsPrime} we exclude the overlapping of a third word of the form $ a'_k,a''_k$, thus we exclude the fact that $v$ may belong to a third rectangle $J_k$ with $u_k^x \not= u_g^x$ and $u_k^x \not= u_h^x$. Then
	\begin{align*}
	\Card(K^A) &= \sum_{v \in J^A} \mathds{1}_{(p(v)=0)} \\
	&\leq \sum_{h=1}^H \sum_{v \in (u_h+\llbracket 1, \ell'_k \rrbracket \times \llbracket 1, d_h \rrbracket)} \mathds{1}_{(p(v)=0)}  \leq \sum_{h=1}^H  f_{k-1}^A \ell_{k-1} d_h \\
	&\leq \frac{f_{k-1}^A\ell_{k-1}}{\ell'_k} \sum_{h=1}^H  \sum_{v \in J^A} \mathds{1}_{v \in (u_h+\llbracket 1, \ell'_k \rrbracket \times \llbracket 1,d_h \rrbracket)} =  \frac{f_{k-1}^A}{N'_k} \sum_{v \in J^A} \sum_{h=1}^H \mathds{1}_{(v \in J_h)} \\
	&\leq  \frac{2f_{k-1}^A}{N'_k} \Card(J ^A).
	\end{align*}
	
\end{proof}

\section{The new coloring}

Based on our previous construction we define a new coloring for the SFT generated by the Aubrun-Sablik construction. This new subshift will be defined using the alphabet $\mathcal{A}=\mathcal{B}\times \tilde{\mathcal{C}}$, where $\mathcal{B} = \{ 0',0'',1,2 \}$. Consider $\mathcal{A} = \mathcal{B} \times \tilde{\mathcal{C}}$,   $\gamma : \mathcal{A} \to \hat{\mathcal{A}}$ obtained by collapsing the two symbols $0',0''$ to $0$, that is,
\begin{gather*}
\forall\, c \in \mathcal{C}, \ \left\{\begin{array}{ll}
\gamma(0',c) = (0,c), & \gamma(0'',c) = (0,c), \\
\gamma(1,c) = (1,c), &  \gamma(2,c) = (2,c),
\end{array}\right.
\end{gather*}
and
\begin{equation}
\label{eq.Gamma}
\Gamma :  \Sigma^2(\mathcal{A})  \to \Sigma^2(\hat{\mathcal{A}})
\end{equation}
 be the 1-block canonical projection.

Remember that we are denoting $\hat{\mathcal{A}}=\tilde{\mathcal{A}}\times\mathcal{C}$ and $\tilde{\mathcal{A}} = \{0,1,2\}$. Let $\overline{\pi}:\hat{\mathcal{A}}\rightarrow\tilde{\mathcal{A}}$ be the first projection over the alphabet $\tilde{\mathcal{A}}$ as defined in (\ref{eq.overlinepi}). We set $\hat{\Pi} : \Sigma^2(\hat{\mathcal{A}}) \to \Sigma^1(\tilde{\mathcal{A}})$ defined as
\begin{equation}
\label{eq.hatPi}
\left\{
\begin{array}{rcl}
\hat \Pi:\Sigma^2(\hat{\mathcal{A}}) & \to & \Sigma^1(\tilde{\mathcal{A}}) \\
x & \mapsto & y=( \overline{\pi}(x_{(i,0)}) )_{i\in\ZZ}.
\end{array}
\right.
\end{equation}
We will always apply $\hat\Pi$ for configurations that are vertically aligned for the symbols in $\tilde{\mathcal{A}}$ and therefore there is no problem in selecting the zero row with indices $(i,0)$ where $i\in\ZZ$.

Let  $\mathcal{F}$ be the pullback of $\hat{\mathcal{F}}$ by $\Gamma$ and $X$ be the subshift generated by $\mathcal{F}$,
\begin{displaymath}
\dis \mathcal{F}:=\{p\in\mathcal{A}^{\llbracket 1,D\rrbracket^2}:\Gamma(p)\in\hat{\mathcal{F}}\},\quad X:=\Gamma^{-1}(\hat X)= \Sigma^2(\mathcal{A},\mathcal{F}).
\end{displaymath}
Let be
\begin{equation}
\label{eq.piandPi}
\pi = \overline{\pi} \circ \gamma\quad \mbox{and}\quad \Pi = \hat\Pi \circ \Gamma.
\end{equation}

\begin{observation}
	We will also use the projection $\Pi$ as defined before for finite patterns without any distinction. Note that the extended set of forbidden patterns $\mathcal{F}$ forces every locally admissible configuration to be vertically aligned with respect to the initial alphabet $\tilde{\mathcal{A}}$ provided we identify the two  duplicated symbols $0'$ and $0''$.
\end{observation}

We can define the bidimensional subshifts generated by each step of the iteration process. Consider $k$ large enough such that we have $\ell_k> D$ where $D\geq 2$ is defined by the set of forbidden patterns $\mathcal{F}\subset \mathcal{A}^{\llbracket 1,D\rrbracket^2}$. We will denote
\begin{equation}
\label{eq.Lk}
\dis L_k:=\mathcal{L}(X,\ell_k)
\end{equation}
that is, the language of $X$ of size $\ell_k$ as defined in (\ref{languagelk}). We say that a pattern $w$ belongs to $L_k$ if and only if it is globally admissible with respect to $X$. Let $\langle L_k \rangle$ be the corresponding concatenated subshift as defined in Definition~\ref{def.concatenatedsubshift}, that is,
\begin{equation}
\label{eq.concatenated-subshift.Lk}
\dis \langle L_k \rangle := \bigcup_{u \in \llbracket 1, \ell_k \rrbracket^2} \bigcap_{v \in \mathbb{Z}^2} \sigma^{-(u+v \ell_k)} [L_k].
\end{equation}
Note that every pattern in $L_{k+1}$ is obtained by concatenating $N_k^2$ patterns of $L_k$ and the subshifts satisfy $\langle L_{k+1} \rangle \subset \langle L_k \rangle$.


We define two intermediate sub-languages of $\hat X$ of size $\ell_k$ by,
\begin{equation}
\label{eq.AkBk.hat}
\dis \forall\, k \geq0,  \quad
\left\{
\begin{array}{l} 
\hat A_k :=\{ w \in\mathcal{L}(\hat X,\ell_k) : \Pi(w) \in \mathcal{L}(\langle \tilde A_k \rangle, \ell_k) \}, \\
\hat  B_k :=\{ w \in \mathcal{L}(\hat X,\ell_k) : \Pi(w) \in  \mathcal{L}(\langle \tilde B_k \rangle, \ell_k ) \},
\end{array}
\right.
\end{equation}
and two sub-languages of $X$,
\begin{equation}
\label{eq.AkBk}
\forall\, k \geq0,  \quad 
\left\{
\begin{array}{l} 
A_k := \{ w \in \mathcal{A}^{\llbracket 1, \ell_k \rrbracket^2} : \Gamma(w) \in \hat A_k \}, \\
B_k := \{ w \in \mathcal{A}^{\llbracket 1, \ell_k \rrbracket^2} : \Gamma(w) \in \hat B_k \}
\end{array}
\right.
\end{equation}
Every pattern of $A_{k+1}$ (respectively $B_{k+1}$) is made of $N_k^2$ patterns of $A_k$ (respectively $B_k)$. In particular $\langle A_{k+1} \rangle \subseteq \langle A_k \rangle$, $\langle B_{k+1} \rangle \subseteq \langle B_k \rangle$.

We recall two definitions. The reconstruction function is associated to a subshift generated by a set of forbidden words {\color{red} which was also described in \cite{CSS, SS} on a different context}. The relative complexity function is associated to a  shift equivariant extension of a dynamical system. The role of the reconstruction function is clearly put forward in Chazottes-Hochman \cite{CH}. The fact that the subshift of finite type obtained in Aubrun-Sablik \cite{AS} or \cite{CH} has zero entropy is relatively easy to prove. We actually need a more precise estimate of the growth of the complexity. An exponential growth proportional to the boundary of a square  (not proportional to the volume of a square) is enough for instance. This issue seems to be missing in \cite{CH}.

\begin{definition}
\label{def.reconstructionfunction}
Let $\hat{\mathcal{A}}$ be a finite alphabet, $D\geq1$, $\hat{\mathcal{F}} \subseteq \hat{\mathcal{A}}^{\llbracket 1,D \rrbracket^2}$, and $\hat{X}=\Sigma^2(\hat{\mathcal{A}},\hat{\mathcal{F}})$ be the subshift generated by the forbidden patterns $\hat{\mathcal{F}}$, as defined before. We define the {\it reconstruction  function} of the subshift $\hat{X}$ as the function $R^{\hat{X}} : \mathbb{N}^* \to \mathbb{N}^*$ which associates to every $\ell$ the smallest $R$ such that every locally $\hat{\mathcal{F}}$-admissible word in $\mathcal{A}^{\llbracket 1, 2R \rrbracket^2}$ admits a globally $\hat{\mathcal{F}}$-admissible restriction in its central block of length $\ell$.
\end{definition}

We will denote by $M(\hat{\mathcal{F}},R)\subseteq\hat{\mathcal{A}}^{\llbracket 1,R\rrbracket^2}$ the set of all square patterns of size $R$ in $\hat{\mathcal{A}}$ such that no pattern of $\hat{\mathcal{F}}$ appears inside, that is,
\begin{equation}
\label{eq.Mk.hat}
\dis M(\hat{\mathcal{F}},R)  \dis :=\{ w\in\hat{\mathcal{A}}^{\llbracket 1,R\rrbracket^2}:\forall\,p \in\hat{\mathcal{F}},\ \forall\,u\in\llbracket 0,R-D\rrbracket^2,\ p \not\sqsubset\sigma^u(w)\} 
\end{equation}

We will use the reconstruction function for the subshift $\hat{X}$ and the sequence $(R_k')_{k\geq 0}$ defined as
\begin{equation}
\label{eq.Rkprime}
\dis R_k':=R^{\hat{X}}(2\ell_k')=\inf\{R>2\ell_k' :\forall\,w\in M(\hat{\mathcal{F}}, R),\ \exists\, x\in X,\ w|_{Q(2\ell_k',R)}=x|_{Q(2\ell_k',R)}\},
\end{equation}
where $Q(2\ell_k',R)$ is the central block of length $2\ell_k'$, formally defined as
\begin{equation}
\label{eq.Q}
\dis Q(2\ell_k',R):=T(2\ell_k',R)+\llbracket 1,2\ell_k'\rrbracket^2, 
\end{equation}
where $T(2\ell_k',R)=\left(\Big\lfloor\frac{R}{2}-\ell_k'\Big\rfloor,\Big\lfloor\frac{R}{2}-\ell_k'\Big\rfloor \right)\in\ZZ^2$.

\begin{remark}
The reconstruction function exists for every subshift as stated in Proposition~\ref{proposition.Sebastian}, but establishing its growth or computability is not always possible.
\end{remark}

\begin{definition}
\label{def.Cn}
Let $\tilde{X}\subset \Sigma^1(\tilde{\mathcal{A}})$ be the effectively closed subshift described before and $\hat{X}\subset \Sigma^2(\hat{\mathcal{A}})$ be the SFT given by the simulation Theorem~\ref{tAS} that simulates $\tilde{X}$. The {\it relative complexity function} of the simulation is the function $C^{\hat{X}} : \mathbb{N}^* \to \mathbb{N}^*$ defined by
\begin{displaymath}
\dis C^{\hat{X}}(\ell):=\sup_{\tilde{w}\in\mathcal{L}(\tilde{X},\ell)} \Card\big(\{ \hat{w}\in\mathcal{L}(\hat{X},\ell):\hat{\Pi}(\hat{w})=\tilde{w}\}\big).
\end{displaymath}
\end{definition}

The two following propositions give us an idea of the growth of each of the functions (reconstruction and relative complexity). The proofs of these two results are in Appendix~\ref{appendix}. They are very technical proofs that are based on the construction described by Aubrun-Sablik~\cite{AS} and the iteration process described previously.

\begin{proposition}
\label{prop.Sebastian3}
Let $\tilde{X}$ be the one-dimensional effectively closed subshift defined before and $\hat{X}$ be the bidimensional SFT from the Aubrun-Sablik theorem. There is a constant $K >0$ and a polynomial $P(n)$ such that \[ R^{\hat{X}}(n) = P(n)K^n.  \]
\end{proposition}

\begin{proposition}
\label{prop.Sebastian2}
Let $\hat{X}$ be the $\ZZ^2$-SFT in the Aubrun-Sablik construction. There is a constant $K >0$ and a polynomial $P(n)$ such that
\begin{displaymath}
\dis C^{\hat{X}}(n) = P(n) K^n.
\end{displaymath}
\end{proposition}

As a result of these two propositions, we have the next lemma that gives us important bounds for the reconstruction function and the relative complexity function that will be necessary in our final proof.

\begin{lemma}[A priori estimates]
\label{lemma:bounds}
Let $R^{\hat{X}}$ and $C^{\hat{X}}$ be the reconstruction and relative complexity function of the SFT given by Aubrun-Sablik, then
\begin{enumerate}
	\item $\displaystyle \limsup_{n \to + \infty} \frac{1}{n} \ln( C^{\hat{X}}(n) ) < +\infty$,
	\item $\displaystyle \limsup_{n\to+\infty} \frac{1}{n} \ln(R^{\hat{X}}(n)) < +\infty$.
\end{enumerate}
\end{lemma}

The demonstration of these properties is more technical and uses computability theory and Turing machines. These proofs can be found in Appendix~\ref{appendix} but for now on we will assume that they are true.

To simplify the notations, we write
\begin{equation}
\label{eq.Rk.Ck.Qk.Tk.Mkprime}
\begin{array}{ll}
\dis R_k' := R^{\hat{X}}(2\ell_k'),  & \dis C_k' := C^{\hat{X}}(\ell_k'),   \\
&   \\
\dis Q_k' := \mathcal{Q}(2\ell_k' 2,R_k')\subset\ZZ^2,  & \dis T_k' := T(2\ell_k',R_k')\in\ZZ^2,  \\
&   \\
\dis \hat M_k' = M(\hat{\mathcal{F}},R_k')\subset \hat{\mathcal{A}}^{\llbracket1,R_k'\rrbracket^2}, & \dis M_k' = \Gamma^{-1}(\hat M_k')\subset \mathcal{A}^{\llbracket1,R_k'\rrbracket^2}.  \\
\end{array}
\end{equation}

We denote by $[M_k']$ the cylinder generated by the set $M_k'$, which consists of the configurations that are $\mathcal{F}$-locally admissible in $\llbracket 1,R_k'\rrbracket^2$. We compute the topological entropy of patterns that are most of the time (in terms of translations of $\mathbb{Z}^2$) globally admissible with respect to $\hat{\mathcal{F}}$. We naturally  point out the relative complexity function. Notice that the relative entropy is computed using the volume of the square.

\begin{lemma}
	\label{lemma:CoveringComplexity}
	Let  $n > 2 \ell >2$ be some integers, $\epsilon \in(0,1)$ be some real number, and $S \subseteq \llbracket 0,n-2 \ell \rrbracket^2$ be a subset satisfying $\Card(S) \geq n^2(1-\epsilon)$. Let $\hat E$ be the set
	\begin{gather*}
	\hat E := \big\{ w \in \hat{\mathcal{A}}^{\llbracket 1, n \rrbracket^2} : 
	\forall\,  u \in S , \ \sigma^u(w)|_{\llbracket 1, 2 \ell\, \rrbracket^2} \in \mathcal{L}(\hat X, 2\ell)  \big\}.
	\end{gather*}
	Then
	\begin{displaymath}
	\dis \frac{1}{n^2} \ln(\Card(\hat E)) \leq \frac{1}{\ell} \ln(\Card(\tilde{\mathcal{A}})) +  \frac{1}{\ell^2}\ln(C^{\hat{X}}(\ell)) + \epsilon \ln (\Card(\hat{\mathcal{A}})).
	\end{displaymath}
\end{lemma}

\begin{proof}
	Here we consider $n$ as a multiple of $\ell$ in order to simplify the notations since we are interested in the limit when $n\rightarrow+\infty$ there is no problem. We decompose the square $\llbracket 1,n \rrbracket^2$ into a disjoint union of squares of size $\ell$,
	\begin{displaymath}
	\dis \llbracket1,n\rrbracket^2=\bigcup_{v \in \llbracket 0, \frac{n}{\ell} -1 \rrbracket^2} \left( \ell v+ \llbracket 1 ,\ell \rrbracket^2 \right).
	\end{displaymath}
	
	We define the set of indices $v$ that intersect $S$, more precisely, we have
	\begin{displaymath}
	\dis V:=\left\{v\in\Big\llbracket 0,\frac{n}{\ell}-2\Big\rrbracket^2:\left(\ell v+\llbracket 0,\ell -1\rrbracket^2\right)\bigcap S\neq \varnothing \right\}.
	\end{displaymath}
	
	Then for every  $w \in \hat E$, $v \in V$, and $u \in \big( \ell v+ \llbracket 0 ,\ell-1 \rrbracket^2 \big) \bigcap S$, therefore
	\begin{displaymath}
	\dis \left(\ell v + \llbracket 1 + \ell , 2\ell \rrbracket^2 \right) \subseteq \left( u+\llbracket 1, 2\ell \rrbracket^2 \right).
	\end{displaymath}
	Since we are taking $u\in S$ we have that
	\begin{displaymath}
	\dis \sigma^u(w) |_{\llbracket 1,2\ell \rrbracket^2} \in \mathcal{L}(\hat X, 2\ell),
	\end{displaymath}
	and then
	\begin{displaymath}
	\dis \sigma ^{\ell v+(\ell,\ell)}(w)|_{\llbracket 1 , \ell \rrbracket^2} \in \mathcal{L}(\hat X, \ell).
	\end{displaymath}
	
	The restriction of $w$ on every square $\big(\ell v + \llbracket 1 + \ell , 2\ell \rrbracket^2 \big)$ is globally admissible with respect to $\hat{\mathcal{F}}$. Note that these squares are pairwise disjoint and the cardinality of their union  is at least $n^2(1-\epsilon)$, since
	\begin{displaymath}
	\dis \Card\left(\bigcup_{v\in V}\left(\ell v+\llbracket 1+\ell,2\ell\rrbracket^2\right)\right)=\Card\left(\bigcup_{v\in V}\left(\ell v+\llbracket 0 ,\ell- 1\rrbracket^2\right)\right)\geq\Card(S).
	\end{displaymath}
	
	Hence we have proved that $\hat E$ is a subset of the set of patterns $w$ made of independent and disjoint words $(w_v)_{v \in V}$, with $w_v \in \mathcal{L}(\hat X,\ell)$, and  of arbitrary symbols on  $\llbracket 0, n-2\ell \rrbracket^2 \setminus S$ of size at most $\epsilon n^2$. Using the trivial bound $\Card(\mathcal{L}(\tilde X, \ell)) \leq \Card(\tilde{\mathcal{A}})^{\ell}$, we have
	\begin{displaymath}
	\dis \Card(\hat E) \leq \left( \Card(\tilde{\mathcal{A}})^{\ell} \cdot C^{\hat{X}}(\ell) \right)^{\left(n/\ell\right)^2}\cdot \Card(\hat{\mathcal{A}})^{\epsilon n^2}
	\end{displaymath}
	and therefore
	\begin{displaymath}
	\dis \frac{1}{n^2} \ln(\Card(\hat E)) \leq \frac{1}{\ell} \ln(\Card(\tilde{\mathcal{A}})) +  \frac{1}{\ell^2}\ln(C^{\hat{X}}(\ell)) + \epsilon \ln (\Card(\hat{\mathcal{A}})).
	\end{displaymath}
\end{proof}

\chapter{Analysis of the zero-temperature limit}

Consider the full shift $\Sigma^2(\mathcal{A})$ and the finite set of forbidden patterns $\mathcal{F}$ for the subshift $X$. We denote by $F$ the cylinder defined by
\begin{equation}
\label{eq.F}
F:=[\mathcal{F}].
\end{equation}

We consider
\begin{equation}
\label{eq.varphi}
\left\{ 
\begin{array}{rcl}
\varphi : \Sigma^2(\mathcal{A}) & \to & \RR \\
x & \mapsto & \varphi(x)= \mathds{1}_F(x). 
\end{array}
\right.
\end{equation}
We consider $(\beta_k)_{k\geq0}$ as in Definition~\ref{notation}. We denote by $\mathcal{G}(\beta_k\varphi)\subset \mathcal{M}_1(\Sigma^2(\mathcal{A}))$ the set of the equilibrium measures for the potential $\varphi$ at inverse temperature $\beta_k$.

Since our potential $\varphi$ has finite range, it is regular and as in Theorem~\ref{teo.equilibrium=gibbs.invariant} the set of equilibrium measures for $\beta\varphi$ is equal to the set of shift invariant Gibbs measures. Our main goal is to prove that for such a sequence $(\beta_k)_{k\geq 0}$ when $\beta_k\rightarrow+\infty$ any sequence of equilibrium measures $\mu_{\beta_k}$ does not converge when $k\rightarrow+\infty$.

An invariant measure that has support inside $X$ gives zero mass to $F$. We quantify in the following lemma this estimate when the support of the measure is close to $X$, that is inside $\langle L_k \rangle$.

\begin{lemma}
	\label{lemma:BoundFromAboveObservable}
	Let be $k \geq0$ and $\nu$ be  an ergodic probability measure on $\Sigma^2(\mathcal{A})$ such that $\Supp(\nu) \subseteq \langle L_k \rangle$. Then
	\begin{displaymath}
	\nu(F) \leq \frac{2D}{\ell_k}
	\end{displaymath}
\end{lemma}

\begin{proof}
	We assume that $\Supp(\nu) \subseteq \langle L_k \rangle$ where $L_k = \mathcal{L}(X,\ell_k)$ the language of size $\ell_k$ of the subshift $X=\Sigma^2(\mathcal{A},\mathcal{F})$. By Birkhoff's ergodic theorem, for $\nu$-almost every point $x$
	\begin{displaymath}
	\dis \nu(F) =\lim_{n\to+\infty} \frac{ \Card(\{  u \in \Lambda_n : \sigma^u(x) \in F\})}{\Card(\Lambda_n)}.
	\end{displaymath}
	We choose such a point $x \in \langle L_k \rangle$ and $s\in \llbracket 1,\ell_k \rrbracket^2$ such that $\sigma^{s}(x)$ and all its translates  $y_t = \sigma^{s+t \ell_k}(x)$, $t \in \mathbb{Z}^2$,  satisfy $y_t|_{\llbracket 1, \ell_k \rrbracket^2} \in L_k$. We take a sub-sequence $\tilde{\Lambda}_n$ of $\Lambda_n$ with size a multiple of $\ell_k$ defined as
	\begin{displaymath}
	\dis \tilde\Lambda_n := \llbracket -n \ell_k, n \ell_k-1 \rrbracket.
	\end{displaymath}
	
	Note that
	\begin{displaymath}
	\begin{array}{rcl}
	\nu(F)  & = & \dis \lim_{n\to+\infty} \frac{ \Card(\{  u \in \tilde\Lambda_n-s : \sigma^u(x) \in F\})}{\Card(\tilde\Lambda_n)}, \quad y = \sigma^s(x). \\
	& = & \dis \lim_{n\to+\infty} \frac{ \Card(\{  u \in \tilde\Lambda_n : \sigma^u(y) \in F\})}{\Card(\tilde\Lambda_n)}.
	\end{array}
	\end{displaymath}
	By definition of $L_k$ as described in (\ref{eq.Lk}) we have that
	\begin{displaymath}
	x\in\langle L_k\rangle \Rightarrow \forall\, w \in\mathbb{Z}^2, \ \sigma^{s+w \ell_k}(x)=\sigma^{w \ell_k}(y)|_{\llbracket 1,\ell_k \rrbracket^2} \in L_k
	\end{displaymath}
	and
	\begin{displaymath}
	\forall v \in \llbracket 0,\ell_k-D\rrbracket^2,\forall w\in\mathbb{Z}^2, \sigma^{v+w \ell_k}(y)|_{\llbracket 1,D\rrbracket^2}\notin\mathcal{F}.
	\end{displaymath}
	
	Thus for a fixed $w\in\ZZ^2$ we have that the number of possible $v\in \llbracket 0, \ell_k-1 \rrbracket^2$ such that $\sigma(v+w \ell_k)(y)\in\mathcal{F}$ is bounded by
	\begin{displaymath}
	\dis \Card\left( \llbracket 0, \ell_k-1 \rrbracket^2 \setminus \llbracket 0, \ell_k-D \rrbracket ^2\right) \leq  2D \ell_k.
	\end{displaymath}
	
	Therefore if we calculate this bound in the box $\tilde{\Lambda}_n$ we obtain that
	\begin{displaymath}
	\dis \Card(\{  u \in \tilde\Lambda_n : \sigma^u(y) \in F\}) \leq  (2n)^2 2D \ell_k.
	\end{displaymath}
	Since $\Card(\tilde\Lambda_n)=(2n)^2\ell_k^2$, we take the quotient on each side and take the limit with $n\to+\infty$ we obtain $\nu(F) \leq 2D/ \ell_k$.
\end{proof}

We show in the following lemma that an equilibrium measure at low temperature have most of its support close to the largest compact invariant set on which the potential is zero. We quantify more precisely the speed of convergence of the measure on the set of locally admissible patterns as the size of the box goes to infinity.

\begin{lemma}
	\label{lemma:sizeNeighborhoodSubshift}
	For every $k$ and every equilibrium measure $\mu_{\beta_k}$,
	\begin{equation}
	\label{eq.epsilon_k}
	\dis \mu_{\beta_k}\left(\Sigma^2(\mathcal{A}) \setminus [M_k']\right) \leq \frac{R_k'^2}{\beta_k}\ln( \Card(\mathcal{A})) =: \epsilon_k
	\end{equation}
	where $R_k'$ as defined in (\ref{eq.Rkprime}) and $M_k'$ as defined in (\ref{eq.Rk.Ck.Qk.Tk.Mkprime}).
\end{lemma}

\begin{proof}
	If $x\notin [M_k']$, there exists $u\in\llbracket 1,R_k'-D\rrbracket^2$ such that $\sigma^{u}(x)\in F$ and therefore $\varphi(\sigma^{u}(x))=1$. Thus we obtain
	\begin{displaymath}
	\begin{array}{rcl}
	\dis \int \beta_k\varphi d\mu_{\beta_k} & = & \dis \int \beta_k \mathds{1}_F(y) d\mu_{\beta_k}(y) \\
	& \geq & \dis \beta_k R_k'^2 \cdot \mu_{\beta_k}\left(\Sigma\backslash [M_k']\right),
	\end{array}
	\end{displaymath}
	and therefore
	\begin{displaymath}
	\dis -\beta_k \int \varphi d\mu_{\beta_k} \leq -\beta_k R_k'^2\cdot \mu_{\beta_k}\left(\Sigma\backslash [M_k']\right).
	\end{displaymath}
	
	We have that $P(\beta_k\varphi)\geq 0$ and also by the variational principle we obtain
	\begin{displaymath}
	\dis 0\leq P(\beta_k \varphi) =h(\mu_{\beta_k}) -\beta_k \int \varphi d\mu_{\beta_k} \leq h_{top}(\Sigma)-\beta_k R_k'^2\cdot \mu_{\beta_k}\left(\Sigma\backslash [M_k']\right).
	\end{displaymath}
	Since $h_{top}(\Sigma)\leq \ln(\Card(\mathcal{A}))$ we have
	\begin{displaymath}
	\dis \mu_{\beta_k}\left(\Sigma \setminus [M_k']\right) \leq \frac{R_k'^2}{\beta_k}\ln( \Card(\mathcal{A})).
	\end{displaymath}
\end{proof}

The following lemma shows that the topological entropy of the extension depends on the frequency of the symbol $0$ and not on the topological entropy of the base dynamics. By lifting patterns of the 1D subshift  we can only expect an exponential growth proportional to the size of the boundary of a box. As the Aubrun-Sablik extension has zero entropy, we use, as in Chazottes-Hochman~\cite{CH}, the idea of duplicating the zero symbol in the vertical direction of $\mathbb{Z}^2$ in order to obtain an exponential growth proportional to the size of the volume of a box.

\begin{lemma}
	\label{lemma.topological.entropy.bound}
	For every $k\geq0$,
	\begin{displaymath}
	\dis \ln(2) f_k^B \leq h_{top}\big(\langle B_k\rangle\big)
	\end{displaymath}
	A similar estimate holds for $\langle A_k\rangle$.
\end{lemma}

\begin{proof}
	Since $\langle B_k \rangle$ is the concatenated subshift generated by the dictionary $B_k$ as defined in (\ref{def.concatenatedsubshift}), we have
	\begin{displaymath}
	h_{top}(\langle B_k\rangle)=\frac{1}{\ell_k^2}\ln(\Card(B_k)).
	\end{displaymath}
	
	Let be $\tilde{w}\in\mathcal{L}(\langle \tilde{B}_k\rangle,\ell_k)$ such that $f_k(\tilde{w})=f_k^B$. $\tilde{w}$ can be seen as a subword of a concatenation $bb'$ of two words of $\tilde B_k$. By Lemma~\ref{lemma.onedimensional}, $bb'$ is a subword of some configuration $\tilde{x}\in\tilde{X}$.
	
	By our construction there exists $\hat{x}\in\hat{X}$ such that $\tilde{x}=\hat\Pi(\hat{x})$ and $\tilde{w}=\hat\Pi(\hat{w})$ where $\hat{w}=\hat{x}|_{\llbracket 1, \ell_k \rrbracket^2}\in \hat B_k$. Thus we obtain
	\begin{displaymath}
	\dis \Card(B_k)\geq\Card(\{w\in\mathcal{A}^{\llbracket 1,\ell_k\rrbracket^2}:\Gamma(w)=\hat w\})=2^{\ell_k^2f_k(\tilde w)},
	\end{displaymath}
	and therefore
	\begin{displaymath}
	\dis h_{top}(\langle B_k \rangle) \geq \ln(2) f_k^B.
	\end{displaymath}

\end{proof}

The following corollary is our first main estimate of the pressure. We bound from below the pressure by taking the pressure of  the maximal entropy measure of the concatenated subshifts $\langle A_k \rangle$ or $\langle B_k \rangle$. We use here the large scale $\ell_k$ because $\beta_k$ has already been defined using the small scale $\ell'_k$ (see Definition \ref{notation}).

\begin{corollary}
	\label{cor.bound.below.pressure}
	For every $k\geq1$,
	\begin{displaymath}
	P(\beta_k \varphi) \geq \max(f_k^A,f_k^B) \ln(2) -2D \frac{\beta_k}{\ell_k}.
	\end{displaymath}
\end{corollary}

\begin{proof}
	Follows from Lemma~\ref{lemma.topological.entropy.bound} and Lemma~\ref{lemma:BoundFromAboveObservable}.
\end{proof}

Next, we will need to define some notations for standard definitions. Consider $\Sigma^2(\mathcal{A})$ and $\mu$ be a $\sigma$-invariant probability measure. The {\it canonical generating partition of $\Sigma^2(\mathcal{A})$} is the partition
\begin{equation}
\label{eq.G}
\mathcal{G}:=\{[a]_0:a\in\mathcal{A}\}.
\end{equation}
We will denote the {\it base generating partition} as the partition
\begin{displaymath}
\dis \mathcal{G}_\ast := \{ G_0^*, G_1^*, G_2^* \} \quad \mbox{where} \quad G_{\tilde{a}}^*:=\big\{x\in \Sigma^2(\mathcal{A}): \pi(x(0))=\tilde{a} \big\}, \, \tilde{a}\in\tilde{\mathcal{A}}.
\end{displaymath}

For each $k\in\NN$, we will denote by $\mathcal{U}_k$ the partition
\begin{equation}
\label{eq.U}
\dis \mathcal{U}_k := \Big\{\, [M_k'],\, \Sigma^2(\mathcal{A}) \setminus [M_k'] \, \Big\}.
\end{equation}

For each $\epsilon \in (0,1)$ we will define
\begin{equation}
\label{eq.H}
\dis H(\epsilon) := -\epsilon \ln(\epsilon) -(1-\epsilon)\ln(1-\epsilon).
\end{equation}

We introduce a notion of relative entropy which measures  the dynamical entropy of a measure conditioned to be close to $X$.

\begin{definition}
	\label{def.RelativeDynamicalEntropy}
	The {\it relative dynamical entropy of size $k$} of an invariant probability measure $\mu$ is the quantity
	\begin{displaymath}
	\dis h_{rel}(\mu):=\sup_{\mathcal{P}} \left\{\lim_{n\to+\infty}\frac{1}{n^2}H\Big(\,\mathcal{P}^{\llbracket 1,n\rrbracket^2}\mid \mathcal{G}_\ast^{\llbracket 1,n\rrbracket^2} \bigvee\mathcal{U}_k^{\llbracket 0,n-R_k \rrbracket^2},\mu\,\Big) \right\}
	\end{displaymath}
	where the supremum is taken over every finite partition $\mathcal{P}$.
\end{definition}

The relative dynamical entropy is well defined for each $k$ and we can use a version of the Kolmogov-Sinai Theorem (Theorem~\ref{teo.Kolmogorov.Sinai}) for $h_{rel}$. This theorem gives us that the supremum of the definition is attained by a generating partition of the $\sigma$-algebra of $\Sigma^2(\mathcal{A})$.

For each $n\in\NN$ consider the set $V_n\subset\Sigma$ defined as
\begin{displaymath}
\dis V_n:=\left\{x\in\Sigma^2(\mathcal{A}): \pi(x_{(i,j_1)})=\pi(x_{(i,j_2)}), \forall i,j_1,j_2\in\llbracket1,n\rrbracket \right\},
\end{displaymath}
that is, the set of configurations that are vertically aligned over the projection $\pi$ on the alphabet $\tilde{\mathcal{A}}$ in the box $\llbracket1,n\rrbracket^2$. If we consider $\mu_{\beta_k}$ some equilibrium measure at inverse temperature $\beta$ we have that
\begin{displaymath}
\dis \lim_{n \to + \infty} \mu_{\beta_k}\left(\Sigma^2(\mathcal{A})\setminus V_n\right)=0.
\end{displaymath}

Note that
\begin{displaymath}
\begin{array}{rcl}
h_{rel}(\mu_{\beta_k}) & = & \dis \sup_{\mathcal{P}}\left\{\lim_{n\to+\infty}\frac{1}{n^2}H\big(\mathcal{P}^{\llbracket 1,n\rrbracket^2}\mid \mathcal{G}_\ast^{\llbracket 1,n\rrbracket^2} \bigvee\mathcal{U}_k^{\llbracket 0,n-R_k \rrbracket^2},\mu_{\beta_k}\big)\right\} \\
&   &  \\
& = & \dis \lim_{n\to+\infty}\frac{1}{n^2}H\big(\mathcal{G}^{\llbracket 1,n\rrbracket^2}\mid\mathcal{G}_\ast^{\llbracket 1,n\rrbracket^2} \bigvee\mathcal{U}_k^{\llbracket 0,n-R_k \rrbracket^2},\mu_{\beta_k}\big) \\
&   & \\
& = & \dis \lim_{n \to+\infty} \left[\int_{V_n} H(\mathcal{G}^{\llbracket 1,n\rrbracket^2},\mu_x) d\mu_{\beta_k}(x) + \int_{\Sigma^2(\mathcal{A})\setminus V_n} H(\mathcal{G}^{\llbracket 1,n\rrbracket^2},\mu_x) d\mu_{\beta_k}(x)\right], \\
\end{array}
\end{displaymath}
where $(\mu_x)_{x\in\Sigma}$ is a family of conditional measures with respect to $\mathcal{G}_\ast^{\llbracket 1,n\rrbracket^2} \bigvee\mathcal{U}_k^{\llbracket 0,n-R_k \rrbracket^2}$.

Hence if we consider a configuration $x\in V_n$, the number of possible configurations in $\mathcal{G}^{\llbracket1,n\rrbracket^2}$ is bounded by $\Card(\mathcal{A})^n$. Therefore
\begin{displaymath}
H(\mathcal{G}^{\llbracket 1,n\rrbracket^2},\mu_x) \leq n\cdot \ln(\Card(\mathcal{A})),
\end{displaymath}
and then $h_{rel}(\mu_{\beta_k})<+\infty$.

The next lemma gives us an upper bound of the entropy of the equilibrium measure $\mu_{\beta_k}$ for each $k\in\NN$.

\begin{lemma}
	\label{lemma.upper.bound.entropy}
	For every $k$ and every equilibrium measure $\mu_{\beta_k}$
	\begin{displaymath}
	\dis h(\mu_{\beta_k}) \leq  h_{rel}(\mu_{\beta_k}) + \left(\frac{8}{R_k'}+ \epsilon_k\right) \ln(\Card(\tilde{\mathcal{A}})) +  H(\epsilon_k).
	\end{displaymath}
\end{lemma}

\begin{proof}
	We take the supremum over all finite partitions of $\Sigma^2(\mathcal{A})$, so we can always consider that we are taking $\mathcal{P}\succeq \mathcal{G}_\ast$ and $\mathcal{P}\succeq \mathcal{U}_k$ and therefore $\mathcal{P}\succeq \tilde{\mathcal{G}}\bigvee \mathcal{U}_k$. For consequence we obtain
	\begin{displaymath}
	\dis \mathcal{P}^{\llbracket 1,n \rrbracket^2} \succeq \mathcal{G}_\ast^{\llbracket 1,n \rrbracket^2} \bigvee \mathcal{U}_k^{\llbracket 0,n-R_k' \rrbracket^2}.
	\end{displaymath}
	
	By the definition of relative entropy
	\begin{displaymath}
	\begin{array}{rcl}
	\dis H\Big(\mathcal{P}^{\llbracket 1,n \rrbracket^2},\mu_{\beta_k}\Big) & = & \dis H\Big(\mathcal{P}^{\llbracket 1,n \rrbracket^2}\bigvee \mathcal{G}_\ast^{\llbracket 1,n \rrbracket^2} \bigvee \mathcal{U}_k^{\llbracket 0,n-R_k' \rrbracket^2},\mu_{\beta_k}\Big)\\
	&   & \\
	& = & \dis H\Big(\mathcal{P}^{\llbracket 1,n\rrbracket^2}\mid\tilde{\mathcal{G}}^{\llbracket 1,n \rrbracket^2}\bigvee\mathcal{U}_k^{\llbracket 0,n-R_k' \rrbracket^2},\mu_{\beta_k}\Big) + \\
	&   &  \dis + H \Big( \mathcal{G}_\ast^{\llbracket 1,n \rrbracket^2}\bigvee\mathcal{U}_k^{\llbracket 0,n-R_k' \rrbracket^2},\mu_{\beta_k}\Big) \\
	&   & \\
	& = & \dis H\Big(\mathcal{P}^{\llbracket 1,n\rrbracket^2}\mid \mathcal{G}_\ast^{\llbracket 1,n \rrbracket^2}\bigvee\mathcal{U}_k^{\llbracket 0,n-R_k' \rrbracket^2},\mu_{\beta_k}\Big) + \\
	&   & \dis + H \Big( \mathcal{G}_\ast^{\llbracket 1,n \rrbracket^2} \mid \mathcal{U}_k^{\llbracket 0,n-R_k' \rrbracket^2}, \mu_{\beta_k} \Big)  + H \Big( \mathcal{U}_k^{\llbracket 0,n-R_k' \rrbracket^2}, \mu_{\beta_k} \Big).
	\end{array}
	\end{displaymath}
	
	The first term of the right hand side is computed using the relative dynamical entropy (Definition \ref{def.RelativeDynamicalEntropy}). The third term is bounded from above using Lemma \ref{lemma:sizeNeighborhoodSubshift} (provided $\epsilon_k < e^{-1}$),
	\begin{displaymath}
	\begin{array}{rcl}
	\dis H \big( \mathcal{U}_k^{\llbracket 0,n-R_k' \rrbracket^2}, \mu_{\beta_k} \big) & = &\dis \sum_{P\in \mathcal{U}_k^{\llbracket 0,n-R_k' \rrbracket^2}}-\mu_{\beta_k}(P)\ln(\mu_{\beta_k}(P)) \\
	& \leq & \dis  n^2 H(\mathcal{U}_k, \mu_{\beta_k}) \\
	& \leq & n^2 H(\epsilon_k).
	\end{array}
	\end{displaymath}
	
	We now compute the term in the middle. We choose $\epsilon'_k > \epsilon_k$ and define
	\begin{displaymath}
	\dis U_n:=\Big\{ x \in \Sigma^2(\mathcal{A}):\Card\left\{u\in\llbracket0,n-R_k'\rrbracket^2:\sigma^u(x)\in[M_k']\right\}\geq n^2(1-\epsilon'_k )\Big\}.
	\end{displaymath}
	
	By Birkhoff ergodic theorem we have that
	\begin{displaymath}
	\lim_{n \to +\infty} \mu_{\beta_k}(U_n) = 1.
	\end{displaymath}
	
	Note that
	\begin{displaymath}
	\begin{array}{rcl}
	\dis H\left(\mathcal{G}_\ast^{\llbracket 1,n\rrbracket^2}\mid\mathcal{U}_k^{\llbracket 0,n-R_k\rrbracket^2},\mu_{\beta_k}\right) & = & \dis \int H\Big(\mathcal{G}_\ast^{\llbracket1,n\rrbracket^2},\mu_x \Big) d\mu_{\beta_k}(x) \\
	&   &  \\
	& = & \dis \int_{U_n} H\Big(\mathcal{G}_\ast^{\llbracket1,n\rrbracket^2},\mu_x \Big) d\mu_{\beta_k}(x) + \\
	&   & \dis +\int_{\Sigma^2(\mathcal{A})\setminus U_n} H\Big(\mathcal{G}_\ast^{\llbracket1,n\rrbracket^2},\mu_x \Big) d\mu_{\beta_k}(x) \\
	&   & \\
	& \leq & \dis \int_{U_n} H\Big(\mathcal{G}_\ast^{\llbracket 1,n \rrbracket^2},\mu_x\Big) d\mu_{\beta_k}(x) + \\
	&   & \dis +n^2 \mu_{\beta_k}\left(\Sigma^2(\mathcal{A}) \setminus U_n\right) \ln(\Card(\mathcal{A})), \\
	\end{array}
	\end{displaymath}
	and therefore
	\begin{displaymath}
	\dis \limsup_{n\to+\infty}\frac{1}{n^2}H \Big(\mathcal{G}_\ast^{\llbracket 1,n \rrbracket^2}\mid\mathcal{U}_k^{\llbracket 0,n-R_k' \rrbracket^2}, \mu_{\beta_k}\Big)\leq\limsup_{n\to+\infty}\int_{U_n}\frac{1}{n^2}H\Big(\mathcal{G}_\ast^{\llbracket 1,n \rrbracket^2},\mu_x\Big)d\mu_{\beta_k}(x),
	\end{displaymath} 
	where $(\mu_x)_{x \in \Sigma}$ is the family of conditional measures with respect to $\mathcal{U}_k^{\llbracket 0,n-R_k' \rrbracket}$.

	Now consider a fixed $x \in U_n$. We compute the cardinality of elements in $\mathcal{G}_\ast^{\llbracket 1,n \rrbracket^2}$ that are compatible with the constraint
	\begin{displaymath}
	\dis \Card\{u\in\llbracket 0,n-R_k'\rrbracket^2:\sigma^u(x)\in [M_k']\}\geq n^2 (1-\epsilon'_k ).
	\end{displaymath}
	Note that
	\begin{displaymath}
	\dis \mathcal{G}_\ast^{\llbracket 1,n \rrbracket^2} = \bigvee_{u\in\llbracket1,n\rrbracket^2}\sigma^{-u}\left(\mathcal{G}_\ast\right)
	\end{displaymath}
	where $\mathcal{G}_\ast = \{G_0^\ast, G_1^\ast, G_2^\ast\}$ and here we refer to the elements of this partition as patterns defined in $\tilde{\mathcal{A}}^{\llbracket1,n\rrbracket^2}$ because there is a unique equivalence between these objects.

	We denote by $I(x)=I\subset \llbracket0,n-R_k'\rrbracket^2$ such that 
	\begin{displaymath}
	\dis I :=\left\{u\in \llbracket 0,n-R_k'\rrbracket^2:\sigma^u(x)\in[M_k']\right\}.
	\end{displaymath}
	
	Since $x\in U_n$, then
	\begin{displaymath}
	\frac{\Card(I)}{n^2}\geq 1-\epsilon_k'.
	\end{displaymath}
	
	Let $J \subseteq I$ be a maximal subset satisfying for every  $ u,v \in J$, 
	\begin{displaymath}
	\dis \|u-v\|_\infty\geq\frac{1}{2}R_k'.
	\end{displaymath}
	
	For every $u\in J$, consider
	\begin{displaymath}
	\dis I_u:=\{v\in I:\|u-v\|_\infty<\frac{1}{2}R_k'\}
	\end{displaymath}
	Then $I=\bigcup_{u \in J} I_u$. We first observe that the sets $\Big(u+\big\llbracket 1,\lceil R_k'/2\rceil\big\rrbracket^2\Big)_{u\in J}$ are pairwise disjoint. Then
	\begin{displaymath}
	\dis \Card(J) \leq \frac{4n^2}{R_k'^2}.
	\end{displaymath}
	
	We also observe that for every $v_1,v_2\in I_u$, $\|v_1-v_2\|_\infty<R_k'$ and
	\begin{displaymath}
	\dis \Big(v_1+\llbracket 1,R_k' \rrbracket^2 \Big) \bigcap \Big(v_2 + \llbracket 1,R_k' \rrbracket^2 \Big) \neq \emptyset.
	\end{displaymath}
	For each $u\in I$ let be 
	\begin{displaymath}
	\dis K_u: = \bigcup_{v \in I_u} \big(v+\llbracket 1,R_k' \rrbracket^2 \big)\subset \llbracket1,n\rrbracket^2.
	\end{displaymath}
	
	For $v\in I_u$, we have that 
	\begin{displaymath}
	\dis x|_{v+ \llbracket 1,R_k' \rrbracket^2}\in [M_k']
	\end{displaymath}
	and therefore this pattern is locally $\mathcal{F}$-admissible and also satisfies the constraint that all the $\tilde{\mathcal{A}}$-symbols are vertically aligned in $v+ \llbracket 1,R_k \rrbracket^2$ and also in $K_u$.
	
	The width of $K_u$ is less than $2R_k'$, so the cardinality of possible patterns $p\in\tilde{\mathcal{A}}^{K_u}$ satisfying the constraint of vertically aligning of $\tilde{\mathcal{A}}$-symbols is bounded by $\Card(\tilde{\mathcal{A}})^{2R_k'}$. The cardinality of possible patterns over the support $\dis \bigcup_{u\in J}K_u$ is thus bounded by 
	\begin{displaymath}
	\dis \Big(\Card(\tilde{\mathcal{A}})^{2R_k'}\Big)^{4n^2/R_k'^2}= \exp\left(\left[2R_k' \cdot \frac{4n^2}{R_k'^2}\right] \ln (\Card(\tilde{\mathcal{A}}))\right) = \exp\left( \frac{8n^2}{R_k'} \ln (\Card(\tilde{\mathcal{A}}))\right).
	\end{displaymath}
	Since $\dis \bigcup_{u\in J} K_u$ covers $I$, the cardinality of the set of possible patterns over the support $\dis \llbracket 1,n \rrbracket^2\setminus\bigcup_{u\in J}K_u$ is bounded by $\Card(\tilde{\mathcal{A}})^{n^2\epsilon_k'}$. We have proved that, for every $x \in U_n$,
	\begin{displaymath}
	H\Big(\mathcal{G}_\ast^{\llbracket1,n\rrbracket^2},\mu_x\Big)\leq\Big(2R_k' \cdot \frac{4n^2}{R_k'^2}+n^2\epsilon_k' \Big)\ln(\Card(\tilde{\mathcal{A}})).
	\end{displaymath}
	
	We conclude by letting $n\to+\infty$ and $\epsilon_k' \to \epsilon_k$.
\end{proof}

The following lemma is the second main estimate on the pressure. We bound from above the pressure assuming that the generic patterns of the  equilibrium measure exhibit a positive frequency  (here $1/4$) of the symbol 1. Since the potential is non-negative, it is enough to bound from  above the pressure by the entropy of $\mu_{\beta_k}$.

We denote as $\overline{\Pi}:\Sigma^2(\hat{\mathcal{A}})\to \Sigma^2(\mathcal{\tilde{A}})$ the projection on the first coordinate. Using (\ref{eq.Gamma}) we set
\begin{equation}
\Pi_\ast=\Gamma\circ\overline{\Pi}:\Sigma^2(\mathcal{A})\to \Sigma^2(\mathcal{\tilde{A}})
\end{equation}
the projection on the bidimensional configurations over the alphabet $\tilde{\mathcal{A}}$.

\begin{lemma}
	\label{lemma.bounds}
	Let $k\geq 2$ be an  integer and $\mu_{\beta_k}$ be any  equilibrium measure. Then
	\begin{enumerate}
		\item $\displaystyle \mu_{\beta_k}([0]) \leq  \frac{2}{N'_k}  f_{k-1}^A + (1-N_{k-1}^{-1})^{-1} f^B_{k-1} + \epsilon_k$,
		
		\item if $k$ is even and $\mu_{\beta_k}( [1]) > \frac{1}{4}$,
		\begin{multline*}
		h_{rel}(\mu_{\beta_k})  \leq 
		\Big( \frac{2}{N'_k}  f_{k-1}^A + (1-N_{k-1}^{-1})^{-1} \Big( \frac{3}{4} + \epsilon_k \Big) f^B_{k-1}\Big) \ln(2) \\
		+ \frac{1}{\ell'_k} \ln(\Card(\tilde{\mathcal{A}})) +  \frac{1}{\ell'_k{}^2}\ln(C_k') + \epsilon_k \ln (2\Card(\hat{\mathcal{A}})),
		\end{multline*}
		
		\item if $k$ is odd and $\mu_{\beta_k}( [2]) > \frac{1}{4}$, the previous estimate is valid with $f_{k-1}^A$ and $f_{k-1}^B$ permuted,
	\end{enumerate}
where for each $\tilde{a}\in\tilde{\mathcal{A}}$, $\mu_{\beta_k}([\tilde{a}])$ is the measure $\mu_{\beta_k}$ of the cylinder $\Pi_\ast^{-1}([\tilde a]_{(0,0)})=:\Pi_\ast^{-1}[\tilde a]\subset\Sigma^2(\mathcal{A})$.
\end{lemma}

\begin{proof}
	Let be $\Pi_\ast:\Sigma^2(\mathcal{A})\to\Sigma^2(\tilde{\mathcal{A}})$ the projection over the first letter on the $\tilde{A}$-alphabet. By Birkhoff ergodic theorem and Lemma \ref{lemma:sizeNeighborhoodSubshift}, for almost every $x\in\Sigma^2(\mathcal{A})$,
	\begin{displaymath}
	\lim_{n\to+\infty} \frac{1}{n^2} \Card \big(\big\{ u \in \llbracket 0,n-R'_k \rrbracket^2 : \sigma^u(x) \in [M'_k] \big\} \big) = \mu_{\beta_k}([M'_k])
	\end{displaymath}
	and
	\begin{displaymath}
	\lim_{n \to+\infty} \frac{1}{n^2} \Card \big( \big\{ u \in \llbracket 1,n \rrbracket^2 : \pi(x(u))=\tilde a \big\} \big) = \mu_{\beta_k}([\tilde a]), \quad \forall\, \tilde a \in \tilde{\mathcal{A}}.
	\end{displaymath}
	Here we are denoting $\mu_{\beta_k}([\tilde a])$ for the measure $\mu_{\beta_k}$ of the cylinder $\Pi_\ast^{-1}[\tilde a]$, but we suppress the pre-image of the projection $\pi$ to simplify our notation.
	
	We choose $n> R'_k$. An element  of the partition $\mathcal{G}_*^{\llbracket 1, n \rrbracket^2} \bigvee \mathcal{U}^{\llbracket 0,n-R'_k \llbracket^2}$ is of the form $G_p^* \cap U_S$ where $p \in \tilde{\mathcal{A}}^{\llbracket 1, n \rrbracket^2}$ is a pattern and $S \subseteq \llbracket 0,n-R'_k\rrbracket^2$ is a subset, that satisfies
	\[
	U_S := \left\{ x \in \Sigma^2(\mathcal{A}) : \forall\,  u \in S, \ \sigma^u(x) \in [M'_k], \ \forall\, u \in \llbracket 0,n-R'_k\rrbracket^2 \setminus S, \ \sigma^u(x) \not\in [M'_k] \right\},
	\]
	\[
	\dis G_p^* := \left\{ x \in \Sigma^2(\mathcal{A}) : \left(\Pi_\ast(x)\right)|_{\llbracket1,n\rrbracket^2} = p \right\}.
	\]
	
	We set $\epsilon>\epsilon_k$ and $\eta < \mu_{\beta_k}([0])$. 
	By the Lemma~\ref{lemma:sizeNeighborhoodSubshift} we have that $\mu_{\beta_k} \big(\Sigma^2(\mathcal{A}) \setminus [M'_k] \big) \leq \epsilon_k$ and then
	\begin{displaymath}
	\dis \lim_{n\to+\infty} \mu_{\beta_k} \left( \bigcup\nolimits_S \left\{ U_S :  \Card(S) \geq n^2(1-\epsilon) \right\} \right) = 1.
	\end{displaymath}
	For $n$ large enough, we choose $S\subseteq\llbracket 0,n-R_k'\rrbracket^2$ such that $U_S \not= \emptyset$ and $\Card(S)\geq n^2(1-\epsilon)$. By definition of $M'_k$ and $T'_k$, if $x\in U_S$, then for every $u\in S$, $\sigma^u(x)|_{\llbracket 1,R'_k \rrbracket^2}$ is a locally admissible pattern with respect to $\mathcal{F}$ and
	\begin{displaymath}
	\dis \sigma^{u+T'_k}(x)|_{\llbracket 1,2\ell'_k\rrbracket^2} \in \mathcal{L}(X,2\ell'_k).
	\end{displaymath}
	
	Define for every $n>R_k'$ and every pattern $p \in \tilde{\mathcal{A}}^{\llbracket 1, n \rrbracket^2}$ the set
	\begin{displaymath}
	K_n(p):=\{u\in\llbracket 1,n\rrbracket^2:p(u)=0\}.
	\end{displaymath}
	As we are considering $\mu_{\beta_k}([0])>\eta$
	\begin{displaymath}
	\dis \lim_{n\to+\infty}\mu_{\beta_k}\left(\bigcup\nolimits_p\left\{G_p^*:\Card(K_n(p))>n^2\cdot \eta\right\} \right)=1.
	\end{displaymath}
	We may choose $p$ such that $U_S \cap G_p^* \not= \emptyset$ and $\Card(K_p) > n^2 \eta$. Using the objects as defined in (\ref{eq.I}), (\ref{eq.IA}), (\ref{eq.JA}) and (\ref{eq.KA}), one obtains
	\begin{displaymath}
	\dis T'_k+S \subseteq I(p,\ell'_k) \quad \mbox{and} \quad \tau'_k+I(p,\ell'_k)\subseteq J^A(p,\ell'_k)\bigsqcup J^B(p,\ell'_k)=:J^A\bigsqcup J^B,
	\end{displaymath}
	therefore by our choice of $S$ we obtain
	\begin{equation}
	\label{eq.Lemma15.1}
	n^2(1-\epsilon)\leq \Card(S) = \Card(\tau'_k+T'_k+S) \leq \Card \big(J^A \bigsqcup J^B \big) \leq n^2.
	\end{equation}
	
	Besides that we have
	\begin{displaymath}
	\dis n^2\eta \leq \Card(K_n(p)) \leq \Card(K^A \bigsqcup K^B)+n^2 \epsilon
	\end{displaymath}
	and by the Lemma~\ref{lemma:FrequencyOf0} we have
	\begin{displaymath}
	\dis \Card(K_n(p)) \leq  \frac{2}{N'_k}\Card(J^A)f_{k-1}^A +  \big(1-N_{k-1}^{-1}\big)^{-1} \Card(J^B) f_{k-1}^B + n^2 \epsilon.
	\end{displaymath}
	
	We divide each term by $n^2$ and take the limit with $n\to+\infty$, $\epsilon \to \epsilon_k$, and $\eta \to \mu_{\beta_k}([0])$. Thus we proved the first item of this lemma.

	We now assume that $k$ is even and $\mu_{\beta_k}([1]) > \frac{1}{4}$. We choose $p\in\tilde{\mathcal{A}}^{\llbracket1,n\rrbracket^2}$ such that $G_p^* \cap U_S \not= \emptyset$ and
	\begin{equation}
	\label{eq.Lemma15.2}
	\dis \Card \left( \left\{ u \in \llbracket 1,n\rrbracket^2 : p(u)=1 \right\} \right) > \frac{n^2}{4}.
	\end{equation}
	Let be $x \in G_p^* \cap U_S$ and $(\mu_x)_{x \in\Sigma}$ be the family of conditional measures with respect to the partition $\mathcal{G}_*^{\llbracket 1, n \rrbracket^2} \bigvee \mathcal{U}^{\llbracket 0,n-R_k \llbracket^2}$. We use the trivial upper bound of the entropy, so
	\begin{equation}
	\label{eq.Lemma15.6}
	\dis H( \mathcal{G}^{\llbracket 1,n \rrbracket^2}, \mu_x) \leq \ln(\Card(E_{p,S}))
	\end{equation}
	where
	\begin{displaymath}
	\dis E_{p,S} := \big\{ w \in \mathcal{A}^{\llbracket 1,n \rrbracket^2} :  \pi(w)=p \ \text{and} \ \ 
	\forall u \in S, \ \sigma^{u+T'_k}(w)|_{\llbracket 1,2\ell'_k\rrbracket^2} \in \mathcal{L}(X,2\ell'_k) \big\}.
	\end{displaymath}
	Also consider
	\begin{displaymath}
	\dis \hat E_{p,S} := \Gamma(E_{p,S}). 
	\end{displaymath}
	
	Note that every word in $E_{p,S}$ is obtained from a word in $\hat E_{p,S}$ by duplicating twice a symbol $0$ and by Lemma~\ref{lemma:CoveringComplexity} we can conclude that
	\begin{gather*}
	\ln(\Card(E_{p,S}) ) \leq \ln( \Card(\hat E_{p,S}) ) + \Card( K_p) \ln(2) \mbox{ and } \\
	\frac{1}{n^2} \ln(\Card(\hat E_{p,S})) \leq \frac{1}{\ell'_k} \ln(\Card(\tilde{\mathcal{A}})) +  \frac{1}{\ell'_k{}^2}\ln(C'_k) + \epsilon_k \ln (\Card(\hat{\mathcal{A}})),
	\end{gather*}
	thus
	\begin{equation}
	\label{eq.Lemma15.3}
	\begin{array}{rcl}
	\dis \frac{1}{n^2} \ln(\Card(E_{p,S})) & \leq &  \dis \Big( \frac{2}{N'_k}\Card(J^A) f_{k-1}^A + (1-N_{k-1}^{-1})^{-1} \Card(J^B) f^B_{k-1}+ n^2 \epsilon_k \Big) \frac{\ln(2)}{n^2} + \\
	&   & \dis + \frac{1}{\ell'_k} \ln(\Card(\tilde{\mathcal{A}})) +  \frac{1}{\ell'_k{}^2}\ln(C_k') + \epsilon_k \ln (\Card(\hat{\mathcal{A}})).
	\end{array}
	\end{equation}
	
	The symbol $1$ does not appear in $J^B=J^B(p,\ell'_k)$, so we can affirm
	\begin{displaymath}
	\dis \big\{ u \in \llbracket 1,n\rrbracket^2 : p(u)=1 \big\} \subset J^A  \bigsqcup \left( \llbracket 1,n \rrbracket^2 \setminus (J^A\bigsqcup J^B) \right).
	\end{displaymath}
	Since we are assuming (\ref{eq.Lemma15.2}) and using (\ref{eq.Lemma15.1}) we obtain that
	\begin{equation}
	\label{eq.Lemma15.4}
	\dis \Card(J^A) \geq n^2 \Big( \frac{1}{4} - \epsilon_k \Big) \quad\mbox{and}\quad \Card(J^B) \leq n^2 \Big( \frac{3}{4}+\epsilon_k \Big).
	\end{equation}
	
	By replacing the upper bound for $\Card(J^B)$ given in (\ref{eq.Lemma15.4}) and $\Card(J^A)\leq n^2$ in (\ref{eq.Lemma15.3}) we obtain that
	\begin{equation}
	\label{eq.Lemma15.5}
	\begin{array}{rcl}
	\dis \frac{1}{n^2} \ln(\Card(E_{p,S})) & \leq &	\dis \left( \frac{2}{N'_k}f_{k-1}^A + (1-N_{k-1}^{-1})^{-1}\left( \frac{3}{4}+\epsilon_k \right) f^B_{k-1}+ \epsilon_k \right)\ln(2) + \\
	   &   & \dis  + \frac{1}{\ell'_k} \ln(\Card(\tilde{\mathcal{A}})) +  \frac{1}{\ell'_k{}^2}\ln(C_k') + \epsilon_k \ln (\Card(\hat{\mathcal{A}})). \\
	\end{array}
	\end{equation}
	By integrating with respect to $\mu_{\beta_k}$ in both sides and taking the limit when $n\to+\infty$ we obtain item 2 of this lemma. Item 3 has an analogous proof.
\end{proof}

\begin{theorem}
\label{theorem.final}
Let $X=\Sigma^2(\mathcal{A},\mathcal{F})$ be the bidimensional SFT described before, which is generated by the finite set of forbidden patterns $\mathcal{F}\subset \mathcal{A}^{\llbracket 1,D\rrbracket^2}$ defined over the alphabet $\mathcal{A}$. Let $F$ be the cylinder generated by the set $\mathcal{F}$ as described in (\ref{eq.F}) and $\varphi:\Sigma^2(\mathcal{A})\to\RR$ be the locally constant potential defined as $\varphi = \mathds{1}_F$. Let $X_A$, respectively $X_B$, be the compact sets of configurations in $X$ that have only the symbol $1$, respectively $2$, in terms of the $\tilde{\mathcal{A}}$ alphabet, therefore, $X_A$ and $X_B$ are two disjoint invariant compact sets. Then there exists a sequence of inverse temperatures $(\beta_k)_{k\geq0}$ such that  for every  equilibrium measure $\mu_{\beta_k}$  associated to the potential $\beta_k\varphi$, the support of every  accumulation point $\mu_\infty^A$ or $\mu_\infty^B$, of the subsequence $(\mu_{\beta_{2k+1}})_{k\geq0}$ or $(\mu_{\beta_{2k}})_{k\geq0}$, is included in $X_A$ or $X_B$.
\end{theorem}

\begin{proof}
	We consider $X=\Sigma^2(\mathcal{A},\mathcal{F})$ the SFT as described before, $F$ as in (\ref{eq.F}) and $\varphi = \mathds{1}_F$. We denote by $\mu_{\beta_k}$ an equilibrium measure at inverse temperature $\beta_k$. We will prove that as $\beta_k \to+\infty$ the sequence $(\mu_{\beta_k})_{k >0}$ does not converge.
	
	Assume $k$ is an even number and $\mu_{\beta_k}([1]) > \frac{1}{4}$. Let $\mu_k^B$ be the measure of maximal entropy of the subshift $\langle L_k \rangle$.  On the one hand, from Corollary~\ref{cor.bound.below.pressure} we have that
	\begin{displaymath}
	\dis P(\beta_k \varphi) \geq h(\mu_k^B) -\int \! \beta_k \varphi \, d\mu_k^B \geq  f_k^B \ln(2) -2D \frac{\beta_k}{\ell_k}.
	\end{displaymath}
	
	By the item 3 of Definition~\ref{notation} we have that
	\begin{displaymath}
	\dis N_k\geq N_k'\cdot \frac{k\beta_k}{N_k' \rho_{k-1}^B}=\frac{k\beta_k}{\ell_k'f_{k-1}^B} \Rightarrow \frac{\beta_k}{\ell_k} \leq \frac{\beta_k}{\ell_k'} \leq \frac{1}{k}f_{k-1}^B.
	\end{displaymath}
	Since $k$ is even, $f_{k-1}^B = f_k^B$, one obtains,
	\begin{displaymath}
	\dis f_k^B\ln(2)-2D\frac{\beta_k}{\ell_k} \geq \frac{k \beta_k}{\ell_k'}-2D\frac{\beta_k}{\ell_k} > 0 \Rightarrow 2D\frac{\beta_k}{\ell_k}\leq f_k^B\ln(2).
	\end{displaymath}
	
	On the other hand
	\begin{multline*}
	P(\beta_k \varphi) \leq  
	\Big( \frac{2}{N'_k}  f_{k-1}^A + (1-N_{k-1}^{-1})^{-1} \Big( \frac{3}{4} + \epsilon_k \Big) f^B_{k-1}\Big) \ln(2) \\
	+ \frac{1}{\ell'_k} \ln(\Card(\tilde{\mathcal{A}}))  +  \frac{1}{\ell'_k{}^2}\ln(C(\hat X, \ell'_k)) + \epsilon_k \ln (2\Card(\hat{\mathcal{A}})) \\
	+ \Big( \frac{8}{R(\hat X, \ell'_k)}+ \epsilon_k  \Big) \ln(\Card(\tilde{\mathcal{A}}))+  H(\epsilon_k).
	\end{multline*}
	
	We have that
	\begin{displaymath}
	\dis \epsilon_k \ll f_{k-1}^B \ \ \text{and} \ \ H(\epsilon_k) \ll f_{k-1}^B.
	\end{displaymath}
	Indeed, from item 2 of Lemma~\ref{lemma:bounds} shows that there exist constants $\Xi,\xi$ such that
	\[
	\forall\, k\geq 1, \ R_k' \leq \Xi 2^{\xi \ell'_k}.
	\]
	Recalling the definition of $\epsilon_k=\frac{{R_k'}^2}{\beta_k}\ln( \Card(\mathcal{A}))$ given in (\ref{eq.epsilon_k}) and using item 2 of Definition~\ref{notation}, one gets,
	\begin{displaymath}
	\dis \frac{\epsilon_k}{(f_{k-1}^B)^2} \leq \frac{\epsilon_k \beta_k}{2^{k\ell'_k}}= \frac{{R_k'}^2\ln(\Card(\mathcal{A}))}{2^{k\ell'_k}} \leq \Xi^2 \ln(\Card(\mathcal{A})) 2^{(2\xi -k)\ell'_k} \ll 1, \\
	\end{displaymath}
	and therefore
	\begin{gather*}
	\frac{\epsilon_k}{f_{k-1}^B} \leq \frac{\epsilon_k}{(f_{k-1}^B)^2} \ \Rightarrow \ \epsilon_k \ll f_{k-1}^B \mbox{ and }  \\
	H(\epsilon_k) \leq 2 \epsilon_k \ln \Big( \frac{1}{\epsilon_k} \Big) \ll \sqrt{\epsilon_k} \ll f_{k-1}^B.
	\end{gather*}
	
	As $\ell_k f_k^B$ counts the number of $0$'s in the word $b_k$ and at each step of the construction the number is at least multiplied by 2, we have $\ell_{k-1}f_{k-1}^B\geq 2^{k-1}$,
	\begin{gather*}
	\frac{1}{\ell'_k} = \frac{1}{N'_k \ell_{k-1}} \ll f_{k-1}^B, \quad R_k' \geq \ell'_k, \quad \frac{1}{R(\hat X,\ell'_k)} \ll f_{k-1}^B.
	\end{gather*}
	Item 1 of Lemma~\ref{lemma:bounds} implies
	\begin{displaymath}
	\frac{1}{\ell'_k{}^2} \ln(C_k') \ll f_{k-1}^B.
	\end{displaymath}
	
	Item 1 of Definition\ref{notation} shows,
	\begin{displaymath}
	\frac{f_{k-1}^A}{N'_k} \leq  \frac{f_{k-1}^B}{k}, \quad \frac{f_{k-1}^A}{N'_k}  \ll f_{k-1}^B.
	\end{displaymath}
	We  proved that $P(\beta_k\phi)$ is bounded from below by a quantity equivalent to $f_{k}^B\ln(2)$ and bounded from above by a quantity equivalent to $\frac{3}{4} f_k^B \ln(2)$. We obtain a contradiction. We have proved that $\mu_{\beta_k}([1]) \leq  \frac{1}{4}$ for every even $k$ and every equilibrium measure $\mu_{\beta_k}$. Similarly $\mu_{\beta_k}([2]) \leq  \frac{1}{4}$ for every odd $k$ and every equilibrium measure $\mu_{\beta_k}$. As
	\begin{displaymath}
	\mu_{\beta_k}([0]) \leq  \frac{2}{N'_k}  f_{k-1}^A + (1-N_{k-1}^{-1})^{-1} f^B_{k-1} + (f_{k-1}^B)^2 \frac{R_k' \ln(\Card(\mathcal{A}))}{\exp(k \ell'_k)},
	\end{displaymath}
	we have proved
	\begin{gather*}
	\liminf_{k \to +\infty} \inf_\mu \big\{ \mu ([2]) : \mu \ \text{is an equilibrium measure at $\beta_{2k}$} \  \big\} \geq \frac{3}{4}, \\
	\liminf_{k \to +\infty} \inf_\mu \big\{ \mu ([1]) : \mu \ \text{is an equilibrium measure at $\beta_{2k+1}$} \  \big\} \geq \frac{3}{4},
	\end{gather*}
	and therefore $(\mu_{\beta_k})_{k\geq 0}$ does not converge.
\end{proof}

\appendix

\chapter{Computability results}
\label{appendix}

We thank Sebastian Barbieri for his help to compute the upper bounds for the relative complexity and for the reconstruction function. Sebastian stimulated us to prove that we can enumerate $\tilde{\mathcal{F}}$ in an increasing way and with a execution time that is at most exponential.

First we prove the upper bound for the relative complexity function given by Proposition~\ref{prop.Sebastian2}.

\begin{proof}[Proof of Proposition~\ref{prop.Sebastian2}]
	Let us denote by $C_n(\texttt{Layer}_k(\hat{X}))$ the complexity of the projection to the $k$-th layer. and by $C_n(\texttt{Layer}_k(\hat{X}) | \texttt{Layer}_j(\hat{X})  )$ the complexity of the projection to the $k$-th layer given that there is a fixed pattern on the $j$-th layer. Clearly we have that 
	\begin{multline*}
	C^{\hat{X}}(n) \leq C_n(\texttt{Layer}_1(\hat{X})) \cdot C_n(\texttt{Layer}_2(\hat{X})) \cdot C_n(\texttt{Layer}_3(\hat{X})| \texttt{Layer}_2(\hat{X}) ) \cdot \\ \cdot C_n(\texttt{Layer}_4(\hat{X})|\texttt{Layer}_2(\hat{X})).
	\end{multline*}
	
	\begin{itemize}
		\item \texttt{Layer} 1: As this layer is given by all $x\in \tilde{\mathcal{A}}^{\ZZ^2}$ so that $x_{u} = x_{u+(0,1)}$ for every $u \in \ZZ^2$, a trivial upper bound for the complexity is \[ C_n(\texttt{Layer}_1(\hat{X})) = \mathcal{O}(|\tilde{\mathcal{A}}|^n).   \] 
		
		In fact, as in the end the only configurations which are allowed are those whose horizontal projection lies in the effective subshift $\ZZ$, a better bound is given by $C_n(\texttt{Layer}_1(\hat{X})) = \mathcal{O}(\exp( n\ h_{\mbox{top}}(\hat{X})))$. For simplicity, we shall just keep the trivial bound.
		
		\item \texttt{Layer} 2: The complexity of every substitutive subshift in $\ZZ^2$ is $\mathcal{O}(n^2)$. To see this, suppose that the substitution sends symbols of some alphabet $\mathcal{A}_2$ to $n_1\times n_2$ arrays of symbols. By definition, every pattern of size $n$ occurs in a power of the substitution. If $k$ is such that $\min\{n_1,n_2\}^{k-1} \leq n \leq \min\{n_1,n_2\}^{k}$, then necessarily any pattern of size $n$ occurs in the concatenation of at most $4$ $k$-powers of the substitution. There are $|\mathcal{A}_2|^4$ choices for the $k$-powers and at most $(\max\{n_1,n_2\}^k)^2\leq (n\max\{n_1,n_2\})^2$ choices for the position of the pattern. It follows that there are at most $(|\mathcal{A}_2|^4\max\{n_1,n_2\}^2) n^2 = \mathcal{O}(n^2)$ patterns of size $n$. We obtain,
		\[ C_n(\texttt{Layer}_2(\hat{X})) = \mathcal{O}(n^2)\]
		
		\item \texttt{Layer} 3: It can be checked directly from the Aubrun-Sablik construction that the symbols on the third layer satisfy the following property: if the symbols on the substitution layer are fixed, then for every $u \in \ZZ^2$ the symbol at position $u$ is uniquely determined by the symbols at positions $u-(0,1), u-(1,1)$ and $u-(-1,1)$.  In consequence, it follows that knowing the symbols at positions in the ``U shaped region'' \[U = ( \{0\} \times \llbracket 1,n-1 \rrbracket ) \cup (\llbracket 0,n-1 \rrbracket\times \{0\}) \cup (\{n-1 \} \times \llbracket 1,n-1 \rrbracket) \]  completely determines the pattern. Therefore, if this layer has an alphabet $\mathcal{A}_3$, we have \[
		C_n(\texttt{Layer}_3(\hat{X})| \texttt{Layer}_2(\hat{X}) ) \leq |\mathcal{A}_3|^{3n-2} \leq \mathcal{O}(K_1^n),   \]
		for some positive integer $K_1$.
		
		\item \texttt{Layer} 4: $\mathcal{M}_{\texttt{Search}}$ The same argument for Layer $3$ holds for Layer 4. Therefore, if the alphabet of layer $4$ is $\mathcal{A}_4$ we have that for some positive integer $K_2$,
		\[
		C_n(\texttt{Layer}_4(\hat{X})| \texttt{Layer}_2(\hat{X}) ) \leq |\mathcal{A}_4|^{3n-2} \leq \mathcal{O}(K_2^n).   \]
	\end{itemize}
	
	Putting the previous bounds together, we conclude that there is some constant $K>0$ such that
	\[C^{\hat{X}}(n) = \mathcal{O}(n^2 K^n).\]	
\end{proof}

\begin{corollary}
Let $\hat{X}$ be the $\ZZ^2$-SFT in the Aubrun-Sablik construction. There is a constant $K_C >0$ such that
\[ \limsup_{n \to \infty} \frac{1}{n} \log(C_n(X)) \leq K_C.  \]
\end{corollary}

Now we will work on the upper bound for the reconstruction function. We fix a Turing machine $\mathcal{M}$ that enumerates $\tilde{\mathcal{F}}$ see below the set of forbidden words that define $\tilde{X}=\Sigma(\tilde{A},\tilde{\mathcal{F}})$. In general, the reconstruction function $R^{\hat{X}}$ as defined in (\ref{eq.Rkprime}) of the Aubrun-Sablik construction is not computable, but in our construction we can obtain the properties as stated in Proposition~\ref{proposition:algorithm} that we prove below.

\begin{proof}[Proof of Proposition~\ref{proposition:algorithm}]
If the integer $n\geq 1$ is such that $p=N_k$, then $\tilde{\mathcal{F}}'(n)=\tilde{\mathcal{F}}(n)$. We will consider now the case where the integer $n\geq 1$ is such that $p<N_k$. We have obviously $\tilde{\mathcal{F}}(n) \subseteq \tilde{\mathcal{F}}'(n)$. If we assume that $k$ is even, from Notation~\ref{notation:BaseLanguageBis} we have that
\[a_k = a_{k-1}(1^{\ell_{k-1}})^{N_k-2}a_{k-1}\quad \mbox{and}\quad b_k = (b_{k-1})^{N_k}.\]

We set
\[\overleftarrow{a_k}(1) = \overrightarrow{a_k}(1) = a_{k-1},\]
\[\overleftarrow{b_k}(1) = \overrightarrow{b_k}(1) = b_{k-1},\]
\[\overleftrightarrow{1_k}(1) = 1_{k-1}:=1^{\ell_{k-1}} \ \mbox{and} \ \overleftrightarrow{2_k}(1) = 2_{k-1}:=2^{\ell_{k-1}}.\]
Then we define by induction if $2 \leq p<N_k$ then
\[\overleftarrow{a_k}(p) = \overleftarrow{a_k}(p-1) 1_{k-1} = a_{k-1}(1_{k-1})^{p-1}\]
and
\[\overrightarrow{a_k}(p)= 1_{k-1}\overrightarrow{a_k}(p-1) = (1_{k-1})^{p-1} a_{k-1},\]
else \ $\overleftarrow{a_k}(N_k) = \overrightarrow{a_k}(N_k) = a_k$. We also define
\[\overleftarrow{b_k}(p)= \overleftarrow{b_k}(p-1) b_{k-1} = (b_{k-1})^{p},\]
\[\overrightarrow{b_k}(p) = b_{k-1}\overrightarrow{b_k}(p-1) = (b_{k-1})^{p},\]
\[\overleftrightarrow{1_k}(p)=  \overleftrightarrow{1_k}(p-1)1_{k-1}= (1_{k-1})^{p}\]
and
\[\overleftrightarrow{2_k}(p)=  \overleftrightarrow{2_k}(p-1)1_{k-1}= (2_{k-1})^{p}.\]

If $w$ has length less than $p\ell_{k-1}$ and is a sub-word of some $w_1w_2$, say $w_1=a_k$ and $w_2=b_k$, by dragging $w$ from the left end point of $w_1w_2$ to the right end point of $w_1w_2$, the word $w$ appears  successively as a sub-word of $\overleftarrow{a_k}(p+1)$, $\overleftrightarrow{1_k}(p+1)$, $\overrightarrow{a_k}(p+1)$,  $\overrightarrow{a_k}(p+1)\overleftarrow{b_k}(p+1)$, $\overleftarrow{b_k}(p+1)$.  A similar reasoning is also true for $w_1=b_k$ and $w_2=a_k$. We have proved $\tilde{\mathcal{F}}(n) =\tilde{\mathcal{F}}'(n)$.

We have also proved that $\tilde{X}=\Sigma^1(\tilde{\mathcal{A}},\tilde{F})$, because we have proved that it is enough to list all the forbidden words of length $n$ and for that it is sufficient to search in the concatenation of subwords of length $(p+1)\cdot\ell_{k-1}$ as described before. Thus \[\Sigma^1(\tilde{\mathcal{A}},\overline{\mathcal{F}})=\Sigma^1(\tilde{\mathcal{A}},\tilde{\mathcal{F}}).\]
	
To compute the time to enumerate successively the words of $\tilde{\mathcal{F}}(n) $ when $\ell_{k-1} < n  \leq \ell_k$, we produce an algorithm given in Table~\ref{table:algorithm}. The time to read/write on the tapes, to update the words $(\overleftarrow{a_k}(p),\overrightarrow{a_k}(p),\overleftarrow{b_k}(p),\overrightarrow{b_k}(p)$, $\overleftrightarrow{1_k}(p),\overleftrightarrow{2_k}(p))$ by adding a word of length $\ell_{k-1}$, to concatenate two words $\overrightarrow{w_1}\overleftarrow{w_2}$ from that list, and to check that a given word $w$ of length $n$ is a sub-word of $\overrightarrow{w_1}\overleftarrow{w_2}$ is polynomial in $n$. Therefore, the time to enumerate every word up to length $n$ in an alphabet $\tilde{\mathcal{A}}$ is bounded by $P(n)|\tilde{\mathcal{A}}|^n$ where $P(n)$ is a polynomial.
\end{proof}

\begin{table}[p]
	\begin{center}
	\bf A program enumerating the set of forbidden words
	\end{center}
	\begin{adjustbox}{width=\columnwidth,center}
	\begin{tabular}{l}
		\# Initialize $(\ell_{0},\beta_{0},\rho_{0}^A,\rho_{0}^B)$  \\ 
		$(\ell_{-},\beta_{-},\rho^A_{-},\rho^B_{-}) \leftarrow  (2,0,1,1)$   \\
		\# Allocate and Initialize 4 tapes $(a_k,b_k,1_k,2_k)$ \\
		$(a_-,b_-,1_-,2_-)  \leftarrow  (01,02,11,22)$  \\
		\# Allocate and Initialize 6 tapes $(\overleftarrow{a_k}(1),\overrightarrow{a_k}(1),\overleftarrow{b_k}(1),\overrightarrow{b_k}(1),\overleftrightarrow{1_k}(1),\overleftrightarrow{2_k}(1))$ \\
		$(\overleftarrow{a_+}, \overrightarrow{a_+}, \overleftarrow{b_+},\overrightarrow{b_+}, \overleftrightarrow{1_+}, \overleftrightarrow{2_+}) \leftarrow (a_-,a_-,b_-,b_-,1_-,2_-)$  \\
		\# Compute recursively the next length $\ell_1$ \\
		$(\ell_{+}, \beta_{+},\rho^A_{+},\rho^B_{+}) \leftarrow S(\ell_{-},\beta_{-},\rho^A_{-},\rho^B_{-} )$  \\
		$N_+ \leftarrow \ell_+/\ell_-$ ; $parity$ $\leftarrow$ $even$ ; $n \leftarrow 3$ ; $p \leftarrow 2$ \\
		\# Allocate and Initialize  an intermediate tape recording a possibly forbidden word \\
		$w \leftarrow \emptyset$ \\
		while ($n\geq1$) \\
		\quad if ($n =\ell_{+}+1)$ then \\
		\quad\quad \# Remember the previous $(\ell_{k-1},\beta_{k-1},\rho_{k-1}^A,\rho_{k-1}^B)$ and update the new one \\
		\quad\quad $(\ell_{-},\beta_{-},\rho^A_{-},\rho^B_{-}) \leftarrow (\ell_{+}, \beta_{+},\rho^A_{+},\rho^B_{+})$  ; $(\ell_{+}, \beta_{+},\rho^A_{+},\rho^B_{+}) \leftarrow S(\ell_{-},\beta_{-},\rho^A_{-},\rho^B_{-} )$  \\
		\quad\quad \# Remember $(a_{k-1},b_{k-1},1_{k-1}, 2_{k-1})$ \\
		\quad\quad $(a_-,b_-,1_-,2_-)  \leftarrow  (\overleftarrow{a_+},\overleftarrow{b_+},\overleftrightarrow{1_+},\overleftrightarrow{2_+})$ \\
		\quad\quad $N_+ \leftarrow \ell_+/\ell_-$ ; $parity$ $\leftarrow$ Permute($parity$) ; $p \leftarrow 2$ \\
		\quad end if \\
		\quad if ($n = (p-1) \ell_{-}+1$) then \\
		\quad\quad Update $(\overleftarrow{a_+}, \overrightarrow{a_+},\overleftrightarrow{1_+}, \overleftarrow{b_+},\overrightarrow{b_+},\overleftrightarrow{2_+})$ according to parity and the two particular \\
		\quad\quad\quad\quad cases  $p=2$ or $p=N_+$ by concatenating   words from  $(a_-,b_-,1_-,2_-) $ \\
		\quad\quad  \# Build the set of words obtained by concatenating two words of length $\ell_{k-1}$  \\
		\quad\quad $W \leftarrow \{ \overrightarrow{a_+}\overleftarrow{b_+}, \ \overrightarrow{a_+}\overleftrightarrow{2}, \ \overleftrightarrow{1}\overleftarrow{b_+}, \ \overleftrightarrow{1}\overleftrightarrow{2},  \overrightarrow{b_+}\overleftarrow{a_+}, \ \overrightarrow{b_+}\overleftrightarrow{1}, \ \overleftrightarrow{2}\overleftarrow{a_+}, \ \overleftrightarrow{2}\overleftrightarrow{1} \}$ \\
		\quad\quad $p \leftarrow p+1$ \\
		\quad end if \\
		\quad for ($m=0,3^n$ excluded) \\
		\quad\quad $w \leftarrow$ write $m$ in  base 3 with n letters in $\{0,1,2\}$ \\
		\quad\quad $is\_forbidden \leftarrow true$ \\
		\quad\quad for ($w_1w_2 \in W$) \\ 
		\quad\quad\quad if ($w$ is a sub-word of $w_1w_2$) then $is\_forbidden \leftarrow false$ \\
		\quad\quad end for \\
		\quad\quad if ($is\_forbidden$) then Print the word $w$ \\
		\quad end for \\
		\quad $n \leftarrow n+1$ \\
		end while \\
	\end{tabular}
\end{adjustbox}
\caption{Algorithm that enumerates $\tilde{\mathcal{F}}$.} \label{table:algorithm}
\end{table}

\vspace{0.5cm}

Denote by $R^{\tilde{X}} \colon \NN \to \NN$ the reconstruction functions of $\tilde{X}$ given $\tilde{\mathcal{F}}$. From Lemma~\ref{lemma.onedimensional} we know there exists a constant $C_1>0$ such that $R^{\tilde{X}}(n)\leq C_1 n$.

For $n \in \NN$, let $N = 2n+1$ be the length of the sides of the square $B_n:=\llbracket-n,n\rrbracket^2\subset \ZZ^2$, and let $k\in \NN$ such that $4^{k-1} < N \leq 4^{k}$. 

As before, let $\hat{X}=\Sigma(\hat{\mathcal{A}},\hat{\mathcal{F}})$ be the $\ZZ^2$-SFT in the Aubrun-Sablik construction associated to $\tilde{X}$ and the Turing machine $\mathcal{M}$. Now we will give estimates on the reconstruction function $R^{\hat{X}} \colon \NN\to \NN$ of $\hat{X}$ given $\hat{\mathcal{F}}$. Of course, a formal proof of these estimates would require a restatement of the construction of Aubrun-Sablik with all its details, which is out of the scope of this thesis. Instead, we shall argue that the bounds we give suffice, making reference to the properties of the Aubrun-Sablik construction.

A description of $\hat{\mathcal{F}}$ is given in~\cite{AS} in an (almost) explicit manner for all layers except the substitution layer. For the substitution layer, a description of the forbidden patterns can be extracted in an explicit way from the article of Mozes~\cite{Mozes}.

The behavior of layers 2,3 and 4 in the Aubrun-Sablik construction is mostly independent of layer 1, except for the detection of forbidden patterns which leads to the forbidden halting state of the machine in the third layer. Because of that reason the analysis of the reconstruction function $R^{\hat{X}}$ can be split into two parts:
\begin{enumerate}
	\item \textbf{Structural: } Assuming that the contents of the first layer are globally admissible (the configuration in the first layer is an extension of a configuration from $\tilde{X}$), we give a bound that ensures that the contents of layers $2,3$ and $4$ are globally admissible, that is:
	\begin{itemize}
		\item The contents of layer $2$ correspond to a globally admissible pattern in the substitutive subshift and the clock.
		\item The contents of layer $3$ and $4$ correspond to valid  space-time diagrams of Turing machines that correctly align with the clocks.
	\end{itemize}
	\item \textbf{Recursive: } A bound that ensures that the contents of the first layer are globally admissible. This bound will of course depend upon $R^{\tilde{X}}$ and $\tau$.
\end{enumerate}

Finally, we are able to prove the upper bound for the reconstruction function given by Proposition~\ref{prop.Sebastian3}.

\begin{proof}[Proof of Proposition~\ref{prop.Sebastian3}]
Let us begin with the structural part, as it is simpler and does not depend upon $\tilde{X}$. Let $p$ be a pattern with support $B_n$ and assume that the first layer of $p$ is thus globally admissible.
	
From Mozes's construction of SFT extensions for substitutions~\cite{Mozes} it can be checked that any locally admissible pattern of support $B_n$ of Mozes's SFT extension of a primitive substitution (The Aubrun-Sablik substitution is primitive) is automatically globally admissible. Let us take a support large enough such that the second layer of $p$ occurs within four $4^k \times 2^k$ macrotiles of the substitution in any locally admissible pattern of that support.
	
Next, a clock runs on every strip of the Aubrun-Sablik construction. By the previous argument, the largest zone which intersects $p$ in more than one position is of level at most $k$. Therefore its largest computation strip has horizontal length $2^k$. In order to ensure that the clock starts on a correct configuration on every strip contained in $p$, we need to witness this pattern inside a locally admissible pattern which stacks $2^{2^k}+2$ macrotiles of level $k$ vertically. Therefore, the pattern $p$ must occur inside four locally admissible patterns of length $4^k \times 2^k(2^{2^k}+2)$. This ensures that the clocks in $p$ are globally admissible.
	
Finally, if every clock occurring in $p$ starts somewhere, then the contents of the third layer are automatically correct, as they are determined by clock every time it restarts. To check that the fourth layer is correct, we just need extend the horizontal length of our pattern twice, so that the responsibility zone of the largest strip is contained in it. 
	
By the previous arguments, it would suffice to witness $p$ inside a locally admissible pattern which contains in its center a $4 \times 2$ array of macrotiles of size $4^k \times 2^k(2^{2^k}+2)$. As $4^{k-1} < N \leq 4^{k}$, there is a constant $C_0 >0$ such that an estimate for this part of the reconstruction function can be written as \[ R_{\mathrm{Struct}}^{\hat{X}}(n) = \mathcal{O}(\sqrt{n}C_0^{\sqrt{n}}).  \]

Let us now deal with the recursive part. We need to find a bound such that the word of length $N$ occurring in the first layer of $p$ is globally admissible. By definition of $R^{\tilde{X}}$, it suffices to have $p$ inside a pattern with support $B_{R^{\tilde{X}}(N)}$ and check that the first layer is locally admissible with respect to $\tilde{\mathcal{F}}$. In other words, we need to have the Turing machines check all forbidden words of length $R^{\tilde{X}}(N)$ in this pattern. Luckily, the number of steps in order to do this is already computed in Aubrun and Sablik's article. After Fact 4.3 of~\cite{AS} they show that, if $p_0,p_1,\dots,p_r$ are the first $r+1$ patterns enumerated by $\mathcal{M}$, then the number of steps $S(p_0,\dots,p_r)$ needed in a computation zone to completely check whether a pattern from $\{p_0,\dots,p_r\}$ occurs in its responsibility zone of level $m$ satisfies the bound,\[ S(p_0,\dots,p_r) \leq T(p_0,\dots,p_r) + (r+1)\max(|p_0|,\dots,|p_r|)m^22^{3m+1},   \]
where $T(p_0,\dots,p_r)$ is the number of steps needed by $\mathcal{M}$ to enumerate the patterns $p_0,p_1,\dots,p_r$.
	
Specifically in our construction, we may rewrite their formula so that the number $S(R^{\tilde{X}}(N))$ of steps needed to check that all forbidden patterns of length at most $R^{\tilde{X}}(N)$ in a responsibility zone of level $m$ satisfies the bound \begin{align*}
S(R^{\tilde{X}}(N)) & \leq \tau(R^{\tilde{X}}(N)) + |\tilde{\mathcal{A}}|^{R^{\tilde{X}}(N)+1}R^{\tilde{X}}(N)k^22^{3k+1}\\
&  \leq P(n)|\tilde{\mathcal{A}}|^{N} + |\tilde{\mathcal{A}}|^{C_1N + 1}C_1 N m^22^{3m+1}
\end{align*}
Simplifying the above bound, it follows that there exists constants $C_2,C_3>0$ such that
\[ S(R^{\tilde{X}}(N)) \leq C_2 m^22^{3m+C_3N}.    \]
	
As $N$ is constant, it follows that there is a smallest $\bar{m}=\bar{m}(N) \in \NN$ such that $2^{\bar{m}} \geq C_4N$ (so that the tape on the computation zone of level $\bar{m}$ can hold words of size $R_Z(N)$) and such that
\[C_2 \bar{m}^22^{3\bar{m}+C_3N} \leq 2^{2^{\bar{m}}}+2,\]
so that the number $2^{2^{\bar{m}}}+2$ of computation steps in the zone of level $\bar{m}$ is enough to check all the words of size $R^{\tilde{X}}(N)$. It follows that a bound for the recursive part of $R^{\hat{X}}$ is given by
\[ R_{\mathrm{recursive}}^{\hat{X}}(n) = \mathcal{O}(2^{\bar{m}+2^{\bar{m}(N)}}).   \]
	
In order to turn this into an explicit asymptotic expression we need to find a suitable bound for $\bar{m}(N)$. Notice that if $m \geq 6$ we simultaneously have that $m^2 \leq 2^m$ and $4m \leq 2^{m-1}$. We may then write for $m \geq 6$,
\[ C_2 m^22^{3m}e^{C_3N} \leq C_2 2^{4m+C_3N} \leq C_2 2^{C_3N}2^{2^{m-1}}.   \]
Therefore, it suffices to find $\bar{m} =\bar{m}(N)$ such that
\[ C_22^{C_3N} \leq 2^{2^{\bar{m}-1}}.   \]
From here, it follows that there is a constant $C_5 >0$ such that any value of $\bar{m}$ satisfying
\[ \bar{m} \geq C_5 + \log_2(N),    \]
satisfies the above bound. We get that
\[ R_{\mathrm{recursive}}^{\hat{X}}(n) = \mathcal{O}(N 2^{C_5N}) = n4^{C_5n}.   \]
	
Finally, putting together the structural and recursive asymptotics, we obtain that there is a constant $K >0$ such that
\[  R^{\hat{X}}(n) = \mathcal{O}(\max \{ \sqrt{n}C_0^{\sqrt{n}} , \mathcal{O}(n4^{C_5n})\}) = \mathcal{O}(nK^n).    \]
\end{proof}

\begin{corollary}
Under the same hypotheses as in Proposition~\ref{prop.Sebastian3}, there is a constant $K >0$ such that \[ \limsup_{n \to \infty}\frac{1}{n}\log(R_{\hat{X}}(n)) \leq K.  \]
\end{corollary}

\bibliographystyle{acm}
\bibliography{bib}

\end{document}